\documentclass[12pt]{article}
\usepackage[full]{textcomp}  

\usepackage{amsmath}
\usepackage{amsthm}
\usepackage{multicol}
\usepackage{cancel}
\usepackage{comment}
\usepackage{hyperref}
\usepackage{xfrac}
\usepackage{mathrsfs}
\hypersetup{
	colorlinks=true,
	linkcolor=blue,
	citecolor=red,
	breaklinks=true
}

\usepackage[T1]{fontenc}
\usepackage[left=1.3in, right=1.3in, top=1.25in, bottom=1.25in]{geometry}
\usepackage{dirtytalk}
\usepackage[small]{titlesec}
\usepackage{mathtools}
\usepackage{comment}
\usepackage{epigraph}

\usepackage{fancyhdr}
\usepackage{natbib}
\setcitestyle{authoryear,open={(},close={)}}

\usepackage{sgame}

\ssualfalse
\gamemathtrue

\usepackage{pgfplots}
\usepackage{tikz}
\usetikzlibrary{positioning}
\pgfplotsset{width = 10cm,compat = 1.9}
\usetikzlibrary{patterns}
\usepackage{array}

\def\t{\textrm}

\newtheorem{theorem}{Theorem}
\newtheorem{lemma}{Lemma}

\newtheorem*{axiom*}{Axiom}
\newtheorem{proposition}{Proposition}

\newtheorem{remark}{Remark}

\newtheorem*{theorem*}{Theorem}
\theoremstyle{definition} 
\newtheorem{definition}{Definition}
\newtheorem{example}{Example}

\newcommand{\spbr}{:}

\usepackage{pifont}
\usepackage{fontawesome5}
\newcommand{\cmark}{\ding{51}}%
\newcommand{\xmark}{\ding{55}}%
\usepackage{hyperref}
\usepackage[nameinlink,noabbrev,sort&compress]{cleveref}
\crefname{assumption}{assumption}{assumptions}
\usepackage{algorithm}
\usepackage{algpseudocode}

\usepackage{graphicx}

\usepackage{kpfonts} 

\usepackage{inconsolata}

\newcolumntype{C}{>{$}c<{$}} 

\makeatletter
\def\squarecorner#1{
	\pgf@x=\the\wd\pgfnodeparttextbox%
	\pgfmathsetlength\pgf@xc{\pgfkeysvalueof{/pgf/inner xsep}}%
	\advance\pgf@x by 2\pgf@xc%
	\pgfmathsetlength\pgf@xb{\pgfkeysvalueof{/pgf/minimum width}}%
	\ifdim\pgf@x<\pgf@xb%
	\pgf@x=\pgf@xb%
	\fi%
	\pgf@y=\ht\pgfnodeparttextbox%
	\advance\pgf@y by\dp\pgfnodeparttextbox%
	\pgfmathsetlength\pgf@yc{\pgfkeysvalueof{/pgf/inner ysep}}%
	\advance\pgf@y by 2\pgf@yc%
	\pgfmathsetlength\pgf@yb{\pgfkeysvalueof{/pgf/minimum height}}%
	\ifdim\pgf@y<\pgf@yb%
	\pgf@y=\pgf@yb%
	\fi%
	\ifdim\pgf@x<\pgf@y%
	\pgf@x=\pgf@y%
	\else
	\pgf@y=\pgf@x%
	\fi
	\pgf@x=#1.5\pgf@x%
	\advance\pgf@x by.5\wd\pgfnodeparttextbox%
	\pgfmathsetlength\pgf@xa{\pgfkeysvalueof{/pgf/outer xsep}}%
	\advance\pgf@x by#1\pgf@xa%
	\pgf@y=#1.5\pgf@y%
	\advance\pgf@y by-.5\dp\pgfnodeparttextbox%
	\advance\pgf@y by.5\ht\pgfnodeparttextbox%
	\pgfmathsetlength\pgf@ya{\pgfkeysvalueof{/pgf/outer ysep}}%
	\advance\pgf@y by#1\pgf@ya%
}
\makeatother

\pgfdeclareshape{square}{
	\savedanchor\northeast{\squarecorner{}}
	\savedanchor\southwest{\squarecorner{-}}
	
	\foreach \x in {east,west} \foreach \y in {north,mid,base,south} {
		\inheritanchor[from=rectangle]{\y\space\x}
	}
	\foreach \x in {east,west,north,mid,base,south,center,text} {
		\inheritanchor[from=rectangle]{\x}
	}
	\inheritanchorborder[from=rectangle]
	\inheritbackgroundpath[from=rectangle]
}

\tikzset{
	pref/.style={
		align = center,
	},
}

\makeatletter
\newcommand{\threedots}{%
	\mathrel{\mathpalette\threedots@\relax}%
}
\newcommand{\threedots@}[2]{%
	\sbox\z@{$\m@th#1:$}%
	\vbox to\ht\z@{%
		\hbox{$\m@th#1.$}%
		\vss
		\hbox{$\m@th#1.$}%
		\vss
		\hbox{$\m@th#1.$}%
	}%
}
\makeatother

\DeclareFontFamily{U}{mathb}{\hyphenchar\font45}
\DeclareFontShape{U}{mathb}{m}{n}{%
	<-6> mathb5
	<6-7> mathb6
	<7-8> mathb7
	<8-9> mathb8
	<9-10> mathb9
	<10-12> mathb10
	<12-> mathb12 }{}
\DeclareSymbolFont{mathb}{U}{mathb}{m}{n}
\DeclareMathSymbol{\square}   {2}{mathb}{"9C}

\newcommand*\circled[1]{#1^\circ}
\newcommand*\circledinv[1]{#1^{\textcolor{white}{\circ}}}

\newcommand*\squared[1]{#1^\squa}
\newcommand*\squaredinv[1]{#1^{\textcolor{white}{\squa}}}

\usepackage{subcaption}

\newcommand{\wor}{w}
\newcommand{\Wor}{W}

\newcommand{\fir}{f}
\newcommand{\Fir}{F}

\newcommand{\spref}{\succ}
\newcommand{\wpref}{\succsim}

\newcommand{\unmat}{\varnothing}

\newcommand{\age}{a}
\newcommand{\Age}{A}
\newcommand{\bge}{b}

\newcommand{\rema}{\mat_0}

\newcommand{\Coal}{C}

\newcommand{\mat}{\mu}

\newcommand{\nat}{\nu}

\newcommand{\mech}{\psi}

\newcommand{\squa}{\square}
\newcommand{\matcir}{\mat^\circ}
\newcommand{\matsqu}{\mat^\squa}

\newcommand{\mattil}{\tilde{\mat}}

\newcommand{\matProp}{\mat_1}
\newcommand{\matExch}{\mat_2}
\newcommand{\matTTC}{\mat^{\t{TTC}}}
\newcommand{\matDA}{\mat^{\t{DA}}}

\newcommand{\Propose}{Propose }
\newcommand{\Exchange}{Exchange }

\newcommand{\Graph}{G}
\newcommand{\Edges}{E}
\newcommand{\Vertices}{V}
\newcommand{\Path}{P}
\newcommand{\Qath}{Q}

\newcommand{\edg}{e}
\newcommand{\Imp}{I}

\newcommand{\Lat}{L}
\newcommand{\lgeq}{ \geq}

\newcommand{\join}{\vee}
\newcommand{\meet}{\wedge}

\newcommand{\AgeFree}{\Age^\Free}
\newcommand{\Free}{F}

\newcommand{\quo}{q}
\newcommand{\wdom}{\geq}

\newcommand{\sema}{\rema^\sigma}

\title{Matching With Pre-Existing Binding Agreements: The Agreeable Core}

\author{Peter Doe\footnote{California Institute of Technology; email: \href{mailto:pdoe@caltech.edu}{pdoe@caltech.edu}. This paper has benefited from the support and guidance of my advisors Federico Echenique and Luciano Pomatto, as well as conversations with Peter Caradonna, Axel Niemeyer, Ke Shi, and the participants of the Caltech theory group and EC'24.}}
\date{\today}

\begin{document}
	\maketitle
	
	\begin{abstract}
		Matching market models ignore prior commitments.
		Yet many job seekers, for example, are already employed, and the same holds for many other matching markets.
		I analyze two-sided matching markets with pre-existing binding agreements between market participants.
		In this model, a pair of participants bound to each other by a pre-existing agreement must agree to any action they take.
		To analyze their behavior, I propose a new solution concept, the \textit{agreeable core}, consisting of the matches which cannot be renegotiated without violating the binding agreements.
		My main contribution is an algorithm that constructs such a match by a novel combination of the Deferred Acceptance and Top Trading Cycles algorithms.
		The algorithm is robust to various manipulations and has applications to numerous markets including the resident-to-hospital match, college admissions, school choice, and labor markets.
	\end{abstract}
	\section{Introduction}
	
	In matching markets, pre-existing agreements are common.
	For example, when a student is admitted to a college through an Early Decision program, she commits to attend the college; she is bound to the college, and it now controls her right to participate in the regular admission cycle.  
	When a professional athlete signs on to a sports team, that team purchases her right to sign on to other teams.
	Both examples include market participants---whether students or athletes---who have bound themselves to others.
	They are denied the right to find a new partner unless they are released from their agreements.
	
	The standard model of matching markets ignores these interdependencies.
	It gives participants unrestricted rights to form new agreements, regardless of their earlier agreements.
	That is clearly unrealistic.
	
	I propose a new model to solve this problem.
	It makes it possible to analyze such markets.
	At bottom, my model requires that any action taken by one person must receive the approval of the person to whom she is bound.
	For example, a professional athlete can only seek another position with the approval of her team.
	Without the approval, she faces penalties for breaching her agreement.
	To manage these constraints, I introduce the concept of an \textit{agreeable} group of participants.
	A group is agreeable if no member of the group is bound to someone outside of the group by a pre-existing agreement.
	In my example, an agreeable group only contains the athlete if it also contains her team, and vice versa.
	Critically, neither the athlete nor the team needs to be released from an agreement by anyone outside of the group.
	
	My solution, the \textit{agreeable core}, consists of the outcomes that cannot be renegotiated by any agreeable group.
	For a candidate outcome $\mat$, the agreeable core considers every agreeable group and checks whether the group can achieve a better outcome for its members.
	If no such agreeable group exists, then the agreeable core includes $\mat$.
	In the professional sports example, an agreeable group may contain some athletes and their respective teams.
	The agreeable core allows those athletes and teams to renegotiate their new contracts \textit{before they are signed} so long as every member of the group benefits.
	No person outside of this group can impede the negotiation because the group is agreeable.
	
	In the agreeable core participants are granted maximum flexibility in dissolving agreements and forming new matches.
	Critically, when a group considers renegotiating a candidate outcome, it does so before the pre-existing agreements are dissolved; after the pre-existing agreements are dissolved, the candidate outcome becomes the set of binding agreements.
	In effect, each participant has the flexibility to secure a replacement before releasing her partner from their binding agreement.
	For example, a team may condition the release of a player from their pre-existing contract on whether it secures a more-preferred replacement to sign a new contract.
	Without this flexibility, market congestion---where participants delay decision-making due to the lack of guarantees on replacement agreements---can result when participants with binding agreements hesitate to dissolve them without certainty of better outcomes.
	For instance, if a player is released from her pre-existing contract, the particular team she signs on to next directly affects which other players are available to her original team.
	The agreeable core resolves this by allowing the formation of a new agreement immediately after the old agreement is dissolved.
	
	Notably, I show that there are only two ways that agreements are dissolved in the agreeable core (\Cref{pro:nofreeblockingpairs}).
	First, some agreements are dissolved unconditionally by both parties.
	Both parties are able to find better alternative partners regardless of the action the other takes.
	The outcome would be the same with or without the agreement.
	For example, this occurs when a team and an athlete jointly agree to cancel their contract; whether or not the contract initially existed is irrelevant to their future decisions.
	Second, the remaining agreements are only dissolved through \say{trades.} 
	In a trade, two or more participants exchange the partners to whom they bound.
	For example, this occurs when two teams trade players.
	I show that at every outcome in the agreeable core, every dissolved agreement is of one of these forms.
	
	The main contribution of this paper is a two-stage algorithm, the \textit{Propose-Exchange} algorithm (PE), which always produces an outcome in the agreeable core.
	The novel feature of the PE is how it leverages \Cref{pro:nofreeblockingpairs} to partition the participants according to how they dissolve their agreements.
	The PE uses a cascading process to determine which agreements can be dissolved unconditionally.
	Among participants who unconditionally dissolve their agreements, the PE then uses Deferred Acceptance algorithm (DA) from two-sided matching theory to assign a match in the core \citep{gale_college_1962}.
	For the participants who are bound by some agreements, the PE allows participants to trade their partners as in the Top Trading Cycles algorithm from the object allocation literature \citep{shapley_cores_1974}.
	
%
	The PE algorithm can replace existing algorithms in markets that suffer from a lack of participation.
	Two prominent applications of market design---\say{The Match} conducted by the National Resident Matching Program (NRMP) and open-enrollment programs---are only incomplete markets.
	In the NRMP, some residents are offered posts outside of The Match.
	Prospective residents are forced to decide between accepting an early offer and participating in The Match.
	In open-enrollment programs, students can simultaneously hold offers from both the school district and private schools, leading to market congestion.
	Both problems arise because some agents accept offers through a decentralized system.
	The PE algorithm resolves this problem by integrating the centralized market with the decentralized market.
	Both the NRMP and open-enrollment programs use a version of the DA or TTC, so the PE can implement either.
	Incorporating the decentralized market is also straightforward: simply take the outcome of the decentralized market as the set of binding agreements.
	Because the PE (and the agreeable core in general) leaves no agent worse off than they are with their pre-existing binding agreements, the PE encourages agents to participate who normally would not.
	Participating in the PE is a weakly dominant strategy for agents who have created binding agreements in the decentralized market.
	
	Second, the agreeable core provides an explainable solution in matching with minimum constraints.
	In some applications agents have minimum quotas that the designer must respect.
	In the context of matching residents to hospitals in Japan, the Japanese government seeks to guarantee that some regions receive a minimum number of residents \citep{kojima_designing_2018}.
	In public-school open-enrollment, the designer may have a preference for maintaining socioeconomic diversity at the schools; these are frequently written as minimum constraints assigned to different socioeconomic tiers (see\cite{fragiadakis_improving_2017} for discussion of these examples).
	In the United States Military Academy, cadets are assigned to positions subject to minimum manning constraints \citep{fragiadakis_improving_2017}.
	To accommodate minimum constraints in the agreeable core, the designer only needs to create artificial binding agreements.
	For example, if the designer adds an agreement between a hospital and a resident, then the hospital is guaranteed to match to (at least one) resident.
	The agreeable core provides a robust justification for the outcome: no other outcome could be reached without violating either agents' preferences or the minimum constraints.
	
	The results of this paper are grounded in the formalization of binding agreements as an \textit{initial match} denoted $\rema$.
	In this formalization, each participant is initially matched to at most one other participant.
	For concreteness, I label one side \textit{workers} and the other side \textit{firms}, and I refer to groups of agents as \textit{coalitions}.
	The initial match rules out any participant being \say{double-booked;} otherwise, one participant may be bound to two others, creating ambiguity as to which agreement has precedence.
	Similarly, the initial match only allows for binding agreements to be two-way.
	For example, this formulation requires that if a student is bound to a college, then that college is bound to this student.
	There are ways to allow for some types of one-way agreements, but these require modifying participants' preferences.
	In this formalization, agreeable coalitions of agents are those which only include one participant if and only if her initial match is also a member of the coalition.
	
	The agreeable core is an entirely different approach compared to previous research on matching with an initial match.
	Previous research emphasizes the properties of specific algorithms, such as strategy-proofness or efficiency \citep{combe_reallocation_2024, combe_design_2022, guillen_matching_2012, hafalir_efficient_2023, hamada_strategy-proof_2017}.
	In contrast, this paper first develops a solution concept and then constructs an algorithm.
	The advantage is that the outcome in the agreeable core can be justified \textit{without} relying upon the properties of the particular algorithm used to select it.
	Arguably, outcomes are easier to explain than algorithms.\footnote{For example, the statement \textit{your child is at highest ranked school you listed where she is above the school's cutoff} is easier for parents to understand compared to \textit{we used the only algorithm that satisfies non-wastefulness, population monotonicity, weak Maskin monotonicity, and mutual best}; see \cite{morrill_alternative_2013}.}
	The trade-off is that the PE algorithm does not have the same incentive properties that are often baked into existing algorithms; however, I show that the PE satisfies a weakened version of strategy-proofness.
	
	The Propose-Exchange algorithm is novel in its combination of both the Deferred Acceptance and Top Trading Cycles algorithms and has no similar predecessors.
	To the best of my knowledge, the only other algorithm capable of implementing both the DA and the TTC is the Stable Improvement Cycles algorithm of \cite{abdulkadirog_generalized_2011}, which operates in a very different fashion.
	My use of the DA to divide the matching problem into two is entirely new and has promising applications in other markets with an initial match.

	The rest of the paper proceeds as follows.
	In \Cref{sec:example} I motivate the agreeable core through an illustrative example.
	\Cref{sec:model} presents the model.
	In \Cref{sec:PEalg} I present the proof of my main result, the Propose-Exchange algorithm that always produces a match in the agreeable core.
	\Cref{sec:manipulability} contains several results related to the manipulability of the Propose-Exchange algorithm.
	I defer a discussion of the related literature until \Cref{sec:discussion}, where I discuss how the agreeable core presents an alternative understanding of several economic applications.

	\section{A Motivating Example}\label{sec:example}
	In this section I introduce an example to illustrate my main definitions.
	This example highlights the limitations of the standard solution concept---the core---in matching markets where an initial match exists (the pre-existing binding agreements).
	By way of reminder, a match is in the \textit{core} if no group of agents, known as a \textit{blocking coalition}, can strictly improve their outcomes by forming an alternative match solely among themselves.
	The core does not account for the binding agreements and fails to improve upon the initial match.
	
	\begin{example}\label{exa1}
		There are four workers ($\wor_1$, $\wor_2$, $\wor_3$, and $\wor_4$) and four firms ($\fir_A$, $\fir_B$, $\fir_C$, and $\fir_D$).
		All workers prefer $\fir_A$ to $\fir_B$ to $\fir_C$ to $\fir_D$, except worker $\wor_1$ who swaps the order of $\fir_A$ and $\fir_B$.
		All firms prefer $\wor_3$ to $\wor_1$ to $\wor_2$ to $\wor_4$, except for firm $\fir_A$ who swaps the order of $\wor_1$ and $\wor_2$.
		Worker $\wor_1$ and firm $\fir_A$ have a contract, as do worker $\wor_2$ and firm $\fir_B$, and also $\wor_4$ and firm $\fir_D$.
		Worker $\wor_3$ and firm $\fir_C$ do not have a contract.
		In the language of my model, these contracts are the initial match $\rema$ to which any agent can appeal (the set of pre-existing binding agreements which cannot be dissolved without the agreement both parties).
		Any outcome must guarantee that all agents are weakly better off than under the initial match.
		The initial match is essential because it limits the participants' flexibility in forming new contracts.
		The preferences are summarized in \Cref{fig:example}, with the initial match circled.
		\begin{figure}[t!]
			\begin{center}
				\begin{tikzpicture}
					
					\def\lshift{-3.92};
					\def\rshift{1};
					\def\colwid{0.933};
					\def\rowhei{1.14};
					\def\frow{-3.03};

					\def\mlw{0.5mm};
					
					\node (aligner) [] {};
					
					\draw[line width=\mlw] (\lshift,\frow+\rowhei+\rowhei+\rowhei) ellipse (0.3 and 0.3);
					\draw[line width=\mlw] (\lshift+\colwid,\frow+\rowhei+\rowhei+\rowhei) ellipse (0.3 and 0.3);
					\draw[line width=\mlw] (\lshift+\colwid+\colwid,\frow) ellipse (0.3 and 0.3);
					\draw[line width=\mlw] (\lshift+\colwid+\colwid+\colwid,\frow+\rowhei) ellipse (0.3 and 0.3);
					
					\draw[line width=\mlw] (\rshift,\frow+\rowhei+\rowhei) ellipse (0.3 and 0.3);
					\draw[line width=\mlw] (\rshift+\colwid,\frow+\rowhei+\rowhei) ellipse (0.3 and 0.3);
					\draw[line width=\mlw] (\rshift+\colwid+\colwid,\frow) ellipse (0.3 and 0.3);
					\draw[line width=\mlw] (\rshift+\colwid+\colwid+\colwid,\frow+\rowhei) ellipse (0.3 and 0.3);

					\node[pref] (players) [left = 2.5mm of aligner] {
						\begin{tabular}{C | C | C | C}
							\wor_1 & \wor_2 & \wor_3 & \wor_4   \\
							\hline  \hline 
							\fir_B & \fir_A & \fir_A & \fir_A \\
							\fir_A & \fir_B & \fir_B & \fir_B \\
							\fir_C & \fir_C & \fir_C & \fir_C \\
							\fir_D & \fir_D & \fir_D & \fir_D \\
							\unmat & \unmat & \unmat & \unmat \\
						\end{tabular}
					};	
					
					\node[pref] (coaches) [right = 2.5mm of aligner]{
						\begin{tabular}{C | C | C | C}
							\fir_A & \fir_B & \fir_C & \fir_D   \\
							\hline  \hline 
							\wor_3 & \wor_3 & \wor_3 & \wor_3 \\
							\wor_2 & \wor_1 & \wor_1 & \wor_1 \\
							\wor_1 & \wor_2 & \wor_2 & \wor_2 \\
							\wor_4 & \wor_4 & \wor_4 & \wor_4 \\
							\unmat & \unmat & \unmat & \unmat \\
						\end{tabular}
					};	
					
					\node[pref] (explanation) [below = 4 cm of aligner] {
						\begin{tikzpicture}
							\draw[line width = \mlw] (0,-0.13) ellipse (0.3 and 0.3);
							\node (alignerTwo) [] {};
							\node  [right = 0.25cm of alignerTwo, align=center] {= initial match $\rema$} ;
						\end{tikzpicture}
					};

				\end{tikzpicture}
				\caption{Preferences in \Cref{exa1}, listed from most to least preferred, with $\unmat$ indicating a preference for remaining unmatched; for example, this first column reads $\wor_1$ strictly prefers $\fir_B$ to $\fir_A$ to $\fir_C$ to $\fir_D$ to being unmatched.
				The circles indicate the initial match $\rema$; for example, $\wor_1$ is initially matched (that is, under contract) to $\fir_A$.
				}
				\label{fig:example}
			\end{center}
		\end{figure}

		Consider the core of this market.
		At any core outcome, worker $\wor_3$ must be matched to firm $\fir_A$ because they mutually rank each other as best; otherwise, the coalition of $\{\wor_3, \fir_A\}$ blocks the match.
		However, this implies that either $\wor_1$ or $\wor_2$ is \textit{not} matched to $\fir_A$ or $\fir_B$ and thus is worse-off than under $\rema$.
		This a violation of the initial match $\rema$.
		Therefore there is no match in the core that improves upon the initial match.
	\end{example}
	
	The failure of the core to provide a match that improves upon the initial match arises from the blocking coalitions allowed.
	Allowing every subset of agents to block is too permissive and ignores the initial match $\rema$.
	The core is usually justified by arguing that agents in a blocking coalition could form contracts among only themselves, which allows for coalitions such as $\{\wor_3, \fir_A\}$.
	
	Although the core is unsatisfactory, there are two Pareto improvements of the initial match, indicated in \Cref{fig:examplePE}.
	In both, $\wor_1$ is matched to $\fir_B$ and $\wor_2$ is matched to $\fir_A$.
	The first Pareto improvement, labeled $\bar\mat$, matches $\wor_3$ to $\fir_C$ and $\wor_4$ to $\fir_D$.
	Every blocking coalition contains $\{\wor_3, \fir_A\}$ or $\{\wor_3, \fir_B\}$ because no firm wants $\wor_4$ more than its partner in $\bar \mat$, and both $\wor_1$ and $\wor_2$ are matched to their most-preferred partners.
	Consider $\{\wor_3, \fir_A\}$ first.
	Both $\wor_3$ and $\fir_A$ prefer each other to the proposed match $\bar\mat$.
	But would worker $\wor_1$ release $\fir_A$ from her contract to go and match to $\wor_3$?
	Worker $\wor_1$'s release of $\fir_A$ is contingent upon $\wor_1$ signing a contract with $\fir_B$, but $\fir_B$ has the same constraint: $\wor_2$ must be induced to release $\fir_B$, which cannot be done without guaranteeing that $\wor_2$ matches to $\fir_A$.
	But the premise of this blocking coalition is that $\fir_A$ will match to $\wor_3$ instead of $\wor_2$, so $\wor_2$ would not consent to this plan.
	In the language of my model, the coalition $\{\wor_3, \fir_A\}$ is not agreeable and thus cannot renegotiate its contracts; a similar argument follows for the coalition $\{\wor_3, \fir_B\}$.
	
	\begin{figure}
		\begin{center}
			\vspace{0.5cm}
			\begin{subfigure}{0.9\textwidth}
				\centering
				\begin{tikzpicture}
					
					\def\lshift{-3.92};
					\def\rshift{1};
					\def\colwid{0.933};
					\def\rowhei{1.14};
					\def\frow{-3.03};

					\def\mlw{0.5mm};
					
					\node (aligner) [] {};
					
					\draw[line width=\mlw, dashed, dash pattern=on 8pt off 3pt] (\lshift,\frow+\rowhei+\rowhei+\rowhei+\rowhei) ellipse (0.3 and 0.3);
					\draw[line width=\mlw, dashed, dash pattern=on 8pt off 3pt] (\lshift+\colwid,\frow+\rowhei+\rowhei+\rowhei+\rowhei) ellipse (0.3 and 0.3);
					\draw[line width=\mlw, dashed, dash pattern=on 8pt off 3pt] (\lshift+\colwid+\colwid,\frow+\rowhei+\rowhei) ellipse (0.3 and 0.3);
					\draw[line width=\mlw, dashed, dash pattern=on 8pt off 3pt] (\lshift+\colwid+\colwid+\colwid,\frow+\rowhei) ellipse (0.3 and 0.3);
					
					\draw[line width=\mlw, dashed, dash pattern=on 8pt off 3pt] (\rshift,\frow+\rowhei+\rowhei+\rowhei) ellipse (0.3 and 0.3);
					\draw[line width=\mlw, dashed, dash pattern=on 8pt off 3pt] (\rshift+\colwid,\frow+\rowhei+\rowhei+\rowhei) ellipse (0.3 and 0.3);
					\draw[line width=\mlw, dashed, dash pattern=on 8pt off 3pt] (\rshift+\colwid+\colwid,\frow+\rowhei+\rowhei+\rowhei+\rowhei) ellipse (0.3 and 0.3);
					\draw[line width=\mlw, dashed, dash pattern=on 8pt off 3pt] (\rshift+\colwid+\colwid+\colwid,\frow+\rowhei) ellipse (0.3 and 0.3);

					\node[pref] (players) [left = 2.5mm of aligner] {
						\begin{tabular}{C | C | C | C}
							\wor_1 & \wor_2 & \wor_3 & \wor_4   \\
							\hline  \hline 
							\fir_B & \fir_A & \fir_A & \fir_A \\
							\fir_A & \fir_B & \fir_B & \fir_B \\
							\fir_C & \fir_C & \fir_C & \fir_C \\
							\fir_D & \fir_D & \fir_D & \fir_D \\
							\unmat & \unmat & \unmat & \unmat \\
						\end{tabular}
					};	
					
					\node[pref] (coaches) [right = 2.5mm of aligner]{
						\begin{tabular}{C | C | C | C}
							\fir_A & \fir_B & \fir_C & \fir_D   \\
							\hline  \hline 
							\wor_3 & \wor_3 & \wor_3 & \wor_3 \\
							\wor_2 & \wor_1 & \wor_1 & \wor_1 \\
							\wor_1 & \wor_2 & \wor_2 & \wor_2 \\
							\wor_4 & \wor_4 & \wor_4 & \wor_4 \\
							\unmat & \unmat & \unmat & \unmat \\
						\end{tabular}
					};	
					
				\end{tikzpicture}
				\caption{First Pareto Improvement, $\bar\mat$}
			\end{subfigure}
			\begin{subfigure}{0.9\textwidth}
				\centering
				\begin{tikzpicture}
					
					\def\lshift{-3.92};
					\def\rshift{1};
					\def\colwid{0.933};
					\def\rowhei{1.14};
					\def\frow{-3.03};
					\def\mlw{0.5mm};
					\node (aligner) [] {};
					
					\draw[line width=\mlw, dashed, dash pattern=on 8pt off 3pt] (\lshift,\frow+\rowhei+\rowhei+\rowhei+\rowhei) ellipse (0.3 and 0.3);
					\draw[line width=\mlw, dashed, dash pattern=on 8pt off 3pt] (\lshift+\colwid,\frow+\rowhei+\rowhei+\rowhei+\rowhei) ellipse (0.3 and 0.3);
					\draw[line width=\mlw, dashed, dash pattern=on 8pt off 3pt] (\lshift+\colwid+\colwid,\frow+\rowhei) ellipse (0.3 and 0.3);
					\draw[line width=\mlw, dashed, dash pattern=on 8pt off 3pt] (\lshift+\colwid+\colwid+\colwid,\frow+\rowhei+\rowhei) ellipse (0.3 and 0.3);
					
					\draw[line width=\mlw, dashed, dash pattern=on 8pt off 3pt] (\rshift,\frow+\rowhei+\rowhei+\rowhei) ellipse (0.3 and 0.3);
					\draw[line width=\mlw, dashed, dash pattern=on 8pt off 3pt] (\rshift+\colwid,\frow+\rowhei+\rowhei+\rowhei) ellipse (0.3 and 0.3);
					\draw[line width=\mlw, dashed, dash pattern=on 8pt off 3pt] (\rshift+\colwid+\colwid,\frow+\rowhei) ellipse (0.3 and 0.3);
					\draw[line width=\mlw, dashed, dash pattern=on 8pt off 3pt] (\rshift+\colwid+\colwid+\colwid,\frow+\rowhei+\rowhei+\rowhei+\rowhei) ellipse (0.3 and 0.3);

					\node[pref] (players) [left = 2.5mm of aligner] {
						\begin{tabular}{C | C | C | C}
							\wor_1 & \wor_2 & \wor_3 & \wor_4   \\
							\hline  \hline 
							\fir_B & \fir_A & \fir_A & \fir_A \\
							\fir_A & \fir_B & \fir_B & \fir_B \\
							\fir_C & \fir_C & \fir_C & \fir_C \\
							\fir_D & \fir_D & \fir_D & \fir_D \\
							\unmat & \unmat & \unmat & \unmat \\
						\end{tabular}
					};	
					
					\node[pref] (coaches) [right = 2.5mm of aligner]{
						\begin{tabular}{C | C | C | C}
							\fir_A & \fir_B & \fir_C & \fir_D   \\
							\hline  \hline 
							\wor_3 & \wor_3 & \wor_3 & \wor_3 \\
							\wor_2 & \wor_1 & \wor_1 & \wor_1 \\
							\wor_1 & \wor_2 & \wor_2 & \wor_2 \\
							\wor_4 & \wor_4 & \wor_4 & \wor_4 \\
							\unmat & \unmat & \unmat & \unmat \\
						\end{tabular}
					};	
					
				\end{tikzpicture}
			\caption{Second Pareto improvement, $\dot\mat$}
			\end{subfigure}
			\caption{Pareto improvements of $\rema$.
			\\
			\\
			Note: throughout I use solid lines to denote the initial match and dashed lines to denote possible other matches}
			\label{fig:examplePE}
		\end{center}
	\end{figure}
	
	The story is different for the other Pareto improvement, labeled $\dot\mat$.
	In this match, $\wor_3$ is matched to $\fir_D$ and $\wor_4$ to $\fir_C$.
	Here, the coalition $\{\wor_3, \fir_C\}$ blocks the match.
	Because neither $\wor_3$ nor $\fir_C$ is under contract, no agent can prevent them from renegotiating a new match.
	This coalition qualifies as agreeable.
	The agreeable core intuitively selects the first match but not the second.
	
	To illustrate the mechanics of the Propose-Exchange algorithm, the following steps outline how tentative matches are proposed and refined until no further improvements can be made.
	To compute the first Pareto improvement ( $\bar\mat$ ), I leverage the Propose-Exchange algorithm.
	In this example the Propose stage takes worker $\wor_3$, who is initially unmatched, declares him \say{active.}
	The Propose stage allows active workers to make proposals to their favorite firm which has not rejected them so far.
	In the first step, both $\wor_3$ proposes to $\fir_A$, who tentatively accepts him.
	Because $\fir_A$ receives a proposal she prefers to her initial worker $\wor_1$, $\wor_1$ is now declared \say{active} as well.
	This guarantees that every firm weakly prefers the outcome of the Propose stage to the initial match $\rema$ because she only releases her initial worker once she has a more-preferred tentative match.
	In the second step, $\wor_1$ proposes to $\fir_B$, who tentatively accepts $\wor_1$.
	Again, because $\fir_B$ receives a proposal she prefers to her initial worker $\wor_2$, $\wor_2$ is now declared \say{active.}
	In the third step, $\wor_2$ proposes to $\fir_A$, who rejects him.
	In the fourth step, $\wor_2$ proposes to his initial firm $\fir_B$; the Propose stage requires that $\fir_B$ accept $\wor_2$'s proposal and reject $\wor_1$.
	This guarantees that every worker weakly prefers the outcome of the Propose stage to the initial match $\rema$.
	Continuing in this fashion, $\wor_1$ proposes to his initial firm $\fir_1$, which causes $\wor_3$ to be rejected.
	Worker $\wor_3$ proposes to $\fir_B$ and is rejected, and then to $\fir_C$ and is tentatively accepted.
	These steps are visualized in \Cref{fig:examplePropose}.
	The outcome of the Propose stage is denoted $\matProp$ and is depicted in panel (h).
	However, an agreeable blocking coalition still exists because workers $\wor_1$ and $\wor_2$ would prefer to exchange their initial firms $\fir_A$ and $\fir_B$, and these firms also would prefer the exchange.

	\begin{figure}
		\centering
		\begin{subfigure}{0.225\textwidth}
			\centering
			\begin{tikzpicture}
				\def\heightdist{2.30940107676cm}
				\def\mlw{0.5mm}
				\def\rotationangle{30}
				\def\minsizenode{1cm}
				\def\vertdist{2cm}
				
				\path (0,0) 
				node[circle, minimum size = \minsizenode](fa) {\large$\fir_A$}
				to ++(180:\vertdist)
				node[circle, minimum size = \minsizenode](w1){\large$\wor_1$}
				to ++(0-\rotationangle:\heightdist)
				node[circle, minimum size = \minsizenode](fb) {\large$\fir_B$}
				to ++(180:\vertdist)
				node[circle, minimum size = \minsizenode](w2) {\large$\wor_2$}
				to ++(0-\rotationangle:\heightdist)
				node[circle, minimum size = \minsizenode] (fc) {\large$\fir_C$}
				to ++(180:\vertdist)
				node[circle, minimum size = \minsizenode] (w3) {\large$\wor_3$}
				to ++(0-\rotationangle:\heightdist)
				node[circle, minimum size = \minsizenode] (fd) {\large$\fir_D$}
				to ++(180:\vertdist)
				node[circle, minimum size = \minsizenode] (w4) {\large$\wor_4$};
				
				\draw[->, line width = \mlw, dashed, dash pattern=on 8pt off 3pt] (w3) to (fa);
				
			\end{tikzpicture}
			\caption{Step 1}
		\end{subfigure}
		\begin{subfigure}{0.225\textwidth}
			\centering
			\begin{tikzpicture}
				\def\heightdist{2.30940107676cm}
				\def\mlw{0.5mm}
				\def\rotationangle{30}
				\def\minsizenode{1cm}
				\def\vertdist{2cm}

				\path (0,0) 
				node[circle, minimum size = \minsizenode](fa) {\large$\fir_A$}
				to ++(180:\vertdist)
				node[circle, minimum size = \minsizenode](w1){\large$\wor_1$}
				to ++(0-\rotationangle:\heightdist)
				node[circle, minimum size = \minsizenode](fb) {\large$\fir_B$}
				to ++(180:\vertdist)
				node[circle, minimum size = \minsizenode](w2) {\large$\wor_2$}
				to ++(0-\rotationangle:\heightdist)
				node[circle, minimum size = \minsizenode] (fc) {\large$\fir_C$}
				to ++(180:\vertdist)
				node[circle, minimum size = \minsizenode] (w3) {\large$\wor_3$}
				to ++(0-\rotationangle:\heightdist)
				node[circle, minimum size = \minsizenode] (fd) {\large$\fir_D$}
				to ++(180:\vertdist)
				node[circle, minimum size = \minsizenode] (w4) {\large$\wor_4$};

				\draw[->, line width = \mlw, dashed, dash pattern=on 8pt off 3pt] (w3) to (fa);
				\draw[->, line width = \mlw, dashed, dash pattern=on 8pt off 3pt] (w1) to (fb);
				
			\end{tikzpicture}
			\caption{Step 2}
		\end{subfigure}
		\begin{subfigure}{0.225\textwidth}
			\centering
			\begin{tikzpicture}
				\def\heightdist{2.30940107676cm}
				\def\mlw{0.5mm}
				\def\rotationangle{30}
				\def\minsizenode{1cm}
				\def\vertdist{2cm}
				
				\path (0,0) 
				node[circle, minimum size = \minsizenode](fa) {\large$\fir_A$}
				to ++(180:\vertdist)
				node[circle, minimum size = \minsizenode](w1){\large$\wor_1$}
				to ++(0-\rotationangle:\heightdist)
				node[circle, minimum size = \minsizenode](fb) {\large$\fir_B$}
				to ++(180:\vertdist)
				node[circle, minimum size = \minsizenode](w2) {\large$\wor_2$}
				to ++(0-\rotationangle:\heightdist)
				node[circle, minimum size = \minsizenode] (fc) {\large$\fir_C$}
				to ++(180:\vertdist)
				node[circle, minimum size = \minsizenode] (w3) {\large$\wor_3$}
				to ++(0-\rotationangle:\heightdist)
				node[circle, minimum size = \minsizenode] (fd) {\large$\fir_D$}
				to ++(180:\vertdist)
				node[circle, minimum size = \minsizenode] (w4) {\large$\wor_4$};
				
				\draw[->, line width = \mlw, dashed, dash pattern=on 8pt off 3pt] (w3) to (fa);
				\draw[->, line width = \mlw, dashed, dash pattern=on 8pt off 3pt] (w1) to (fb);
				\draw[->, line width = \mlw, dashed, dash pattern=on 8pt off 3pt] (w2) to (fa);
			\end{tikzpicture}
			\caption{Step 3}
		\end{subfigure}
		\vspace{1cm}
		\begin{subfigure}{0.225\textwidth}
			\centering
			\begin{tikzpicture}
				\def\heightdist{2.30940107676cm}
				\def\mlw{0.5mm}
				\def\rotationangle{30}
				\def\minsizenode{1cm}
				\def\vertdist{2cm}
				
				\path (0,0) 
				node[circle, minimum size = \minsizenode](fa) {\large$\fir_A$}
				to ++(180:\vertdist)
				node[circle, minimum size = \minsizenode](w1){\large$\wor_1$}
				to ++(0-\rotationangle:\heightdist)
				node[circle, minimum size = \minsizenode](fb) {\large$\fir_B$}
				to ++(180:\vertdist)
				node[circle, minimum size = \minsizenode](w2) {\large$\wor_2$}
				to ++(0-\rotationangle:\heightdist)
				node[circle, minimum size = \minsizenode] (fc) {\large$\fir_C$}
				to ++(180:\vertdist)
				node[circle, minimum size = \minsizenode] (w3) {\large$\wor_3$}
				to ++(0-\rotationangle:\heightdist)
				node[circle, minimum size = \minsizenode] (fd) {\large$\fir_D$}
				to ++(180:\vertdist)
				node[circle, minimum size = \minsizenode] (w4) {\large$\wor_4$};
				
				\draw[->, line width = \mlw, dashed, dash pattern=on 8pt off 3pt] (w3) to (fa);
				\draw[->, line width = \mlw, dashed, dash pattern=on 8pt off 3pt] (w1) to (fb);
				\draw[->, line width = \mlw, dashed, dash pattern=on 8pt off 3pt, color = lightgray] (w2) to (fa);
				\draw[->, line width = \mlw, dashed, dash pattern=on 8pt off 3pt] (w2) to (fb);
				
			\end{tikzpicture}
			\caption{Step 4}
		\end{subfigure}
		\begin{subfigure}{0.225\textwidth}
			\centering
			\begin{tikzpicture}
				\def\heightdist{2.30940107676cm}
				\def\mlw{0.5mm}
				\def\rotationangle{30}
				\def\minsizenode{1cm}
				\def\vertdist{2cm}
				
				\path (0,0) 
				node[circle, minimum size = \minsizenode](fa) {\large$\fir_A$}
				to ++(180:\vertdist)
				node[circle, minimum size = \minsizenode](w1){\large$\wor_1$}
				to ++(0-\rotationangle:\heightdist)
				node[circle, minimum size = \minsizenode](fb) {\large$\fir_B$}
				to ++(180:\vertdist)
				node[circle, minimum size = \minsizenode](w2) {\large$\wor_2$}
				to ++(0-\rotationangle:\heightdist)
				node[circle, minimum size = \minsizenode] (fc) {\large$\fir_C$}
				to ++(180:\vertdist)
				node[circle, minimum size = \minsizenode] (w3) {\large$\wor_3$}
				to ++(0-\rotationangle:\heightdist)
				node[circle, minimum size = \minsizenode] (fd) {\large$\fir_D$}
				to ++(180:\vertdist)
				node[circle, minimum size = \minsizenode] (w4) {\large$\wor_4$};
				
				\draw[->, line width = \mlw, dashed, dash pattern=on 8pt off 3pt] (w3) to (fa);
				\draw[->, line width = \mlw, dashed, dash pattern=on 8pt off 3pt, color = lightgray] (w1) to (fb);
				\draw[->, line width = \mlw, dashed, dash pattern=on 8pt off 3pt, color = lightgray] (w2) to (fa);
				\draw[->, line width = \mlw, dashed, dash pattern=on 8pt off 3pt] (w2) to (fb);
				\draw[->, line width = \mlw, dashed, dash pattern=on 8pt off 3pt] (w1) to (fa);
				
			\end{tikzpicture}
			\caption{Step 5}
		\end{subfigure}
	\vspace{1cm}
		\begin{subfigure}{0.225\textwidth}
			\centering
			\begin{tikzpicture}
				\def\heightdist{2.30940107676cm}
				\def\mlw{0.5mm}
				\def\rotationangle{30}
				\def\minsizenode{1cm}
				\def\vertdist{2cm}
				
				\path (0,0) 
				node[circle, minimum size = \minsizenode](fa) {\large$\fir_A$}
				to ++(180:\vertdist)
				node[circle, minimum size = \minsizenode](w1){\large$\wor_1$}
				to ++(0-\rotationangle:\heightdist)
				node[circle, minimum size = \minsizenode](fb) {\large$\fir_B$}
				to ++(180:\vertdist)
				node[circle, minimum size = \minsizenode](w2) {\large$\wor_2$}
				to ++(0-\rotationangle:\heightdist)
				node[circle, minimum size = \minsizenode] (fc) {\large$\fir_C$}
				to ++(180:\vertdist)
				node[circle, minimum size = \minsizenode] (w3) {\large$\wor_3$}
				to ++(0-\rotationangle:\heightdist)
				node[circle, minimum size = \minsizenode] (fd) {\large$\fir_D$}
				to ++(180:\vertdist)
				node[circle, minimum size = \minsizenode] (w4) {\large$\wor_4$};
				
				\draw[->, line width = \mlw, dashed, dash pattern=on 8pt off 3pt, color = lightgray] (w3) to (fa);
				\draw[->, line width = \mlw, dashed, dash pattern=on 8pt off 3pt, color = lightgray] (w1) to (fb);
				\draw[->, line width = \mlw, dashed, dash pattern=on 8pt off 3pt, color = lightgray] (w2) to (fa);
				\draw[->, line width = \mlw, dashed, dash pattern=on 8pt off 3pt] (w2) to (fb);
				\draw[->, line width = \mlw, dashed, dash pattern=on 8pt off 3pt] (w1) to (fa);
				\draw[->, line width = \mlw, dashed, dash pattern=on 8pt off 3pt] (w3) to (fb);
				
			\end{tikzpicture}
			\caption{Step 6}
		\end{subfigure}
		\begin{subfigure}{0.225\textwidth}
			\centering
			\begin{tikzpicture}
				\def\heightdist{2.30940107676cm}
				\def\mlw{0.5mm}
				\def\rotationangle{30}
				\def\minsizenode{1cm}
				\def\vertdist{2cm}
				
				\path (0,0) 
				node[circle, minimum size = \minsizenode](fa) {\large$\fir_A$}
				to ++(180:\vertdist)
				node[circle, minimum size = \minsizenode](w1){\large$\wor_1$}
				to ++(0-\rotationangle:\heightdist)
				node[circle, minimum size = \minsizenode](fb) {\large$\fir_B$}
				to ++(180:\vertdist)
				node[circle, minimum size = \minsizenode](w2) {\large$\wor_2$}
				to ++(0-\rotationangle:\heightdist)
				node[circle, minimum size = \minsizenode] (fc) {\large$\fir_C$}
				to ++(180:\vertdist)
				node[circle, minimum size = \minsizenode] (w3) {\large$\wor_3$}
				to ++(0-\rotationangle:\heightdist)
				node[circle, minimum size = \minsizenode] (fd) {\large$\fir_D$}
				to ++(180:\vertdist)
				node[circle, minimum size = \minsizenode] (w4) {\large$\wor_4$};
				
				\draw[->, line width = \mlw, dashed, dash pattern=on 8pt off 3pt, color = lightgray] (w3) to (fa);
				\draw[->, line width = \mlw, dashed, dash pattern=on 8pt off 3pt, color = lightgray] (w1) to (fb);
				\draw[->, line width = \mlw, dashed, dash pattern=on 8pt off 3pt, color = lightgray] (w2) to (fa);
				\draw[->, line width = \mlw, dashed, dash pattern=on 8pt off 3pt] (w2) to (fb);
				\draw[->, line width = \mlw, dashed, dash pattern=on 8pt off 3pt] (w1) to (fa);
				\draw[->, line width = \mlw, dashed, dash pattern=on 8pt off 3pt, color = lightgray] (w3) to (fb);
				\draw[->, line width = \mlw, dashed, dash pattern=on 8pt off 3pt] (w3) to (fc);
				
			\end{tikzpicture}
			\caption{Step 7}
		\end{subfigure}
		\begin{subfigure}{0.225\textwidth}
			\centering
			\begin{tikzpicture}
				\def\heightdist{2.30940107676cm}
				\def\mlw{0.5mm}
				\def\rotationangle{30}
				\def\minsizenode{1cm}
				\def\vertdist{2cm}
				
				\path (0,0) 
				node[circle, minimum size = \minsizenode](fa) {\large$\fir_A$}
				to ++(180:\vertdist)
				node[circle, minimum size = \minsizenode](w1){\large$\wor_1$}
				to ++(0-\rotationangle:\heightdist)
				node[circle, minimum size = \minsizenode](fb) {\large$\fir_B$}
				to ++(180:\vertdist)
				node[circle, minimum size = \minsizenode](w2) {\large$\wor_2$}
				to ++(0-\rotationangle:\heightdist)
				node[circle, minimum size = \minsizenode] (fc) {\large$\fir_C$}
				to ++(180:\vertdist)
				node[circle, minimum size = \minsizenode] (w3) {\large$\wor_3$}
				to ++(0-\rotationangle:\heightdist)
				node[circle, minimum size = \minsizenode] (fd) {\large$\fir_D$}
				to ++(180:\vertdist)
				node[circle, minimum size = \minsizenode] (w4) {\large$\wor_4$};
				
				\draw[->, line width = \mlw, dashed, dash pattern=on 8pt off 3pt] (w2) to (fb);
				\draw[->, line width = \mlw, dashed, dash pattern=on 8pt off 3pt] (w1) to (fa);
				\draw[->, line width = \mlw, dashed, dash pattern=on 8pt off 3pt] (w3) to (fc);
				\draw[->, line width = \mlw, dashed, dash pattern=on 8pt off 3pt] (w4) to (fd);
				
			\end{tikzpicture}
			\caption{Output $\matProp$}
		\end{subfigure}
		
		\caption{A visualization of the steps of the Propose stage.
		The black dashed lines indicate active proposals, and the light gray dashed lines indicate rejected proposals.
		Note that $\wor_1$ and $\wor_2$ only make proposals \textit{after} $\fir_A$ and $\fir_B$ have each received a proposal, respectively.
		Worker $\wor_4$ never makes a proposal because $\fir_D$ never receives a proposal.}
		\label{fig:examplePropose}
	\end{figure}
	
	The Exchange stage modifies the outcome of the Propose stage to remove the agreeable blocking coalition $\{\wor_1, \wor_2, \fir_A, \fir_B\}$ of $\matProp$.
	In the Exchange stage, workers $\wor_1$, $\wor_2$, and $\wor_4$ and firms $\fir_A$, $\fir_B$, and $\fir_D$ are active because they have not improved their initial match through the Propose stage, while $\wor_3$ and $\fir_3$ are inactive.
	In the first step, $\wor_1$ points to $\fir_B$ because $\fir_B$ is $\wor_1$'s most-preferred firm.
	Again, $\wor_2$ points to $\fir_A$ because $\fir_A$ is $\wor_2$'s most-preferred active firm.
	Worker $\wor_4$ would point to either $\fir_A$ or $\fir_B$, but neither prefer him to their initial match, so $\wor_4$ is only allowed to point to his own firm $\fir_D$.
	This will guarantee that every firm weakly prefers the outcome of the Exchange stage to the initial match $\rema$.
	Each active firm points to her initial worker.
	The cycle $\wor_1 \rightarrow \fir_B \rightarrow \wor_2 \rightarrow \fir_A \rightarrow \wor_1$ forms, and $\wor_1$ and $\wor_2$ are both permanently matched to the firms they point at.
	The cycle $\wor_4 \rightarrow \fir_D$ also forms, and $\wor_4$ is permanently matched to $\fir_D$.
	The output is $\matExch$, which is depicted in \Cref{fig:exampleExch}.
	As expected, $\matExch$ is the first Pareto improvement that was discussed, the unique element of the agreeable core.
	\begin{figure}
		\centering
		\begin{subfigure}{0.3\textwidth}
			\centering
			\begin{tikzpicture}
				\def\heightdist{2.30940107676cm}
				\def\mlw{0.5mm}
				\def\rotationangle{30}
				\def\minsizenode{1cm}
				\def\vertdist{2cm}
				
				\path (0,0) 
				node[circle, minimum size = \minsizenode](fa) {\large$\fir_A$}
				to ++(180:\vertdist)
				node[circle, minimum size = \minsizenode](w1){\large$\wor_1$}
				to ++(0-\rotationangle:\heightdist)
				node[circle, minimum size = \minsizenode](fb) {\large$\fir_B$}
				to ++(180:\vertdist)
				node[circle, minimum size = \minsizenode](w2) {\large$\wor_2$}
				to ++(0-\rotationangle:\heightdist)
				node[circle, minimum size = \minsizenode] (fc) {\textcolor{lightgray}{\large$\fir_C$}}
				to ++(180:\vertdist)
				node[circle, minimum size = \minsizenode] (w3) {\textcolor{lightgray}{\large$\wor_3$}}
				to ++(0-\rotationangle:\heightdist)
				node[circle, minimum size = \minsizenode] (fd) {\large$\fir_D$}
				to ++(180:\vertdist)
				node[circle, minimum size = \minsizenode] (w4) {\large$\wor_4$};
				
				\draw[->, line width = \mlw, color=lightgray] (fc) edge [in = -30, out = 30, out looseness = 4, in looseness = 4] (fc);
				\draw[<-, line width = \mlw, color=lightgray] (w3) edge [in = 150, out = 210, out looseness = 4, in looseness = 4] (w3);
				
				\draw[->, line width = \mlw, dashed, dash pattern=on 8pt off 3pt, bend left = 30] (w2) to (fb);
				\draw[->, line width = \mlw, dashed, dash pattern=on 8pt off 3pt, bend left = 30] (w1) to (fa);
				\draw[->, line width = \mlw, dashed, dash pattern=on 8pt off 3pt, color=lightgray] (w3) to (fc);
				\draw[->, line width = \mlw, dashed, dash pattern=on 8pt off 3pt, bend left = 30] (w4) to (fd);
				
				\draw[<-, line width = \mlw] (w2) to (fb);
				\draw[<-, line width = \mlw] (w1) to (fa);
				\draw[<-, line width = \mlw] (w4) to (fd);
				
			\end{tikzpicture}
			\caption{$\rema$ (solid), $\matProp$ (dashed)}
		\end{subfigure}
		\begin{subfigure}{0.3\textwidth}
			\centering
			\begin{tikzpicture}
				\def\heightdist{2.30940107676cm}
				\def\mlw{0.5mm}
				\def\rotationangle{30}
				\def\minsizenode{1cm}
				\def\vertdist{2cm}
				
				\path (0,0) 
				node[circle, minimum size = \minsizenode](fa) {\large$\fir_A$}
				to ++(180:\vertdist)
				node[circle, minimum size = \minsizenode](w1){\large$\wor_1$}
				to ++(0-\rotationangle:\heightdist)
				node[circle, minimum size = \minsizenode](fb) {\large$\fir_B$}
				to ++(180:\vertdist)
				node[circle, minimum size = \minsizenode](w2) {\large$\wor_2$}
				to ++(0-\rotationangle:\heightdist)
				node[circle, minimum size = \minsizenode] (fc) {\textcolor{lightgray}{\large$\fir_C$}}
				to ++(180:\vertdist)
				node[circle, minimum size = \minsizenode] (w3) {\textcolor{lightgray}{\large$\wor_3$}}
				to ++(0-\rotationangle:\heightdist)
				node[circle, minimum size = \minsizenode] (fd) {\large$\fir_D$}
				to ++(180:\vertdist)
				node[circle, minimum size = \minsizenode] (w4) {\large$\wor_4$};
				
				\draw[->, line width = \mlw, color=lightgray] (fc) edge [in = -30, out = 30, out looseness = 4, in looseness = 4] (fc);
				\draw[<-, line width = \mlw, color=lightgray] (w3) edge [in = 150, out = 210, out looseness = 4, in looseness = 4] (w3);
				
				\draw[->, line width = \mlw, dashed, dash pattern=on 8pt off 3pt] (w1) to (fb);
				\draw[->, line width = \mlw, dashed, dash pattern=on 8pt off 3pt] (w2) to (fa);
				\draw[->, line width = \mlw, dashed, dash pattern=on 8pt off 3pt, color=lightgray] (w3) to (fc);
				\draw[->, line width = \mlw, dashed, dash pattern=on 8pt off 3pt, bend left = 30] (w4) to (fd);
				
				\draw[<-, line width = \mlw] (w2) to (fb);
				\draw[<-, line width = \mlw] (w1) to (fa);
				\draw[<-, line width = \mlw] (w4) to (fd);
				
			\end{tikzpicture}
			\caption{Step 1}
		\end{subfigure}
		\begin{subfigure}{0.3\textwidth}
			\centering
			\begin{tikzpicture}
				\def\heightdist{2.30940107676cm}
				\def\mlw{0.5mm}
				\def\rotationangle{30}
				\def\minsizenode{1cm}
				\def\vertdist{2cm}
				
				\path (0,0) 
				node[circle, minimum size = \minsizenode](fa) {\large$\fir_A$}
				to ++(180:\vertdist)
				node[circle, minimum size = \minsizenode](w1){\large$\wor_1$}
				to ++(0-\rotationangle:\heightdist)
				node[circle, minimum size = \minsizenode](fb) {\large$\fir_B$}
				to ++(180:\vertdist)
				node[circle, minimum size = \minsizenode](w2) {\large$\wor_2$}
				to ++(0-\rotationangle:\heightdist)
				node[circle, minimum size = \minsizenode] (fc) {\large$\fir_C$}
				to ++(180:\vertdist)
				node[circle, minimum size = \minsizenode] (w3) {\large$\wor_3$}
				to ++(0-\rotationangle:\heightdist)
				node[circle, minimum size = \minsizenode] (fd) {\large$\fir_D$}
				to ++(180:\vertdist)
				node[circle, minimum size = \minsizenode] (w4) {\large$\wor_4$};
				
				\draw[->, line width = \mlw, dashed, dash pattern=on 8pt off 3pt] (w1) to (fb);
				\draw[->, line width = \mlw, dashed, dash pattern=on 8pt off 3pt] (w2) to (fa);
				\draw[->, line width = \mlw, dashed, dash pattern=on 8pt off 3pt] (w3) to (fc);
				\draw[->, line width = \mlw, dashed, dash pattern=on 8pt off 3pt] (w4) to (fd);

			\end{tikzpicture}
			\caption{Output $\matExch$}
		\end{subfigure}
		\caption{A visualization of the steps of the Exchange stage.
		Agents $\wor_3$ and $\fir_C$ are excluded because their $\rema$- and $\matProp$-partners differ.
		There is only one step because each active agent is in a cycle in Step 1.}
		\label{fig:exampleExch}
	\end{figure}
	
	The Propose-Exchange algorithm involves both a \say{free market} phase in the Propose stage (but with participation restrictions on $\wor_1$, $\wor_2$, and $\wor_4$) as well as a \say{trading} phase in which could $\wor_1$, $\wor_2$, and $\wor_4$ exchanged their firm.
	The match at every stage of the Propose-Exchange algorithm is an improvement of the initial match.

	\section{Model}\label{sec:model}
	In this section I present a one-to-one matching model.
	Although many of my applications are many-to-one (e.g.\ many students match to one school), I defer a discussion of the nuances until \Cref{sec:discussion}.
	In most examples, the many-to-one case is a simple extension of the one-to-one model.
	Below I introduce the elements of a \textit{matching problem}, which is a tuple $(\Wor, \Fir, \spref, \rema)$ consisting of workers, firms, a profile of preferences, and an initial match.
	
	There is a set of workers $\Wor$ and the set of firms $\Fir$, and the union of both is the set of agents $\Age \equiv \Wor \cup \Fir$.
	For clarity of exposition I use masculine pronouns for workers and feminine pronouns for firms.
	Every worker $\wor\in\Wor$ has preference $\wpref_\wor$ over $\Fir\cup \{\wor\}$ and every $\fir\in\Fir$ has a preference $\wpref_\fir$ over $\Wor\cup \{\fir\}$.
	A preference for oneself is a preference to be unmatched: if $\age$ prefers $\age$ to $\bge$ this means that $\age$ prefers to remain unmatched than to match to $\bge$.
	Throughout I assume that $\wpref_\age$ is complete, reflexive, transitive, and anti-symmetric; that is, that $\age$ can rank partners from most to least preferred with no ties.
	
	A match is a function $\mat$ that takes in an agent $\age$ and returns the agent $\mat(\age)$ that he or she is matched, where $\age = \mat(\age)$ means that $\age$ is unmatched.
	Formally, \textit{match} is a function $\mat : \Age \rightarrow \Age$ such that:
	\begin{enumerate}
		\item if $\wor \in \Wor$ then $\mat(\wor) \in \Fir \cup \{\wor\}$; and
		\item if $\fir\in \Fir$ then $\mat(\fir) \in \Wor \cup \{\fir\}$; and
		\item $\mat(\mat(\age)) = \age$. 
	\end{enumerate}
	The first two require that the match is two-sided: every worker matches to a firm (or is unmatched) and every firm matches to a worker (or is unmatched).
	The third requires that every agent is matched to the agent matched to him or her.
	If $\mat(\age) = \age$ then $\age$ is \textit{$\mat$-unmatched}; otherwise, $\mat(\age)$ is the \textit{$\mat$-partner} of $\age$ (or possibly $\mat$-firm or $\mat$-worker).
	I write $\mat \wpref_X \mat '$ to mean $\mat(x) \wpref_x \mat ' (x)$ for all $x\in X$.
	
	There is an \textit{initial} match $\rema$.
	The initial match limits the set of matches I consider to the set of matches I consider to those satisfying the following:
	\begin{definition}
		Match $\mat$ is \textit{individually rational} if $\mat \wpref_\Age \rema$.
	\end{definition}
	The interpretation is that if an agent prefers their initial match $\rema$ to the proposed match $\mat$, they retain the right to demand $\rema$, as it represents an enforceable agreement.
	
	\subsection{The Core}
	Here I formally introduce the core, which is the set of all individually rational matches not blocked by any coalition of agents.
	A coalition blocks a match if it can collectively form a match within the coalition that everyone weakly prefers to the current match.
	Formally, a \textit{coalition} $\Coal \subseteq \Age$ is a nonempty\footnote{Coalitions throughout the paper assumed nonempty. For ease of exposition this quantifier will not be listed.} subset of agents who may form a match among themselves.
	Let $\mat(\Coal) \equiv \{\mat(\age) \spbr \age\in \Coal\}$.
	Note that if $\mat(\Coal) \subseteq \Coal$, then $\mat(\Coal) = \Coal$.
	If a coalition weakly prefers a match $\mat ' $ to $\mat$ and $\mat '$ only matches agents in $\Coal$ to agents in $\Coal$, then $\Coal$ may block $\mat$; formally,
	\begin{definition}
		Coalition $\Coal  $ \textit{blocks $\mat$ through $\nat $} if $\nat \wpref_\Coal \mat$, $\nat (\age) \spref_\age\mat(\age)$ for at least one $\age\in \Coal$, and $\nat (\Coal) = \Coal$.
	\end{definition}
	The \textit{core} is the set of all individually rational matches not blocked by any coalition through any match.\footnote{
		Formally, this is the \textit{strong core} because I consider all \textit{weak} blocks (allowing some coalition members to be indifferent between $\mat ' $ and $\mat$).
		In two-sided matching without indifferences all weak blocks are strong blocks.
		Because the coalitions I will consider later will usually contain agents who do not change partners, I use the strong core as it is smaller.
	}
	\subsection{The Agreeable Core}
	\Cref{exa1} demonstrates that the core may be empty.
	The nonexistence of a match that is both individually rational and unblocked by every coalition of agents motivates restricting either the matches a coalition can block through or the coalitions considered.
	The choice is nontrivial and hinges upon the interpretation of the initial match.
	
	If the matches that a coalition can block through are restricted, then the natural requirement is that any coalition can block but only through an individually rational match $\mat$.
	The interpretation is that the initial match is inviolable \textit{ex post}.
	In order to block a match, a coalition needs only to suggest an individually rational match; as long as all agents are weakly better off than at $\rema$, no agent can complain about his or her partner.
	However, it is easy to construct examples where this solution is empty.
	
	The alternative is to restrict the set of coalitions but not the matches they can block through.
	The interpretation is that the initial match is not only inviolable \textit{ex post} but also that any new contract formed by an agent requires the \textit{ex ante} approval of his or her $\rema$-partner.
	I consider only coalitions meeting the following criterion:
	\begin{definition}
		A coalition $\Coal$ is \textit{agreeable} if $\rema(\Coal) = \Coal$.
	\end{definition}
	A coalition $\Coal$ is agreeable if any contract in $\rema$ does not contain both an agent in $\Coal$ and an agent not in $\Coal$.
	By restricting my attention to agreeable coalitions, I require that every agent in a blocking coalition of $\mat$ guarantees his or her $\rema$-partner a weakly better partner at match $\mat '$ than at $\mat$.
	To guarantee such an improvement, the $\rema$-partner's partner at $\mat '$ must also be included in the coalition, which implies that the $\rema$-partner's $\mat '$-partner must also be included in the coalition, and so on.
	\Cref{def:agrcor} formalizes this idea.	
	\begin{definition}\label{def:agrcor}
		The \textit{agreeable core} is the set of individually rational matches not blocked by any agreeable coalition.
	\end{definition}
	
	The agreeable core puts a strong requirement on blocking coalitions: every agent in the coalition \textit{and their $\rema$-partners} must be made weakly better off.
	My interpretation is that if some agent $\age$ is harmed by a block and his or her $\rema$-partner is in the blocking coalition, then $\age$ can veto the block by refusing to dissolve the initial contract.
	The important nuance is that the harmed agent can veto $\mat '$ even if he or she prefers $\mat ' $ to $\rema$.
	
	The veto power inherent in the agreeable core allows one member of a initial match to dictate the matches his or her partner can form.
	The picture to have in mind is both agents in a initial match simultaneously searching for better matches.
	They both agree to cancel their initial match \textit{simultaneous to both confirming new partners.}
	Because the match of one partner influences who is willing to match with the other, both must agree not only to cancel their initial match but also approve of the other's new match.
	By only considering agreeable coalitions, I allow agents to veto a blocking coalition before the coalition acts.
	
	I find the following justification for the agreeable core helpful in explaining the agreeable core and how I allow agents veto blocking coalitions \textit{ex ante}.
	For a given initial match $\rema$, agents are considering forming the  individually rational match $\mat$.
	Before $\mat$ is realized among the agents (say, before the agents cancel their initial agreements and form the $\mat$ agreements), a coalition considers enforcing some match $\mat ' $ among themselves.
	If some agent $\age$ is in the coalition but $\rema(\age)$ is not in the coalition, then $\rema(\age)$ may refuse to permit $\age$ to form $\mat '$ unless $\rema(\age)$ is certain he or she will prefer $\mat ' $ to $\mat$.
	Hence, $\rema(\mat)$ must also be in the coalition.
	
	Perhaps surprisingly, the set of matches not blocked by any agreeable coalition is not a subset of the individually rational matches.
	My  definition of blocking coalition does not allow an agent to demand $\rema$, and hence the restriction to individually rational matches is substantive.
	For a simple example, restrict \Cref{exa1} to just contractor $1$ and city $A$.
	The match $\mat(1) = 1$ and $\mat(A) = A$ is not blocked by any coalition but does not Pareto improve $\rema$.
	
	I devote \Cref{sec:PEalg} to developing the machinery to prove my main result, namely, that the agreeable core is never empty.
	In the remainder of this section I briefly touch on several aspects of the agreeable core that do not require my more involved techniques.
	\Cref{sec:efficiency} shows that the agreeable core is always Pareto efficient, and conversely if $\rema$ is Pareto efficient then $\{\rema\}$ is the agreeable core.
	As alluded to in the introduction, my model features several connections with both the classical model of stability \citep{gale_college_1962} and more recent models of reassignment \citep{combe_design_2022, pereyra_dynamic_2013}.
	In \Cref{sec:connectiontostability} and \Cref{sec:connectiontoreassignment} I develop these connections; as an expository device and a prelude to my algorithm, I highlight the two leading algorithms in two-sided matching---the Deferred Acceptance and the Top Trading Cycles algorithms---and their adaptations used in the literature to guarantee individual rationality.
	
	\subsection{Efficiency}\label{sec:efficiency}
	In this subsection I investigate the efficiency of the agreeable core.	
	My first observation is that no match in the agreeable core is Pareto dominated:\footnote{
		I say that $\nat$ \textit{Pareto dominates} $\mat$ if every agent weakly prefers $\nat$ to $\mat$ and at least one agent strictly prefers $\nat$ to $\mat$.
	} if $\nat$ Pareto dominates $\mat$, then the grand coalition $\Age$ (which is always agreeable) blocks $\mat$ through $\nat$.
	My  second observations is a kind of a converse: if $\rema$ is not Pareto dominated, then $\rema$ is in the agreeable core.
	To see this, suppose (toward a contradiction) that some agreeable coalition $\Coal$ blocks $\rema$ through $\mat$.
	But then because $\rema(\Coal)= \mat(\Coal) = \Coal$, I can define $\mat ' $ that agrees with $\mat$ for agents in $\Coal$ and agrees with $\rema$ everywhere else.
	But $\mat ' $ then Pareto dominates $\rema$, a contradiction to the supposition that $\rema$ is Pareto efficient.	
	\begin{remark}\label{rem:efficiency}
		Every $\mat$ in the agreeable core is Pareto efficient.\footnote{
			If $\mat$ is not Pareto dominated by any $\nat$, then $\mat$ is \textit{Pareto efficient}.
		}
		Moreover, $\rema$ is Pareto efficient if and only if the $\rema$ is the unique element of the agreeable core.
	\end{remark}
	
	\Cref{rem:efficiency} assures us that the agreeable core  satisfies the most common efficiency standard.

	\subsection{Connection to Stability}\label{sec:connectiontostability}
	In this subsection I discuss the parallels between the agreeable core and the classic theory of stability introduced by \cite{gale_college_1962}.
	The models are the same except that the classical model does not include an initial match in the primitives.
	This connection allows me to leverage a significant tool from two-sided stability, the Deferred Acceptance algorithm (DA), in my analysis
	
	In the classic model, a \textit{blocking pair} of a match is any worker and firm pair such that both prefer each other to their match.
	A match is \textit{stable} if all agents prefer their match to being unmatched and there are no blocking pairs of the match.
	It is well-known \citep{Roth_Sotomayor_1990} that the set of stable matches is the core that I defined previously.
	My  definition of the agreeable core guarantees that if $\rema(\age) = \age$ for all $\age \in \Age$, then the agreeable core corresponds to the core because every coalition is agreeable.
	Therefore stability is the special case of the agreeable core when $\rema$ leaves all agents unmatched.
	
	\cite{gale_college_1962} gives an efficient algorithm for constructing a stable match: the Deferred Acceptance algorithm (\Cref{alg:da}).
	Initially, the DA \say{activates} every worker and designates every agent as \say{currently unmatched.}
	At every step of the DA, some active worker matched proposes to the firm he prefers the most {among those he has not proposed to yet} (if he would rather be unmatched, he is matched to himself and deactivated).
	Every firm then reviews the proposals she receives and her current match and rejects all but her most preferred proposal or match.
	The process continues until no more workers are matched active.
	
	\begin{algorithm}
		\caption{Deferred Acceptance (DA) algorithm}\label{alg:da}
		\begin{algorithmic}
			\State {Notation: when I write $\matDA(\age) \gets \wor$, I mean that $\age$ is matched to $\wor$ and $\wor$ is deactivated. If another worker $\wor ' $ was matched to $\age$, then $\age$ rejects $\wor'$, $\wor ' $ is matched to himself, and $\wor ' $ is activated.}
			\State
			\State {set} $\matDA(\fir) \gets \fir$ for all $\fir \in \Fir$.
			\State {activate} every worker.
			\While{some worker $\wor$ is activated}
			\State $\wor$ proposes to his most-preferred firm $\fir$ that he has not yet proposed to; if he would rather be unmatched, instead he proposes to himself and is deactivated, and we set $\matDA(\wor) \gets \wor$.
			\State \textbf{if} $\fir$ prefers $\wor$ to $\matDA(\fir)$ then set $\matDA(\fir) \gets \wor$.
			\State \textbf{else} $\fir$ rejects $\wor$.
			\EndWhile
			\State \textbf{return} $\matDA$
		\end{algorithmic}
	\end{algorithm}

	Although guaranteed to produce a match unblocked by any coalition, the DA fails to satisfy individual rationality (see \cite{pereyra_dynamic_2013} and \cite{ combe_design_2022}).
	There are two ways in which individual rationality can fail.
	First, a worker may strictly prefer his $\rema$-partner to his match.
	\cite{pereyra_dynamic_2013} resolves this issue by requiring that each firm accepts her $\rema$-partner if he proposes to her.
	This modification guarantees that workers find the outcome individually rational because no worker proposes to a less preferred firm without being rejected by his $\rema$-partner.
	
	In my setting firms also have individual rationality constraints.
	The DA fails to accommodate these because a worker makes proposals (and may be matched to another firm) even though his $\rema$-firm has not received a proposal she prefers to the worker.
	I will see in \Cref{sec:proposestage} how to resolve this tension by limiting which workers can propose.

	\subsection{Connection to Reassignment}\label{sec:connectiontoreassignment}
	In this subsection I highlight the connection between the agreeable core and the standard model of reassignment.
	Recent research in reassignment seeks to find a match through a strategyproof mechanism that is both individually rational and maximizes some objective function (see \citep{combe_design_2022, dur_two-sided_2019} for two such examples).
	Because the agreeable core is motivated with first principles (the core) rather than with an objective in mind (obtaining a strategyproof mechanism), there are substantial differences in definitions and results.
	However, both approaches employ the same method: the Top Trading Cycles algorithm (TTC).
	
	The TTC finds a match such that no coalition of workers can reallocate their $\rema$-firms among themselves and improve their matches.
	The TTC starts with every worker and firm \say{active.}
	At every step, every active firm points at the worker she is initially matched to, and every active worker points at his favorite active firm.
	At every step a cycle must form.
	The TTC assigns each worker in the cycle to the firm he points at, and then the agents in the cycle become inactive.
	The process terminates when no agents are active.
	
	I define the TTC in \Cref{alg:ttc}.
	\begin{algorithm}
		\caption{Top Trading Cycles (TTC) algorithm}\label{alg:ttc}
		\begin{algorithmic}
			\State set $\matTTC(\age) = \age$ for all $\age$.
			\State {every agent is activated.}
			\While{at least one agent is active}
			\State every active worker points to his most-preferred of the active firms.
			\State every active firm points to her most preferred of the active workers.
			\State choose an arbitrary cycle $(\wor_1, \fir_2, \ldots \wor_{2k-1} \equiv \wor_1, \fir_{2k} \equiv \fir_{2})$ such that every agent points to the next agent in the cycle.
			\State all agents in the cycle are deactivated.
			\State match every $\wor_k$ to $\fir_{k+1}$.
			\EndWhile
			\State \textbf{return} $\matTTC$
		\end{algorithmic}
	\end{algorithm}
	
	If some agents are matched by $\rema$, then the TTC may not be individually rational.
	To accommodate this, \cite{combe_design_2022} and \cite{combe_reallocation_2023} make the following two modifications.
	First, a firm must point to her $\rema$-worker so long as he is active.
	This guarantees that $\matTTC \wpref_\Wor \rema$.
	Second, no worker may point to a firm if that firm prefers her $\rema$-partner to the worker.
	This guarantees that $\matTTC \wpref_\Wor \rema$.
	
	In my setting, however, these modifications are not enough.
	As I saw in \Cref{sec:connectiontostability}, the agreeable core equals the set of stable matches when all agents are $\rema$-unmatched.
	At least in this case firms must be given power to decide between the workers pointing to them, as in the DA.
	In section \Cref{sec:exchangestage} I incorporate this by limiting which workers and firms participate in the TTC.

	\section{A Proof of Existence: The Propose-Exchange Algorithm}\label{sec:PEalg}
	In this section I present a computationally efficient and economically meaningful algorithm that always produces a match $\matExch$ (defined through this section) in the agreeable core.
	My algorithm is the \textit{Propose-Exchange} algorithm (PE) and is composed of two stages.
	The Propose stage resembles the Deferred Acceptance algorithm and eliminates any block by a coalition that either includes an agent who is unmatched in the initial match or who becomes unmatched by the block.
	The Exchange stage resembles the Top Trading Cycles algorithm and eliminates all blocks that involve reshuffling initial partners among themselves.
	For readers unfamiliar with the Deferred Acceptance and the Top Trading Cycles algorithms, I refer the reader to \Cref{sec:connectiontostability} and \Cref{sec:connectiontoreassignment}, respectively.
	
	The PE directly implies that the agreeable core exists and provides some insight into its structure.
	My main result is the following:
	\begin{theorem}\label{thm:exi}
		$\matExch$ is in the agreeable core.
	\end{theorem}
	The proof (and definition of $\matExch$) occupies the remainder of this section.
	I first introduce a particular directed graph representation of the matching problem in \Cref{sec:graphtheory}, then introduce the Propose stage in \Cref{sec:proposestage}, and finally the Exchange stage in \Cref{sec:exchangestage}.
	I conclude this section by noting how the introduction of initial matches creates additional complexity in analyzing the structure of the agreeable core.
	All omitted proofs are contained in \Cref{appendix}.

	\subsection{A Graph-Theoretic Depiction}\label{sec:graphtheory}
	Despite my parsimonious definition of the agreeable core, so far testing whether $\mat$ is in the agreeable core requires checking whether any coalition can block $\mat$ through any $\mat '$, which is only feasible in small examples.
	My main result from this subsection is a characterization of blocking coalitions in terms of paths in a directed graph, which is computationally efficient.
	I use the language of graph theory to formalize my ideas.
	
	A \textit{digraph} $\Graph$ is a pair $ (\Vertices, \Edges)$ where $\Vertices$ is a set of \textit{vertices} and $\Edges$ is a set of \textit{ordered} pairs of vertices called (directed) \textit{edges}, possibly including an edge from a vertex to itself, called a \textit{loop}.
	The one nuance to my construction is that I allow for loops to be repeated once in $\Edges$; formally, $\Edges$ is a \textit{multiset}, but this will not cause any confusion.
	
	I consider digraphs where the vertices are agents, and the edges represent matches.
	Edges going from $\Fir$ to $\Wor$ (and loops) are drawn from $\rema$, while the edges going from $\Wor$ to $\Fir$ (and possibly repeated loops) are drawn from $\mat$ and any blocking pairs of $\mat$.
	I abuse notation and write $\rema$ for both the function and for the set of ordered pairs:
	\begin{align*}
		\rema &= \{(\fir, \wor) \spbr \rema(\fir) = \wor\} \cup \{(\age, \age) \spbr \rema(\age) = \age\}\\
		\mat &= \{(\wor, \fir) \spbr \mat(\wor) = \fir\} \cup \{(\age, \age) \spbr \mat(\age) = \age\}.
	\end{align*}
	It is critical to understand that $\rema$ and $\mat$ go in \textit{opposite} directions (except for any loops).
	I always follow the convention that edges from the initial match travel from $\Fir$ to $\Wor$, so although the matches $\rema$ or $\mat$ may change, from context the direction of the edges is always clear.
	To include the blocking pairs, I define:
	\begin{align*}
		\Imp(\mat) &= \{(\wor, \fir) \spbr \fir \spref_\wor \mat(\wor)  \t{ and } \spref_\fir \mat(\fir) \} \cup \{(\age, \age) \spbr \age \spref_\wor \mat(\age)\}
	\end{align*}
	My main digraph of interest is $(\Age, \rema \cup \mat \cup \Imp(\mat))$.
	That is, the vertices are agents, the first set of edges connects initial partners, and the second set of edges connects all pairs that weakly prefer each other over their $\mat$-partners.
	
	\Cref{fig:blockingGraph} depicts the three types of edges using the set-up of \Cref{exa1} and a match $\mat$ that modifies the initial match $\rema$ by leaving $\fir_D$ unmatched and matching $\wor_4$ to $\fir_C$.
	Subfigure (a) includes the edges from $\rema$, which are either loops (in the case of $\wor_3$ and $\fir_C$) or point from $\Fir$ to $\Wor$.
	Subfigure (b) includes the edges from $\mat$, which point from $\Wor$ to $\Fir$.
	Subfigure (c) includes the blocking pairs of $\mat$, which point from $\Wor$ to $\Fir$.
	\begin{figure}
		\centering
		\begin{subfigure}{0.3\textwidth}
				\centering
				\begin{tikzpicture}
						\def\heightdist{2.30940107676cm}
						\def\mlw{0.5mm}
						\def\rotationangle{30}
						\def\minsizenode{1cm}
						\def\vertdist{2cm}
						
						\path (0,0) 
						node[circle, minimum size = \minsizenode](fa) {\large$\fir_A$}
						to ++(180:\vertdist)
						node[circle, minimum size = \minsizenode](w1){\large$\wor_1$}
						to ++(0-\rotationangle:\heightdist)
						node[circle, minimum size = \minsizenode](fb) {\large$\fir_B$}
						to ++(180:\vertdist)
						node[circle, minimum size = \minsizenode](w2) {\large$\wor_2$}
						to ++(0-\rotationangle:\heightdist)
						node[circle, minimum size = \minsizenode] (fc) {\large$\fir_C$}
						to ++(180:\vertdist)
						node[circle, minimum size = \minsizenode] (w3) {\large$\wor_3$}
						to ++(0-\rotationangle:\heightdist)
						node[circle, minimum size = \minsizenode] (fd) {\large$\fir_D$}
						to ++(180:\vertdist)
						node[circle, minimum size = \minsizenode] (w4) {\large$\wor_4$};
						
						\draw[->, line width = \mlw, dashed, dash pattern=on 8pt off 3pt] (fa) edge [in = -30, out = 30, out looseness = 4, in looseness = 4, draw = none] (fa);
						\draw[->, line width = \mlw] (fc) edge [in = -30, out = 30, out looseness = 4, in looseness = 4] (fc);
						\draw[<-, line width = \mlw] (w3) edge [in = 150, out = 210, out looseness = 4, in looseness = 4] (w3);
						\draw[-, line width = \mlw, dashed, dash pattern=on 8pt off 3pt] (w3) edge [in = 150, out = 210, out looseness = 4, in looseness = 4, draw = none] (w3);
						\draw[<-, line width = \mlw, dashed, dash pattern=on 8pt off 3pt] (fd) edge [in = -30, out = 30, out looseness = 4, in looseness = 4, draw = none] (fd);
						
						\draw[<-, line width = \mlw] (w1) to (fa);
						\draw[<-, line width = \mlw] (w2) to (fb);
						\draw[<-, line width = \mlw] (w4) to (fd);
						
					\end{tikzpicture}
				\caption{Initial Match $\rema$}
				\label{fig:subfig1a}
			\end{subfigure}
		\begin{subfigure}{0.3\textwidth}
				\centering
				\begin{tikzpicture}
						\def\heightdist{2.30940107676cm}
						\def\mlw{0.5mm}
						\def\rotationangle{30}
						\def\minsizenode{1cm}
						\def\vertdist{2cm}

						\path (0,0) 
						node[circle, minimum size = \minsizenode](fa) {\large$\fir_A$}
						to ++(180:\vertdist)
						node[circle, minimum size = \minsizenode](w1){\large$\wor_1$}
						to ++(0-\rotationangle:\heightdist)
						node[circle, minimum size = \minsizenode](fb) {\large$\fir_B$}
						to ++(180:\vertdist)
						node[circle, minimum size = \minsizenode](w2) {\large$\wor_2$}
						to ++(0-\rotationangle:\heightdist)
						node[circle, minimum size = \minsizenode] (fc) {\large$\fir_C$}
						to ++(180:\vertdist)
						node[circle, minimum size = \minsizenode] (w3) {\large$\wor_3$}
						to ++(0-\rotationangle:\heightdist)
						node[circle, minimum size = \minsizenode] (fd) {\large$\fir_D$}
						to ++(180:\vertdist)
						node[circle, minimum size = \minsizenode] (w4) {\large$\wor_4$};
						
						\draw[<-, line width = \mlw, dashed, dash pattern=on 8pt off 3pt] (fa) edge [in = -30, out = 30, out looseness = 4, in looseness = 4, draw = none] (fa);
						\draw[-, line width = \mlw] (fc) edge [in = -30, out = 30, out looseness = 4, in looseness = 4, draw = none] (fc);
						\draw[-, line width = \mlw] (w3) edge [in = 150, out = 210, out looseness = 4, in looseness = 4, draw = none] (w3);
						\draw[->, line width = \mlw, dashed, dash pattern=on 8pt off 3pt] (w3) edge [in = 150, out = 210, out looseness = 4, in looseness = 4] (w3);
						\draw[<-, line width = \mlw, dashed, dash pattern=on 8pt off 3pt] (fd) edge [in = -30, out = 30, out looseness = 4, in looseness = 4] (fd);

						\draw[->, line width = \mlw, dashed, dash pattern=on 8pt off 3pt] (w1) to (fa);
						\draw[->, line width = \mlw, dashed, dash pattern=on 8pt off 3pt] (w2) to (fb);
						\draw[->, line width = \mlw, dashed, dash pattern=on 8pt off 3pt] (w4) to (fc);

					\end{tikzpicture}
				\caption{Proposed Match $\mat$}
				\label{fig:subfig1b}
			\end{subfigure}
		\begin{subfigure}{0.3\textwidth}
				\centering
				\begin{tikzpicture}
						\def\heightdist{2.30940107676cm}
						\def\mlw{0.5mm}
						\def\rotationangle{30}
						\def\minsizenode{1cm}
						\def\vertdist{2cm}
						
						\path (0,0) 
						node[circle, minimum size = \minsizenode](fa) {\large$\fir_A$}
						to ++(180:\vertdist)
						node[circle, minimum size = \minsizenode](w1){\large$\wor_1$}
						to ++(0-\rotationangle:\heightdist)
						node[circle, minimum size = \minsizenode](fb) {\large$\fir_B$}
						to ++(180:\vertdist)
						node[circle, minimum size = \minsizenode](w2) {\large$\wor_2$}
						to ++(0-\rotationangle:\heightdist)
						node[circle, minimum size = \minsizenode] (fc) {\large$\fir_C$}
						to ++(180:\vertdist)
						node[circle, minimum size = \minsizenode] (w3) {\large$\wor_3$}
						to ++(0-\rotationangle:\heightdist)
						node[circle, minimum size = \minsizenode] (fd) {\large$\fir_D$}
						to ++(180:\vertdist)
						node[circle, minimum size = \minsizenode] (w4) {\large$\wor_4$};
						
						\draw[-, line width = \mlw, dashed, dash pattern=on 8pt off 3pt] (fa) edge [in = -30, out = 30, out looseness = 4, in looseness = 4, draw = none] (fa);
						\draw[-, line width = \mlw] (fc) edge [in = -30, out = 30, out looseness = 4, in looseness = 4, draw = none] (fc);
						\draw[-, line width = \mlw] (w3) edge [in = 150, out = 210, out looseness = 4, in looseness = 4, draw = none] (w3);
						\draw[-, line width = \mlw, dashed, dash pattern=on 8pt off 3pt] (w3) edge [in = 150, out = 210, out looseness = 4, in looseness = 4, draw = none] (w3);
						\draw[<-, line width = \mlw, dashed, dash pattern=on 8pt off 3pt] (fd) edge [in = -30, out = 30, out looseness = 4, in looseness = 4, draw = none] (fd);
						
						\draw[->, line width = \mlw, dash pattern=on 0.5pt off 7.5pt, line cap=round] (w1) to (fb);
						\draw[->, line width = \mlw, dash pattern=on 0.5pt off 7.5pt, line cap=round] (w2) to (fa);
						\draw[->, line width = \mlw, dash pattern=on 0.5pt off 7.5pt, line cap=round] (w3) to (fa);
						\draw[->, line width = \mlw, dash pattern=on 0.5pt off 7.5pt, line cap=round] (w3) to (fb);
						\draw[->, line width = \mlw, dash pattern=on 0.5pt off 7.5pt, line cap=round] (w3) to (fc);
						\draw[->, line width = \mlw, dash pattern=on 0.5pt off 7.5pt, line cap=round] (w3) to (fd);

					\end{tikzpicture}
				\caption{Blocking Pairs of $\mat$}
				\label{fig:subfig1c}
			\end{subfigure}
	
			 \vspace{2cm}
			 
			 \begin{subfigure}{0.45\textwidth}
			 	\centering
			 	\begin{tikzpicture}
				 		\def\heightdist{2.30940107676cm}
				 		\def\mlw{0.5mm}
				 		\def\rotationangle{30}
				 		\def\minsizenode{1cm}
				 		\def\vertdist{2cm}
				 		
				 		\path (0,0) 
				 		node[circle, minimum size = \minsizenode](fa) {\large$\fir_A$}
				 		to ++(180:\vertdist)
				 		node[circle, minimum size = \minsizenode](w1){\large$\wor_1$}
				 		to ++(0-\rotationangle:\heightdist)
				 		node[circle, minimum size = \minsizenode](fb) {\large$\fir_B$}
				 		to ++(180:\vertdist)
				 		node[circle, minimum size = \minsizenode](w2) {\large$\wor_2$}
				 		to ++(0-\rotationangle:\heightdist)
				 		node[circle, minimum size = \minsizenode] (fc) {\large$\fir_C$}
				 		to ++(180:\vertdist)
				 		node[circle, minimum size = \minsizenode] (w3) {\large$\wor_3$}
				 		to ++(0-\rotationangle:\heightdist)
				 		node[circle, minimum size = \minsizenode] (fd) {\large$\fir_D$}
				 		to ++(180:\vertdist)
				 		node[circle, minimum size = \minsizenode] (w4) {\large$\wor_4$};
				 		
				 		\draw[-, line width = \mlw, dashed, dash pattern=on 8pt off 3pt] (fa) edge [in = -30, out = 30, out looseness = 4, in looseness = 4, draw = none] (fa);
				 		\draw[->, line width = \mlw] (fc) edge [in = -30, out = 30, out looseness = 4, in looseness = 4] (fc);
				 		\draw[<-, line width = \mlw] (w3) edge [in = 150, out = 210, out looseness = 4, in looseness = 4] (w3);
				 		\draw[-, line width = \mlw, dashed, dash pattern=on 8pt off 3pt] (w3) edge [in = 150, out = 210, out looseness = 4, in looseness = 4, draw = none] (w3);
				 		\draw[<-, line width = \mlw, dashed, dash pattern=on 8pt off 3pt] (fd) edge [in = -30, out = 30, out looseness = 4, in looseness = 4, draw = none] (fd);

				 		\draw[->, line width = \mlw, dash pattern=on 0.5pt off 7.5pt, line cap=round] (w3) to (fd);
				 		\draw[<-, line width = \mlw] (w4) to (fd);
				 		\draw[->, line width = \mlw, dashed, dash pattern=on 8pt off 3pt] (w4) to (fc);
				 		
				 	\end{tikzpicture}
			 	\caption{Acyclic Blocking Path}
			 	\label{fig:subfig1d}
			\end{subfigure}
	 	\begin{subfigure}{0.45\textwidth}
		 		\centering
		 		\begin{tikzpicture}
			 			\def\heightdist{2.30940107676cm}
			 			\def\mlw{0.5mm}
			 			\def\rotationangle{30}
			 			\def\minsizenode{1cm}
			 			\def\vertdist{2cm}
			 			
			 			\path (0,0) 
			 			node[circle, minimum size = \minsizenode](fa) {\large$\fir_A$}
			 			to ++(180:\vertdist)
			 			node[circle, minimum size = \minsizenode](w1){\large$\wor_1$}
			 			to ++(0-\rotationangle:\heightdist)
			 			node[circle, minimum size = \minsizenode](fb) {\large$\fir_B$}
			 			to ++(180:\vertdist)
			 			node[circle, minimum size = \minsizenode](w2) {\large$\wor_2$}
			 			to ++(0-\rotationangle:\heightdist)
			 			node[circle, minimum size = \minsizenode] (fc) {\large$\fir_C$}
			 			to ++(180:\vertdist)
			 			node[circle, minimum size = \minsizenode] (w3) {\large$\wor_3$}
			 			to ++(0-\rotationangle:\heightdist)
			 			node[circle, minimum size = \minsizenode] (fd) {\large$\fir_D$}
			 			to ++(180:\vertdist)
			 			node[circle, minimum size = \minsizenode] (w4) {\large$\wor_4$};
			 			
			 			\draw[-, line width = \mlw, dashed, dash pattern=on 8pt off 3pt] (fa) edge [in = -30, out = 30, out looseness = 4, in looseness = 4, draw = none] (fa);
			 			\draw[<-, line width = \mlw] (fc) edge [in = -30, out = 30, out looseness = 4, in looseness = 4, draw = none] (fc);
			 			\draw[<-, line width = \mlw] (w3) edge [in = 150, out = 210, out looseness = 4, in looseness = 4, draw = none] (w3);
			 			\draw[-, line width = \mlw, dashed, dash pattern=on 8pt off 3pt] (w3) edge [in = 150, out = 210, out looseness = 4, in looseness = 4, draw = none] (w3);
			 			\draw[<-, line width = \mlw, dashed, dash pattern=on 8pt off 3pt] (fd) edge [in = -30, out = 30, out looseness = 4, in looseness = 4, draw = none] (fd);

			 			\draw[->, line width = \mlw, dash pattern=on 0.5pt off 7.5pt, line cap=round] (w2) to (fa);
			 			\draw[<-, line width = \mlw] (w1) to (fa);
			 			\draw[->, line width = \mlw, dashed, dash pattern=on 0.5pt off 7.5pt, line cap=round] (w1) to (fb);
			 			\draw[<-, line width = \mlw] (w2) to (fb);
			 			
			 		\end{tikzpicture}
		 		\caption{Cyclic Blocking Path}
		 		\label{fig:subfig1e}
		 	\end{subfigure}

		\caption{The blocking digraph $(\Age, \rema \cup \mat \cup \Imp(\mat))$ is the union of the digraphs in subfigures (a), (b), and (c).
		Subfigures (d) and (e) depict the two kinds of blocking paths.}
		\label{fig:blockingGraph}
	\end{figure}
	
	A (simple) \textit{path} in $(\Vertices, \Edges)$ is a vector of edges $\Path = (\edg_1, \dots, \edg_n)$ such that the second coordinate of $\edg_k$ equals the first coordinate of $\edg_{k+1}$ for $1 \leq k < n$ and no vertex appears in more than two edges.
	Recall that a loop may appear twice (in both $\rema$ and $\mat \cup \Imp(\mat)$) so it is possible for path to consist of exactly two loops.
	I say a \textit{vertex is in a path} if the path contains an edge that contains the vertex.
	I sometimes abuse notation and write $\Path$ for the vertices in $\Path$.
	
	A path $\Path$ is \textit{complete} if every vertex contained in the path is contained in exactly two edges of the path.
	A path is \textit{alternating} if it no two consecutive edges (including loops) alternate between $\rema$ and $\mat \cup \Imp(\mat)$.\footnote{Although the directed nature of the digraph makes most paths alternating, by formally requiring that a path is alternating I rule out the case that $(\wor, \fir )$ and $(\fir, \fir)$ may both be from $\mat \cup \Imp(\mat)$.}
	Two complete and alternating paths are depicted in subfigures (d) and (e) of \Cref{fig:blockingGraph}.
	For an arbitrary complete and alternating path $\Path$ in $(\Age, \rema \cup \mat \cup \Imp(\mat)) $, I define $\mat^\Path(\age)$ as follows:
	\begin{itemize}
		\item if $(\wor, \fir)$ is in $\Path$, then $\mat^\Path(\wor) = \fir$.
		\item if $\age$ is \textit{not} in $\Path$ but $\mat(\age) $ is in $\Path$, then $\mat^\Path(\age) = \age$;
		\item if $\age$ is \textit{not} in $\Path$ and $\mat(\age)$ is not in $\Path$, then $\mat^\Path(\age) = \mat(\age)$.
	\end{itemize}
	That is, $\mat^\Path$ matches $\age \in \Path$ to the agent whom $\age$ shares an edge from $\mat \cup \Imp (\mat)$ in $\Path$ with; other matches are left unchanged where possible.
	By \Cref{lem:altClean} in the appendix, every agent in $\Path$ is contained in one edge from $\rema$ and one edge is from $\mat \cup \Imp(\mat)$, so $\mat^\Path$ is well defined and $\mat^\Path(\Path) = \Path$.
	
	My main result of this subsection is that a path that is complete, alternating, and contains an edge from $\Imp(\mat)$ corresponds to an agreeable blocking coalition in $(\Age, \rema \cup \mat \cup \Imp(\mat))$.
	I formalize this as follows:
	\begin{definition}
		Path $\Path$ is a \textit{blocking path of $\mat$} if $\Path$ is a complete and alternating path in $(\Age, \rema \cup \mat \cup \Imp(\mat))$ that contains at least one edge from $\Imp(\mat)$.
	\end{definition}
	A blocking path of $\mat$ is aptly named as it corresponds to a blocking coalition of $\mat$.
	\begin{proposition}\label{pro:ac-iff-nbps}
		An individually rational match $\mat$ is in the agreeable core if and only if $\mat$ admits no blocking paths.
		Moreover, if $\Path$ is a blocking path of $\mat$ then $\Path$ blocks $\mat$ through $\mat^\Path$.
	\end{proposition}
	\Cref{pro:ac-iff-nbps} provides a test that is linear in the number of edges to see if $\mat$ is in the agreeable core.\footnote{
		A depth first search initiated from every edge in $\Imp(\mat)$ is sufficient.
	}
	
	The Propose-Exchange algorithm is built on a partition of paths between those that form cycles and those that do not:
	\begin{definition}
		Let $\Path = (\edg_1, \ldots, \edg_n)$.
		If the first coordinate of $\edg_1$ is the second coordinate of $ \edg_n$, then $\Path$ is \textit{cyclic}; otherwise, $\Path$ is \textit{acyclic}.
	\end{definition}
	As the name suggests, cyclic paths start with an agent and then return to that agent.
	In $(\Age, \rema \cup \mat \cup \Imp(\mat))$, a cyclic, complete, and alternating path corresponds to agents (who are $\rema$-matched) trading their $\rema$-partners among themselves.
	Acyclic paths that are also complete and alternating start with a loop and end with a loop, forming a line in the digraph.
	See subfigures (d) and (e) of \Cref{fig:blockingGraph} for example cyclic and acyclic blocking paths.
	In $(\Age, \rema \cup \mat \cup \Imp(\mat))$, an acyclic, complete, and alternating path corresponds to agents trading their $\rema$-firms among themselves, except that two agents are unmatched by one or both sets of edges.
	The Propose-Exchange algorithm works by first producing a match $\matProp$ that admits no acyclic blocking paths, then finding a series of Pareto improvements of $\matProp$ to produce a match $\matExch$ that has no cyclic blocking paths.

	\subsection{The Propose Stage}\label{sec:proposestage}
	The first stage of my algorithm outputs a match $\matProp$ by systematically removing all {acyclic} blocking paths from $(\Age, \rema \cup  \mat \cup \Imp(\mat))$.
	An acyclic blocking path $\Path$ in $(\Age, \rema  \cup \mat \cup \Imp(\mat))$ corresponds to a series of trades, but the agents at either end of the path are either $\rema$-unmatched or $\mat^\Path$-unmatched.
	These may be thought of as a cycle that includes the \say{unmatched} agent.
	
	The \Propose algorithm is a variation of the Deferred Acceptance algorithm.
	The DA is designed for markets where all agents are unmatched under $\rema$ and is defined in \Cref{alg:da}.
	I noted in \Cref{sec:connectiontostability} that the DA may fail individual rationality for both workers and firms.
%
	I provide guarantees to the agents in the \Propose algorithm by only allowing a worker to make a proposal once his $\rema$-firm has received a more preferred proposal and by requiring that a firm accept any proposal from her $\rema$-worker.
	These adjustments, shown in italics, are essential to the success of the Propose stage.
	The Propose stage algorithm is defined in \Cref{alg:prop}.
	\begin{algorithm}
		\caption{\Propose Stage algorithm}\label{alg:prop}
		\begin{algorithmic}
			\State {Notation: when I write $\matProp(\age) \gets \wor$, I mean that $\age$ is matched to $\wor$ and $\wor$ is deactivated. If another worker $\wor ' $ was matched to $\age$, then $\age$ rejects $\wor'$, $\wor ' $ is matched to himself, and $\wor ' $ is activated.}
			\State
			\State \textit{{set} $\matProp \gets \rema$}
			\State {activate} every worker.
			\State \textit{\textbf{if} $\wor$'s $\rema$-firm prefers $\wor$ to being unmatched, \textbf{then} \textbf{deactivate} $\wor$.}
			\While{some worker $\wor$ is active}
			\State $\wor$ proposes to his most-preferred firm $\fir$ that he has not yet proposed to; if he would rather be unmatched, instead he proposes to himself and is deactivated, and we set $\matProp(\wor) \gets \wor$.
			\State \textit{\textbf{if} $\fir$ is $\wor$'s $\rema$-partner, \textbf{then} set $\matProp(\fir) \gets \wor$ and have $\fir$ reject all future proposals.}
			\State \textbf{else if} $\fir$ prefers $\wor$ to $\matProp(\fir)$\textit{ and to being unmatched}, then set $\matProp(\fir) \gets \wor$.
			\State \textbf{else} $\fir$ rejects $\wor$.
			\EndWhile
			\State \textbf{return} $\matProp$
		\end{algorithmic}
	\end{algorithm}
	By construction, $\matProp$ is individually rational.
	If $\wor$ strictly prefers $\rema$ to $\matProp$, then $\wor$ would have proposed to $\rema$ (and not been rejected).
	Again, if $\rema(\fir)$ is matched by $\matProp$ to a firm other than $\fir$, then $\fir$ received a proposal she prefers to $\rema(\fir)$ and hence she prefers $\matProp$ to $\rema$.
	
	I then show that at the end of the \Propose algorithm, no blocking path of $\matProp$ is acyclic.
	\begin{lemma}\label{lem:propose}
		$\matProp$ admits no acyclic blocking paths.
	\end{lemma}
	My proof leverages that an acyclic blocking path $\Path$ in $(\Age,\rema  \cup  \mat \cup \Imp(\mat))$ always begins with either a worker who is $\rema$-unmatched and hence proposes or a firm who is $\mat^\Path$-unmatched (and hence her $\rema$-worker starts out active).
	Because the start and finish of the path are connected by workers who (weakly) prefer the firm they receive in the block, I can show that every worker in the path must have had the opportunity to propose.
	I then show that the path must terminate with either a worker who is $\rema$-matched or a firm who is $\rema$-unmatched, neither of which would reject the proposal made through the path.
	I conclude by showing that every firm accepts the proposal from her $\mat^\Path$-partner, which contradicts that $\mat \neq \mat^\Path$.
	
	The \textbf{while} step admits ambiguity because which worker is selected to propose is not specified.
	I show in \Cref{pro:ord} that the order in which workers are selected is irrelevant. 
	\begin{proposition}\label{pro:ord}
		The output of the Propose stage is independent of the order the workers are called to propose in.
	\end{proposition}
	
	\subsection{The Exchange Stage}\label{sec:exchangestage}
	In the second stage of the algorithm, I eliminate all cyclic blocking paths.
	I do this by allowing agents to trade their initial agreements.
	Cyclic blocking paths in $(\Age, \rema  \cup \mat \cup \Imp(\mat))$ correspond to workers and their $\rema$-firms rearranging their initial matches among themselves.
	No agent in a cyclic path is unmatched by either $\mat$ or $\rema$.
	A cyclic blocking path represents an inefficient allocation for $\Coal$: the coalition could have rearranged their initial matches among themselves and obtained a better match.

	The \Exchange algorithm is an adaptation of the Top Trading Cycles algorithm to find these cycles and remove them.
	The difficulty with using solely the TTC in my setting is that the TTC does not give firms the ability to select \textit{between} workers.
	Although firm's preferences limit the set of acceptable workers, which worker is matched to the firm ultimately depends on the worker the firm is required to point at.
	If only some workers or firms are matched by $\rema$, then the firm's lack of choice can lead to violations of the agreeable core.
	I resolve this by only applying the TTC to workers and firms who did not both find better partners through the \Propose algorithm, with my addition indicated in italics.
	This modification guarantees that the Exchange stage is a Pareto improvement of $\matProp$; by selecting a Pareto improvement, I do not create any new acyclic blocking paths in the blocking digraph.
	The \Exchange algorithm is defined in \Cref{alg:exch}.
	
	
	\begin{algorithm}
		\caption{\Exchange Stage algorithm}\label{alg:exch}
		\begin{algorithmic}
			\State \textit{set $\matExch(\age) = \matProp(\age)$ for all $\age$.}
			\State \textit{every $\wor$ such that $\matProp(\wor) = \rema(\wor)$ is activated with $\rema(\wor)$.}
			\While{at least one worker is active}
				\State every active worker points to his most-preferred of the active firms \textit{who prefer him to her $\rema$-worker}.
				\State every active firm points\textit{ to her $\rema$-worker.}
				\State choose an arbitrary cycle $(\wor_1, \fir_2, \ldots \wor_{2k-1} \equiv \wor_1, \fir_{2k} \equiv \fir_{2})$ such that every agent points to the next agent in the cycle.
				 \State all agents in the cycle sit down.
				\State match every $\wor_k$ to $\fir_{k+1}$.
			\EndWhile
			\State \textbf{return} $\matExch$
		\end{algorithmic}
	\end{algorithm}
	
	My first observation is that the \Exchange algorithm makes no agents worse off than under $\matProp$.
	Workers only point to firms they prefer to $\rema$, and by my simplification of workers' preferences, firms can only be pointed at by workers they prefer to $\rema$.
	The result is that at the end of the \Exchange algorithm, $\matExch$ admits no cyclic blocking paths.
	\begin{lemma}\label{lem:exchange}
		$\matExch$ admits no cyclic blocking paths.
	\end{lemma}
	My proof leverages that if $\wor$ strictly prefers $\fir$ to $\matExch(\wor)$, then $\fir$ must sit down at least one step \textit{before} $\wor$.
	A cyclic blocking path then implies that the firms in the path sit down on average strictly before the workers in the path sit down.
	However, because every worker's $\rema$-firm is in the path and they sit down in the same step, it must be that the firms in the path sit down on average in the same step as the workers in the path sit down.
	This contradiction rules out cyclic blocking paths.
	
	\subsection{Existence}
	I am now ready to prove that $\matExch$ is in the agreeable core.
	\begin{proof}[\unskip\nopunct]
		\textbf{\textit{Proof of \Cref{thm:exi}:}}
		Suppose (toward a contradiction) that $\matExch$ is not in the agreeable core.
		Then by \Cref{pro:ac-iff-nbps} the digraph $(\Age, \rema  \cup  \matExch \cup \Imp(\matExch))$ contains a blocking path $\Path$.
		By \Cref{lem:exchange}, $\Path$ is acyclic.
		But $\Path$ is also blocking path in $(\Age, \rema  \cup  \matProp \cup \Imp(\matProp))$ because $\matExch \cup \Imp(\matExch) \subseteq \matProp \cup \Imp(\matProp)$ and $\Imp(\matExch) \subseteq \Imp(\matProp)$.
		By \Cref{lem:propose}, $\Path$ is not acyclic.
		This is a contradiction, which proves the claim.
	\end{proof}
	
	The importance of the Propose-Exchange algorithm in my proof cannot be understated.
	However, the algorithm has practical implications because it is also computationally efficient.
	The \Propose stage runs in polynomial time because each worker can make at most $|\Fir| + 1$ proposals.
	Similarly, one cycle is removed in every iteration of the \Exchange stage, and at most $|\Fir|$ cycles can be removed.
	An efficient algorithm is necessary for implementing the agreeable core in applications.

	\subsection{Structure}\label{sec:structure}
	In this subsection, I highlight the difficulty in characterizing the underlying structure of the agreeable core and how it relates to other classes of algorithms commonly used to compute core outcomes.
	Although the set of stable matches has a well-understood structure which I summarize in the following paragraph, the agreeable core is not as tame.
	The hurdle in the analysis comes from the \Exchange stage.
	To the best of my knowledge, there are no results from the literature that apply to the agreeable core when every agent is $\rema$-matched.
	
	I briefly summarize the main structural results on the set of stable matches.
	First, a \textit{lattice} is a partially ordered set $(\Lat, \lgeq)$ such that any two elements of $\Lat$ have a unique \textit{least upper bound}, called the \textit{join} of $x$ and $y$, and a unique \textit{greatest lower bound}, called the \textit{meet} of $x$ and $y$.
	That is, there is a unique $x \join y$ such that if $z \lgeq x$ and $z \geq y$ then $z \lgeq x\join y$, and there is a unique $x \meet y$ such that if $x \lgeq z$ and $y \geq x$ then $  x\meet y \lgeq z$.
	A key result in two-sided matching is that the set of stable matches forms a lattice with the partial order $\wpref_\Wor$.\footnote{
		Donald Knuth attributes this to John H. Conway.
	}
	The join of two matches $\mat$ and $\nat$ is the match that gives every worker $\wor$ his more preferred partner from $\{\mat(\wor), \nat(\wor)\}$ and every $\fir$ her less preferred partner from $\{\mat(\fir), \nat(\fir)\}$; the meet is given symmetrically.
	This implies that there is a conflict of interest between the workers and the firms: if every worker weakly prefers a stable $\mat$ to a stable $\nat$, then every firm weakly prefers $\nat$ to $\mat$.
	Moreover, there is a \textit{worker optimal} stable match and a \textit{firm optimal} stable match.
	
	To show that the agreeable core fails to be a lattice, consider the following example.
	Let $\rema(\wor_1) = \fir_A$, $\rema(\wor_2) = \fir_B$, and $\rema(\wor_3) = \fir_C$, and preferences are given as in \Cref{fig:nobest}.
	Both the pair $\wor_2$ and $\fir_B$ and the pair $\wor_3$ and $\fir_C$ prefer to participate in a cycle with the pair $\wor_1$ and $\fir_A$, but $\wor_1$ and $\fir_A$ have opposing preferences over the two possible cycles.
	Worker $\wor_1$ prefers firm $\fir_C$ and firm $\fir_A$ prefers worker $\wor_2$, and so either cycle may be in the agreeable core.
	The agreeable core consists uniquely of the $\bar\mat$ match and the $\dot\mat$ match, a pair which is not ordered by $\wpref_\Wor$.
	In this example there is no worker optimal match.
	\begin{figure}\label{fig:nobest}
		\centering
		\begin{subfigure}[b]{0.45\textwidth}
			\centering
			\begin{tikzpicture}
				\def\lshift{2.0};
				\def\rshift{1.88};
				\def\colwid{0.93};
				
				\def\firstrow{0.99};
				\def\secondrow{-0.16};
				\def\thirdrow{-1.31};
				\def\mlw{0.5mm};
				
				
				\draw[line width = \mlw] (-\colwid-\lshift,\thirdrow) ellipse (0.3 and 0.3); 
				\draw[line width = \mlw] (-\lshift,\secondrow) ellipse (0.3 and 0.3); 
				\draw[line width = \mlw] (+\colwid-\lshift,\secondrow) ellipse (0.3 and 0.3); 
				
				\draw[line width = \mlw] (-\colwid+\rshift,\thirdrow) ellipse(0.3 and 0.3); 
				\draw[line width = \mlw] (0+\rshift,\secondrow) ellipse (0.3 and 0.3); 
				\draw[line width = \mlw] (+\colwid+\rshift,\secondrow) ellipse (0.3 and 0.3); 
				
				\node (aligner) {};
				
				\node[pref] (players) [left = 2mm of aligner] {
					\begin{tabular}{C | C | C}
						\wor_1 & \wor_2 & \wor_3  \\
						\hline  \hline 
						\fir_C &\fir_A & \fir_A \\
						\fir_B & \fir_B& \fir_C \\
						\fir_A&  \unmat &   \unmat\\
						\unmat &  \fir_C & \fir_C \\
					\end{tabular}
				};    
				
				\node[pref] (coaches) [right = 2mm of aligner] {
					\begin{tabular}{C | C | C}
						\fir_A & \fir_B & \fir_C\\
						\hline  \hline 
						\wor_2 &   \wor_1 &\wor_1 \\
						\wor_3 &  \wor_2 &  \wor_3\\
						\wor_1  &  \unmat & \unmat  \\
						\unmat &  \wor_3 & \wor_2 \\
					\end{tabular}
				};    
			\end{tikzpicture}
			\caption{Initial match---$\rema$}
		\end{subfigure}
		\hfill
		\begin{subfigure}[b]{0.45\textwidth}
			\centering
			\begin{tikzpicture}
				\def\lshift{2.0};
				\def\rshift{1.88};
				\def\colwid{0.93};
				
				\def\firstrow{0.99};
				\def\secondrow{-0.16};
				\def\thirdrow{-1.31};
				\def\mlw{0.5mm};
				
				\draw[line width = \mlw, dashed, dash pattern=on 8pt off 3pt] (-\colwid-\lshift,\secondrow) ellipse (0.3 and 0.3); 
				\draw[line width = \mlw, color = lightgray] (-\colwid-\lshift,\thirdrow) ellipse (0.3 and 0.3); 
				\draw[ line width = \mlw, dashed, dash pattern=on 8pt off 3pt] (0-\lshift,\firstrow) ellipse (0.3 and 0.3); 
				\draw[line width = \mlw, color = lightgray] (0-\lshift,\secondrow) ellipse (0.3 and 0.3); 
				\draw[ line width = \mlw, dashed, dash pattern=on 8pt off 3pt] (+\colwid-\lshift,\secondrow) ellipse (0.3 and 0.3); 

				\draw[ line width = \mlw, dashed, dash pattern=on 8pt off 3pt] (-\colwid+\rshift,\firstrow) ellipse(0.3 and 0.3); 
				\draw[line width = \mlw, color = lightgray] (-\colwid+\rshift,\thirdrow) ellipse(0.3 and 0.3); 
				\draw[ line width = \mlw, dashed, dash pattern=on 8pt off 3pt] (0+\rshift,\firstrow) ellipse (0.3 and 0.3); 
				\draw[line width = \mlw, color = lightgray] (0+\rshift,\secondrow) ellipse (0.3 and 0.3); 
				\draw[ line width = \mlw, dashed, dash pattern=on 8pt off 3pt] (+\colwid+\rshift,\secondrow) ellipse (0.3 and 0.3); 
				
				\draw[->, thick, lightgray] 
				(-\colwid-\lshift-0.2,\thirdrow+0.186) 
				to[out=135,in=215] 
				(-\colwid-\lshift,\firstrow-0.55)
				to
				(0-\lshift-0.27,\firstrow-0.17);
				
				\draw[->, thick, lightgray] 
				(-\lshift-0.1,\secondrow-0.28)
				to[out = 240, in = -60]
				(-\colwid-\lshift+0.11,\secondrow-0.3);
				
				\draw[->, thick, lightgray] 
				(-\colwid+\rshift-0.2,\thirdrow+0.186) 
				to[out=135,in=215] 
				(-\colwid+\rshift,\firstrow-0.55)
				to
				(0+\rshift-0.27,\firstrow-0.17); 
				
				\draw[->, thick, lightgray] 
				(+\rshift -0.321,\secondrow)
				to[out = 160, in = 270]
				(-\colwid+\rshift-0.11,\firstrow-0.3); 
				
				\node (aligner) {};
				
				\node[pref] (players) [left = 2mm of aligner] {
					\begin{tabular}{C | C | C}
						\wor_1 & \wor_2 & \wor_3  \\
						\hline  \hline 
						\fir_C & \fir_A & \fir_A \\
						\fir_B & \fir_B & \fir_C\\
						\fir_A&  \unmat &   \unmat\\
						\unmat &  \fir_C &\fir_ C \\
					\end{tabular}
				};    
				
				\node[pref] (coaches) [right = 2mm of aligner] {
					\begin{tabular}{C | C | C}
						\fir_A & \fir_B & \fir_C\\
						\hline  \hline 
						\wor_2 &   \wor_1 & \wor_1 \\
						\wor_3 &  \wor_2 &  \wor_3 \\
						1  &  \unmat & \unmat  \\
						\unmat &  \wor_3 & \wor_2 \\
					\end{tabular}
				};    
			\end{tikzpicture}

			\caption{$\wor_2$, $\fir_A$, and $\fir_B$'s preferred match---$\bar\mat$}
		\end{subfigure}
		\vskip\baselineskip 
		\begin{subfigure}[b]{0.45\textwidth}
			\centering
			\begin{tikzpicture}
				\def\lshift{2.0};
				\def\rshift{1.88};
				\def\colwid{0.93};
				
				\def\firstrow{0.99};
				\def\secondrow{-0.16};
				\def\thirdrow{-1.31};
				\def\mlw{0.5mm};
				
				\draw[line width = \mlw, color = lightgray] (-\colwid-\lshift,\thirdrow) ellipse (0.3 and 0.3); 
				\draw[ line width = \mlw, dashed, dash pattern=on 8pt off 3pt] (-\colwid-\lshift,\firstrow) ellipse (0.3 and 0.3); 
				\draw[ line width = \mlw, dashed, dash pattern=on 8pt off 3pt] (0-\lshift,\secondrow) ellipse (0.3 and 0.3); 
				\draw[ line width = \mlw, dashed, dash pattern=on 8pt off 3pt] (+\colwid-\lshift,\firstrow) ellipse (0.3 and 0.3); 
				\draw[line width = \mlw, color = lightgray] (+\colwid-\lshift,\secondrow) ellipse (0.3 and 0.3); 
				
				\draw[ line width = \mlw, dashed, dash pattern=on 8pt off 3pt] (-\colwid+\rshift,\secondrow) ellipse(0.3 and 0.3); 
				\draw[line width = \mlw, color = lightgray] (-\colwid+\rshift,\thirdrow) ellipse(0.3 and 0.3); 
				\draw[ line width = \mlw, dashed, dash pattern=on 8pt off 3pt] (0+\rshift,\secondrow) ellipse (0.3 and 0.3); 
				\draw[ line width = \mlw, dashed, dash pattern=on 8pt off 3pt] (+\colwid+\rshift,\firstrow) ellipse (0.3 and 0.3); 
				\draw[line width = \mlw, color = lightgray] (+\colwid+\rshift,\secondrow) ellipse (0.3 and 0.3); 
				
				\draw[->, thick, lightgray] 
				(-\colwid-\lshift-0.2,\thirdrow+0.186) 
				to[out=135,in=203] 
				(-\colwid-\lshift,\firstrow-0.55)
				to
				(+\colwid-\lshift-0.27,\firstrow-0.17);
				
				\draw[->, thick, lightgray] 
				(+\colwid-\lshift-0.1,\secondrow+0.28)
				to[out = 120, in = -60]
				(-\colwid-\lshift+0.11,\firstrow-0.3);
				
				\draw[->, thick, lightgray] 
				(-\colwid+\rshift-0.2,\thirdrow+0.186) 
				to[out=135,in=203] 
				(-\colwid+\rshift,\firstrow-0.55)
				to
				(+\colwid+\rshift-0.27,\firstrow-0.17); 
				
				\draw[->, thick, lightgray] 
				(+\colwid+\rshift -0.1,\secondrow-0.28)
				to[out = 240, in = -60]
				(-\colwid+\rshift+0.11,\secondrow-0.3);
				
				\node (aligner) {};
				
				\node[pref] (players) [left = 2mm of aligner] {
					\begin{tabular}{C | C | C}
						\wor_1 & \wor_2 & \wor_3  \\
						\hline  \hline 
						\fir_C & \fir_A & \fir_A \\
						\fir_B & \fir_B & \fir_C\\
						\fir_A&  \unmat &   \unmat\\
						\unmat &  \fir_C & \fir_C \\
					\end{tabular}
				};    
				
				\node[pref] (coaches) [right = 2mm of aligner] {
					\begin{tabular}{C | C | C}
						\fir_A & \fir_B & \fir_C\\
						\hline  \hline 
						\wor_2 &   \wor_1 & \wor_1 \\
						\wor_3 & \wor_ 2 &  \wor_3 \\
						\wor_1  &  \unmat & \unmat  \\
						\unmat &  \wor_3 & \wor_2 \\
					\end{tabular}
				};    
			\end{tikzpicture}
			\caption{$\wor_1$, $\wor_3$, and $\fir_C$'s preferred match---$\dot\mat$}
		\end{subfigure}
		\caption{
			An example showing that the outcomes in the agreeable core cannot be ordered by $\wpref_\Wor$.
		}
	\end{figure}
	
	Despite the impossibility of recovering a complete lattice over the agreeable core as in the classic model of stability, I show that a narrower result continues to hold.
	Given that the lattice structure failed in the example because two competing cycles exist in the agreeable core, an astute reader may conjecture that the lattice structure continues to hold for workers and firms who do not lie in such cycles.
	Suggestively, say that $\age$ is a \textit{free agent} in $\mat$ if $\age$ lies on an acyclic, complete, and alternating path of  $(\Age, \rema, \mat)$.
	My first proposition justifies my terminology:
	\begin{proposition}\label{pro:nofreeblockingpairs}
		If $\mat$ is in the agreeable core, then there are no blocking pairs among free agents in $\mat$.
		Moreover, every free agent $\age$ in $\mat$ weakly prefers $\mat(\age)$ to being unmatched.
	\end{proposition}
	The proof of \Cref{pro:nofreeblockingpairs} shows that these agents are \say{free} to form blocking pairs because each can satisfy a sequence formed by alternating edges from $\rema$ and $\mat $.
	Free agents resemble the agents in the classic model: their $\rema$-partner (if any) is not concerned with the partner she finds.
	
	However, an obstacle arises because the free agents depend on $\mat$; that is, $\age$ may be a free agent in $\mat$ but not in $\nat$.
	What I can show is that, if $\mat$ and $\nat$ share the same set of free agents and they agree on the agents who are not free, then $\mat \join \nat$ is in the agreeable core.
	Toward that end, I say that $\mat$ and $\nat$ are \textit{structurally similar} if they have the same set of free agents and $\mat(\age) = \nat(\age)$ for every agent which is not free.
	The following lemma shows that structurally similar matches in the agreeable core {play nicely} with the join and meet operators defined previously:
	\begin{lemma}\label{lem:joinmeetmatch}
		Let $\mat $ and $\nat$ be structurally similar matches in the agreeable core.
		Then $\mat \join \nat $ is a match.
		The same holds for $\mat \meet \nat$.
	\end{lemma}
	Notably, $\mat \join \nat$ may not be structurally similar to $\mat $ and $\nat$.\footnote{
		I have an example demonstrating this (available upon request), but it is too lengthy to include because it involves eight workers and eight firms.
	}
	The (possible) structural differences between $\mat\join\nat$ and $\mat$ force us to discard any hope of obtaining a lattice-like result.
	However, the join and meet operators still produce matches in the agreeable core:
	\begin{theorem}\label{thm:joinmeet}
		Let $\mat$ and $\nat$ be structurally similar matches in the agreeable core.
		Then $\mat \join \nat$ and $\mat \meet \nat$ are both in the agreeable core.
	\end{theorem}
	The conflict of interest continues to hold for structurally similar matches.
	That is, if $\mat$ and $\nat$ are in the agreeable core \textit{and are structurally similar}, then if every worker weakly prefers $\mat$ to $\nat$, then every firm weakly prefers $\nat$ to $\mat$.
	Conversely, in the classic matching framework, $\rema(\age) = \age$ for every agent and thus every agent is free.
	Every match is then structurally similar and hence my \Cref{thm:joinmeet} generalizes standard results.

	\section{Incentives in the Propose-Exchange algorithm}\label{sec:manipulability}
	This section addresses the incentive properties of the Propose-Exchange algorithm.
	The results provide insight into how robust the PE is to manipulation by participants.
	This is crucial for implementing the PE in practice because the output of the PE is only guaranteed to be in the agreeable core if the inputs are accurate.
	I find that while the PE is more susceptible to more kinds of manipulations than either the DA or the TTC, the new manipulations are difficult to execute.
	
	I consider two kinds of manipulations in these subsections.
	In the first, I allow a worker to arbitrarily misreport his preference.\footnote{It is well-known that a firm can manipulate the DA by misreporting her preference, so I only consider the problem from the worker's perspective.}
	In the second, I allow a worker and a firm to create an artificial initial match, a misreport of $\rema$.
	
	For clarity through this section, I write $\wpref_\wor ' $-Propose stage to indicate the operation of the Propose stage on the matching problem when $\wor$'s preference $\wpref_\wor $ is replaced by $\wpref_\wor ' $.
	A similar shorthand is used when $\rema $ is replaced by $\rema '$.
	
	\subsection{Preference Manipulation}
	In this subsection I discuss preference manipulations by workers.
	I allow a worker $\wor$ to misreport his preference $\wpref_\wor$ by reporting $\wpref_\wor ' $ instead.
	The intuition is that a worker may benefit from manipulating which agents (including himself) are active in the Exchange stage.
	I show that there may exist a worker who can profitably misreport his preference in the PE.
	However, this problem is not unique to the PE, but exists for every algorithm that produces a match in the agreeable core.
	These results are in contrast to their parallels in existing theory: no worker can profitably misreport his preferences in the DA or TTC \citep{dubins_machiavelli_1981, dur_two-sided_2019}.
	I connect these results by showing that only workers who participate in both stages of the PE can profitably misreport their preferences.
	
	Formally, \textit{mechanism} $\mech$ is a function of $(\Wor, \Fir, \wpref, \rema)$ that returns a match $\mech(\Wor, \Fir, \wpref, \rema)$.
	\begin{definition}
		A mechanism is \textit{preference manipulable} if there is at least one matching problem $(\Wor, \Fir, \wpref, \rema)$, worker $\wor$, and preference $\wpref_\wor ' $ such that
		\begin{align*}
			 \mech(\Wor, \Fir, \wpref_{-\wor}, \wpref_{\wor} ' , \rema) \wpref_\wor  \mech(\Wor, \Fir, \wpref, \rema).
		\end{align*}
		In words, if for some example a worker $\wor$ would rather report $\wpref_\wor ' $ instead of $\wpref_\wor$, then $\mech$ is preference manipulable.
	\end{definition}
	A natural question arises as to whether a mechanism exists that is non-preference-manipulable and produces a match in the agreeable core.
	\Cref{pro:man} provides a negative answer:
	\begin{proposition}\label{pro:man}
		If $\mech(\wpref)$ is in the agreeable core for all $\wpref$, then $\mech$ is preference manipulable.
	\end{proposition}
	I prove \Cref{pro:man} through a counterexample.
	The counterexample is driven my the possibility of bossiness within the DA.
	A mechanism is \textit{bossy} if an agent can, by misreporting his preference, affect the matches of the other agents without changing his own.
	Consider the example in \Cref{fig:bossy}.
	Worker $\wor_3$ can cause workers $\wor_1$ and $\wor_2$ to exchange partners by misreporting a preference for firm $A$.
	In the counterexample in the proof of \Cref{pro:man}, there is a worker $\wor_1$ who would like to exchange initial partners with $\wor_2$.
	Worker $\wor_1$ reduces $\wor_2$'s ability to match to an initially unmatched firm by including that firm in his own preferences.
	Effectively, if $\wor_2$ is a free agent then $\wor_1$ will not be able to match to $\rema(\wor_1)$.
	Thus, $\wor_1$ manipulates $\wor_2$'s options to keep $\wor_2$ matched to $\rema(\wor_2)$ to cause an exchange.
	
	\begin{figure} 
		\centering
		\begin{subfigure}[b]{0.45\textwidth}
			\centering
			\begin{tikzpicture}
				\def\vshift{-0.575};
				\def\lshift{2.0};
				\def\rshift{1.88};
				\def\colwid{0.93};

				\def\firstrow{0.99+\vshift};
				\def\secondrow{-0.16+\vshift};
				\def\thirdrow{-1.31+\vshift};
				\def\mlw{0.5mm};


				\draw[line width = \mlw, dashed, dash pattern=on 8pt off 3pt] (-\colwid-\lshift,\firstrow) ellipse (0.3 and 0.3); 
				\draw[line width = \mlw, dashed, dash pattern=on 8pt off 3pt] (-\lshift,\firstrow) ellipse (0.3 and 0.3); 
				\draw[line width = \mlw, dashed, dash pattern=on 8pt off 3pt] (+\colwid-\lshift,\firstrow) ellipse (0.3 and 0.3); 
				\draw[line width = \mlw, dashed, dash pattern=on 8pt off 3pt] (-\colwid+\rshift,\thirdrow) ellipse(0.3 and 0.3); 
				\draw[line width = \mlw, dashed, dash pattern=on 8pt off 3pt] (0+\rshift,\secondrow) ellipse (0.3 and 0.3); 
				\draw[line width = \mlw, dashed, dash pattern=on 8pt off 3pt] (+\colwid+\rshift,\firstrow) ellipse (0.3 and 0.3); 
				
				\node (aligner) {};
				\node[pref] (players) [left = 2mm of aligner] {
					\begin{tabular}{C | C | C}
						\wor_1 & \wor_2 & \wor_3  \\
						\hline  \hline 
						\fir_A & \fir_B & \fir_C \\
						\fir_B & \fir_A& \fir_A \\
						\fir_C &  \fir_C & \fir_B \\
					\end{tabular}
				};    
				
				\node[pref] (coaches) [right = 2mm of aligner] {
					\begin{tabular}{C | C | C}
						\fir_A & \fir_B & \fir_C\\
						\hline  \hline 
						\wor_2 &   \wor_1 & \wor_3 \\
						\wor_3 &  \wor_2 &  \wor_2\\
						\wor_1  & \wor_ 3 & \wor_1  \\
					\end{tabular}
				};
			\end{tikzpicture}
			\caption{Outcome of DA before preference swap}
			\label{fig:bossy1}
		\end{subfigure}
		\hfill
		\begin{subfigure}[b]{0.45\textwidth}
			\centering
			\begin{tikzpicture}
				\def\vshift{-0.575};
				\def\lshift{2.0};
				\def\rshift{1.88};
				\def\colwid{0.93};

				\def\firstrow{0.99+\vshift};
				\def\secondrow{-0.16+\vshift};
				\def\thirdrow{-1.31+\vshift};
				\def\mlw{0.5mm};
				
				\draw[line width = \mlw, dashed, dash pattern=on 8pt off 3pt] (-\colwid-\lshift,\secondrow) ellipse (0.3 and 0.3); 
				\draw[line width = \mlw, dashed, dash pattern=on 8pt off 3pt] (-\lshift,\secondrow) ellipse (0.3 and 0.3); 
				\draw[line width = \mlw, dashed, dash pattern=on 8pt off 3pt] (+\colwid-\lshift,\secondrow) ellipse (0.3 and 0.3); 
				\draw[line width = \mlw, dashed, dash pattern=on 8pt off 3pt] (-\colwid+\rshift,\firstrow) ellipse(0.3 and 0.3); 
				\draw[line width = \mlw, dashed, dash pattern=on 8pt off 3pt] (0+\rshift,\firstrow) ellipse (0.3 and 0.3); 
				\draw[line width = \mlw, dashed, dash pattern=on 8pt off 3pt] (+\colwid+\rshift,\firstrow) ellipse (0.3 and 0.3); 
				
				\node (aligner) {};
				\node[pref] (players) [left = 2mm of aligner] {
					\begin{tabular}{C | C | C}
						\wor_1 &\wor_ 2 & \wor_3  \\
						\hline  \hline 
						\fir_A & \fir_B & \fir_A \\
						\fir_B & \fir_A& \fir_C \\
						\fir_C &  \fir_C & \fir_B \\
					\end{tabular}
				};    
				
				\node[pref] (coaches) [right = 2mm of aligner] {
					\begin{tabular}{C | C | C}
						A & B & C\\
						\hline  \hline 
						\wor_2 &   \wor_1 & \wor_3 \\
						\wor_3 &  \wor_2 &  \wor_2\\
						\wor_1  &  \wor_3 & \wor_1  \\
					\end{tabular}
				};
			
				\draw[<->, line width= \mlw, lightgray] 
				(+\colwid-\lshift+0.18,\secondrow)
				to[out = 10, in = -10]
				(+\colwid-\lshift+0.18,\firstrow);
			\end{tikzpicture}
			\caption{Outcome of DA before preference swap}
			\label{fig:bossy2}
		\end{subfigure}
		\caption{Worker $\wor_3$'s exchange of the order of $\fir_C$ and $\fir_A$ in his preference leaves his partner unchanged, but causes workers $\wor_1$ and $\wor_2$ to receive new partners.}
		\label{fig:bossy}
	\end{figure}
	
	\Cref{thm:nomanPE} formalizes this intuition.
	It shows that a worker only has two avenues through which to profit from a misreport.
	First, the worker may profit from finding a partner in the Exchange stage rather than the Propose stage.
	This is similar to truncating\footnote{moving his initial partner higher; see \cite{roth_truncation_1999}} his preferences.
	Second, the worker may find his partner in the Exchange stage but choose to manipulate which workers who participate in the Exchange stage, as in the counterexample previously discussed.
	\begin{theorem}\label{thm:nomanPE}
		If worker $\wor$ has a profitable misreport $\wpref_\wor ' $, then $\wor$ is active in both stages of the $\wpref_\wor ' $-Propose-Exchange algorithm.
	\end{theorem}
	Because whether a worker is active in the Propose stage is independent of his reported preferences, \Cref{thm:nomanPE} further restricts the set of workers who can profitably misreport.
	A worker can only profitably misreport if he both has a $\rema$-firm and is active in the Propose stage.
	For a market designer, these conditions are easy to verify and provide an upper bound on the number of workers who can profitably manipulate.
	Additionally, \Cref{thm:nomanPE} highlights the informational requirements necessary to profitably misreport.
	A worker must be able to predict the outcome of the Exchange stage, which itself is a complicated object and depends on which agents participate in the Exchange stage.

	\subsection{Manipulating $\rema$}
	In this subsection I complement the analysis of how preferences may be profitably misreported with an analysis of how the initial match may be profitably misreported.
	The concern is that because the initial match $\rema$ affects the output of the PE, a pair of agents may find it in their interest to create a superfluous artificial agreement.
	I show that, while such a manipulation is possible, it usually requires an additional preference manipulation to be successful.
	I conclude that profitably misreporting the initial match requires a similar level of sophistication as a preference manipulation.
	
	Formally, let $\rema$ be given (and fixed throughout this subsection) with $\matExch$ the output of the $\rema$-PE.
	Let worker $\wor$ and firm $\fir$ be both $\rema$-unmatched, and let $\rema ' $ be formed from $\rema$ by matching $\wor$ and $\fir$.
	Let $\matProp ' $ and $\matExch ' $ be the respective outputs of the $\rema ' $-Propose and $\rema ' $-Exchange stages.
	If both $\wor $ and $\fir$ strictly prefer $\matExch ' $ to $\matProp ' $, then $\wor$ and $\fir $ can \textit{profitably misreport} an initial match.
	Profitably misreporting an initial match requires that both $\wor$ and $\fir$ strictly gain from the deviation.
	
	I show that, although it is possible for the PE to be manipulated in this way, its extent is quite limited and involves substantial risk for the worker.
	First, I show in \Cref{thm:art} that any profitable misreport pushes $\wor$ and $\fir$ from the Propose stage into the Exchange stage ($\matProp ' (\wor) = \fir$).
	The intuition is that if $\matProp ' (\fir) \neq \wor$, then $\fir$ has received a better partner in the $\rema ' $-Propose stage and thus that all of the workers have received a worse partner.
	Second, \Cref{thm:art} also shows that for any profitable misreport, $\wor$ \textit{cannot} be active in the $\rema ' $-Propose stage.
	\begin{theorem}\label{thm:art}
		If $\wor$ and $\fir$ can profitably misreport an initial match, then $\matProp ' (\wor) = \fir$ and $\wor$ is not active in the $\rema ' $-Propose stage.
	\end{theorem}
	The interpretation of \Cref{thm:art} is that $\fir$ must prefer $\wor$ to $\fir$'s match when $\wor$ is removed from the matching problem entirely.
	In effect, $\fir$ faces little risk from the misreporting because $\wor$ is as good as (if not better than) what $\fir$ would receive if $\wor$ were not present.
	For $\wor$ however, an initial match with $\fir$ could carry great risk if $\fir$ is low on $\wor$'s preferences relative to $\matProp(\wor)$.
	This strategy may backfire because a mistake in $\wor$'s calculations (or a misrepresentation by $\fir$) could render $\wor$ assigned to $\fir$.
	
	In summary, neither misreporting preferences or the initial match appears likely to succeed without detailed knowledge of the other participants' preferences.
	Misreports frequently expose misreporting agents to a large downside risk.
	These incentive findings inform the broader applicability of the PE, which I discuss in the following section.

	\section{Conclusion}\label{sec:discussion}
	This paper has shown the strength of the agreeable core in providing a theory of equilibrium for a broad class of matching markets.
	The initial match organically models numerous real-world examples, and the Propose-Exchange algorithm is ready to be implemented in a variety of applications.
	In this closing section I discuss three topics.
	First, I close with a discussion of my modeling choices and possible extensions.
	Second, I provide guidance on applying the Propose-Exchange algorithm in several environments.
	Third, I review the connections between this paper and existing research.
	
	\subsection{Future Directions}
	The many-to-one setting introduces complex constraints because firms participating in an agreeable coalition must consider multiple binding agreements.
	The motivation behind my focus on one-to-one matching is driven by two competing models of a firm in many-to-one markets.
	In the first model, each firm is modeled as a collection of unit-demand sub-firms, each endowed with the master firm's preference over individual workers.
	This is the model used in most applications because eliciting a single ranking over workers from each firm is easier than a preference over sets of workers.
	The agreeable core then treats each sub-firm as an individual agent.
	A worker is initially matched to a single sub-firm, and he must include \textit{that} sub-firm in any agreeable coalition.
	This model straightforwardly extends the one-to-one theory, and the same results hold.\footnote{
		There are some interesting additional questions in this environment, such as how a worker should construct his preference over two identical sub-firms which are initially matched to different workers, and whether a firm could rearrange the initial matches of its sub-firms to construct a new agreeable and blocking coalition?
	}
	In the second model, each firm is treated as an agent with a preference over \textit{sets} of workers.
	Even in the classic model without an initial match, restrictions such as substitutability need to be placed on firm preferences to guarantee existence.\footnote{See \cite{echenique_theory_2004} for a unified treatment of the many-to-many case.}
	Beyond the question of existence, the requirement that $\nat(\Coal) = \Coal$ for a coalition $\Coal$ blocking with match $\nat$ implies that the size of an agreeable blocking coalition increases dramatically.
	For example, if a firm seeks to join a coalition, that coalition must include all of its initial workers (who themselves are possibly matched to other firms) and all of the workers it will match to (who themselves may be initially matched to other firms).
	In a market with many workers initially matched, agreeable coalitions quickly must contain almost every agent in the model.
	The usefulness of the agreeable core in this context is unclear, and adapting it to these environments is a future avenue of research.
	
	The agreeable core can provide insights into the formation of the initial match $\rema$.
	The model is agnostic as to how $\rema$ is determined.
	It could be interesting to use the agreeable core or the Propose-Exchange algorithm in combination with a model of the formation of $\rema$ to understand pre-matching dynamics.
	Because the initial match is instrumental in the PE, developing a theory of pre-match formation could be insightful for other market-design applications.
	\Cref{thm:art} addresses one such question, but more questions abound.
	
	\subsection{Applications}
	The PE can unify out-of-match residencies with the NRMP, creating a larger overarching match that nests both and guarantees Pareto efficiency while allowing for early matches.
	It is well-known that a fraction of medical residencies are offered independently of the centralized clearinghouse operated by the NRMP.
	These out-of-match residency programs entice prospective residents to sign binding contracts prior to the operation of the NRMP because these contracts provide guarantees to risk-averse residents.
	Because the rules of the NRMP forbid residents from participating if they have already accepted an out-of-match offer, these two markets operate independently.\footnote{Recently, the NRMP has implemented the \say{All-In} policy in an attempt to curtail residency programs from offering out-of-match residencies. The All-In policy requires that any residency program participating in the NRMP offer residencies exclusively through the NRMP.}
	The out-of-match offers introduce inefficiency by dividing the market temporally.
	Under the PE, the out-of-match market operates essentially unchanged: programs can entice residents with early offers.
	However, if the NRMP uses the Propose-Exchange algorithm, the residents and programs who have already formed contracts are allowed to participate as agents under an initial match.
	\Cref{rem:efficiency} guarantees that the final match is Pareto efficient.
	A similar construction can be used to integrate Early Decision agreements into the regular college admission cycle.
	
	The PE also allows for asymmetrical obligations, such as professional sports contracts or tenured positions, which bind participants unequally.
	For example, an athlete's contract with a team may allow the team to trade the athlete to another without the athlete's consent, but the athlete cannot \say{trade} his team without the team's consent.
	Similarly tenured professor or teacher's contract allows her leave her institution unilaterally, restricts the institutions ability to remove her; see \cite{combe_design_2022} for an application to the French public school system.
	To incorporate this one-way obligation into the PE, I modify the participants' preferences.
	For the professor $\wor$ tenured at (that is, initially matched to) institution $\fir$, I modify $\fir$'s preference $\wpref_\fir$ by moving $\wor$ to the bottom of $\wpref_\fir$.
	This guarantees that $\wor$ is \textit{never} required to remain at $\fir$, but always may choose to do so.
	Without an initial match, the standard model is instead forced to move $\wor$ to the top of $\wpref_\fir$; this achieves the same result ($\wor$ can always match to $\fir$), but suffers from inefficiency \citep{pereyra_dynamic_2013}.
	The one-way contracts that allow for trades, as in professional sports, can similarly be included under additional assumptions.\footnote{
		For instance, a \say{tradable} contract can be included through modifying the athlete's preference by putting the team and being unmatched at the bottom of the preference.
		Therefore, the athlete is always matched to a team, but the identity of the team can change.
		However, there is a tension: if the athlete can express a preference for being unmatched, then the team can terminate the athlete at will.
		Hence, in this model it is essential that the athlete can only be traded to a set of teams which he prefers to being unmatched.\\
		Again, there is a limit to who can have tradable contracts.
		If a team is allowed to trade an athlete, then the PE algorithm must have the teams propose and point.
		This precludes any athlete from trading her team.
		In professional sports this is a reasonable assumption, but caution is needed in more general applications.
}
	
	The initial match can also be leveraged to achieve minimum quotas that balance individual preferences and institutional needs.
	Examples of minimum quotas are minimum enrollment at a school or in a class, or guarantees that some \say{rural} hospitals are matched to residents.
	For instance, a minimum quota of students may be required for a school to operate or for a class to be offered.
	The PE can incorporate these quotas by using the initial match to assign the minimum number of students to the school or class.
	By then modifying the school's or class's priority order (preferences) over students by moving the initially assigned students to the bottom, just above being unmatched, the designer guarantees that the school or class will enroll at least its minimum quota.
	The initial assignments are only binding if no other student desires the school or class.
	In this way, the initial match requires the minimal restriction on students' choices while meeting the institutional objective.
	The agreeable core provides a clear justification for why some students' choices are restricted.
	If a restricted student would like to attend another school, then at least one school would not meet its minimum quota or some student would be harmed.

	\subsection{Connection to the Literature}\label{sec:literaturereview}
	This paper develops a novel theory of matching under initial contracts that bridges object allocation and two-sided matching.
	It connects several literatures on two-sided matching.
	An exhaustive review of the literature is far beyond the scope of this paper, so I list the only the most closely related work and its connections with this paper.
	
	I integrate the classic model of two-sided matching with recent advances in recontracting.
	In the classic model, a stable match always exists and can be found by the DA \citep{gale_college_1962}.
	It is well known that the set of pairwise-stable matches corresponds to the core of a related cooperative game \citep{Roth_Sotomayor_1990}.
	Later research largely discarded the connection with the core in favor of pairwise-stability notions.
	When considering matching with an initial match (in which the intersection of pairwise stable and individually rational outcomes may be empty), \cite{pereyra_dynamic_2013} and \cite{guillen_matching_2012} generalize pairwise-stability by partitioning claims between valid and invalid claims and then removing all valid claims.
	This may be strongly inefficient \citep{combe_reallocation_2024, combe_design_2022}, and hence a mechanism with minimal envy is considered \citep{kwon_justied-envy-minimal_2023}.
	Although efficient, these minimal envy mechanisms are inscrutable to participants: the designer allows some claims but not others only because doing so minimizes some objective.
	My  paper advances this literature by reconnecting the initial back to the core, a more interpretable solution.
	I both minimize envy as in \cite{kwon_justied-envy-minimal_2023} but also provide a clear definition of valid and invalid claims as in \cite{pereyra_dynamic_2013}.
	
	Research in school choice has made extensive use of both the DA and TTC.
	\cite{abdulkadiroglu_school_2003} suggests the Deferred Acceptance algorithm from \cite{gale_college_1962} or the Top Trading Cycles algorithm from \cite{shapley_cores_1974} as desirable and implementable solutions.
	Both algorithms run in polynomial time, are relatively easy to describe, and are strategyproof.
	The DA is fair (no blocking pairs) while the TTC is efficient (Pareto efficient for the students).
	A plethora of researchers seek to combine or modify the two algorithms to reconcile these properties, allowing certain priority violations. \citep{abdulkadirog_generalized_2011,dur_school_2019,kesten_two_2006, kwon_justied-envy-minimal_2023, reny_efficient_2022, troyan_essentially_2020, morrill_making_2013, dur_impossibility_2017}.
	Papers in this vein typically define a set of properties of a mechanism (such as the allowable priority violations, efficiency, strategyproofness, etc.), and then present a satisfactory algorithm, typically a variation of the DA or TTC.
	My work complements this approach by an algorithm derived from first principles rather than with specific objectives in mind.
	My approach draws from cooperative game theory rather than emphasizing certain desirable properties of the final allocation.
	
	A connected branch of matching theory develops methods for matching with minimum quotas.
	Schools are modeled as having both a maximum capacity for students but also a minimum required quota of students.
	One approach is to allow for wasted seats but not envy \citep{fragiadakis_improving_2017}.
	A separate approach uses an auxiliary \say{master list} \citep{ueda_strategy-proof_2012} or \say{precedence list} \citep{fragiadakis_strategyproof_2016, hamada_strategy-proof_2017} as a means to break ties: if two students wish to take an empty seat but the minimum quota requires that only one may do so, the list determines which worker can.
	The algorithms described in both approaches typically either sacrifice efficiency (based on the DA) or fairness (based on the TTC), and both require that all agents are mutually acceptable.
	I develop both approaches by endogenizing the master list into the initial match and not requiring any assumptions on preferences.
	Although a master list is natural in some applications, whether a master list or the initial match is more appropriate depends on the application.
	
	Surprisingly, no authors have connected matching with minimum quotas and the matching with an initial match.
	I combine these subfields with the observation that, if the initial match provides a guarantee for both workers \textit{and firms}, then minimum quotas are the special case when every firm is assigned workers equal to its minimum quota in the initial match.
	The initial match provides a different justification for why some blocking pairs are allowable but others are not, one which I think applies well to school choice.
	
	Finally, the paper closest in spirit to ours is \cite{abdulkadiroglu_house_1999}, \say{House Allocation with Existing Tenants.}
	Their model is one-sided, and they show that a hybrid of the Serial Dictatorship algorithm and the TTC algorithm provides an efficient improvement over the initial match.
	I present a two-sided model with a hybrid algorithm between the DA and the TTC.
	Although my models are different, my approach is remarkably similar to theirs.

	\bibliography{bibliography}

	\appendix
	\setcounter{lemma}{0}
	\renewcommand{\thelemma}{\Alph{section}.\arabic{lemma}}
	\setcounter{definition}{0}
	\renewcommand{\thedefinition}{\Alph{section}.\arabic{definition}}
	\section{Omitted Proofs}\label{appendix}
	
%
	Throughout the appendix I abuse notation and write $\age \in \edg$ to mean that either the first or second coordinate of $\edg$ is $\age$.
	\begin{lemma}\label{lem:altClean}
		If $\Path$ is a complete and alternating path in $(\Age, \rema \cup \mat \cup \Imp(\mat))$, then every agent contained in $\Path$ is in exactly one edge from $\rema$ and one edge from $\mat \cup \Imp(\mat)$.
	\end{lemma}
	\begin{proof}
		Let $\Path = (\edg_1 , \ldots, \edg_n )$ be a complete and alternating path in $(\Age, \rema \cup \mat \cup \Imp(\mat))$ and let $\age$ be contained in $\Path$.
		If $n = 2$, then the statement is trivial because completeness implies every $\age \in \Path$ is in both $\edg_1$ and $\edg_2$ and $\Path$ alternating implies that one of $\{\edg_1, \edg_2\}$ is in $\rema$ and the other is in $\mat \cup \Imp(\mat)$.
		Hence, let $n \geq 3$.
		
		Again, if $\age \in \edg_k \cap \edg_{k+1}$ for $k \geq 1$ then the statement is true because completeness implies $\edg_k $ and $\edg_{k+1}$ are the only edges in $\Path$ containing $\age$, both $\edg_k$ and $\edg_{k+1}$ cannot be from $\rema$ by construction, and $\Path$ alternating implies that both $\edg_k$, and $\edg_{k+1}$ cannot be from $\mat \cup \Imp(\mat)$.
		Therefore, one of $\{\edg_k, \edg_{k+1}\}$ is from $\rema $ and the other from  $\mat \cup \Imp(\mat)$.
		Hence, let $\age \in \edg_1 \cap \edg_n$ and thus $\Path$ is cyclic.
		Let $\age$ be a worker; the argument is symmetric if $\age$ is a firm.
		
		Because there is a bijection\footnote{namely, $\rema$} between the workers and firms contained in $\Path$ and every agent in $\Path$ is contained in two edges of $\Path$, $n$ is even.
		Therefore, if $\edg_1 \in \rema$ then $\edg_n \in \mat \cup \Imp(\mat)$, and if $\edg_1 \in \mat \cup \Imp(\mat)$ then $\edg_n \in  \rema$.
		This proves the result.
	\end{proof}
	
	\begin{proof}[\unskip\nopunct]
		\textbf{\textit{Proof of \Cref{pro:ac-iff-nbps}:}}
		Let $\mat$ be individually rational.
		
		\textit{For the $(\Rightarrow)$ direction:}
		I prove the contrapositive; that is, if $\mat$ admits a blocking path, then $\mat$ is not in the agreeable core.
		Let $\Path = (\edg_1, \ldots , \edg_n)$ be a blocking path in $(\Age, \rema \cup \mat \cup \Imp(\mat))$.
		Note that $\rema(\Path) = \Path$ and $\mat^\Path(\Path) = \Path$.
		
		By the definition of $\Imp(\mat)$, it follows that $\mat^\Path \wpref_\Path \mat$.
		Because $\Path$ is blocking, there is an edge $\edg$ in $\Path$ that is also in $\Imp(\mat)$.
		Hence, both agents in $\edg$ strictly prefer $\mat^\Path $ to $\mat$.
		Therefore, $\Path$ is an agreeable blocking coalition and $\mat$ is not in the agreeable core.
		
		\textit{For the $(\Leftarrow)$ direction:}
		I prove the contrapositive; that is, if $\mat$ is not in the agreeable core then $\mat$ admits a blocking path.
		Let $\mat$ be not in the agreeable core.
		Then there exists an agreeable blocking coalition $\Coal$ that blocks $\mat$ through $\nat$.
		
		Let $\age_1$ be an agent in $\Coal$ such that $\nat (\age_1) \spref_{\age_1} \mat(\age_1)$; such an agent exists by the definition of a blocking coalition.
		I will construct a path $\Path$ from $\age_1$ by iteratively adding alternating edges from $\rema$ and $\nat $ to $\{\age_1, \nat(\age_1) \}$, first with increasing indices and then with decreasing indices.
		I assume that $\age_1 \in \Wor$; the other case follows from a symmetric argument.
		
		Starting with $\edg_1 \equiv (\age_1, \nat(\age_1))$ and $\Path_1 \equiv (\edg_1)$, do the following iteratively.
		Choose an edge $\edg_{k+1}$ from $\rema$ or $\nat $ that is not already present in $\Path_{k}$ such that the second coordinate of $\edg_k$ is the first coordinate of $ \edg_{k+1}$, then define $\Path_{k+1}$ by appending $\edg_{k+1}$ to $\Path_k$.
		Continue until no more edges may be added in this way.
		Finally, repeat the same process starting from $\edg_1$, but \textit{prepending} edges $\edg_0, \edg_{-1}, \ldots$ to $\Path_{k}$.
		
		Observe that $\Path$ is a path in $(\Age, \rema \cup \mat \cup \Imp(\mat))$ because $\nat \wpref_\Coal \mat$.
		Next, observe that because every agent in $\Path$ is contained in at most two edges (one from $\rema$ and the other form $\nat$); every agent in $\Path$ is contained in at least two edges because edges are added until no more can be added without including repeats and therefore $\Path$ is complete.
		Also, $\Path$ is alternating because $\edg_{2k} \in \rema$ and $\edg_{2k-1} \in \nat $.
		Finally, observe that $\edg_1 \in \Imp(\mat)$.
		Therefore, $\Path$ is a blocking path of $\mat$.
		Therefore $(\Age, \rema \cup \mat \cup \Imp(\mat))$ contains a blocking path, completing the proof.
	\end{proof}
	
	\textit{\textbf{Introduction to the proofs of \Cref{lem:propose} and \Cref{pro:ord}:}}
	
	Before proving \Cref{lem:propose}, I first introduce some notation and a short result:
	\begin{definition}
		I say that loop $\edg = (\age, \age)$ is a \textit{proposal source} if either 
		\begin{enumerate}
			\item[1(a)]: $(\age, \age) \in \rema$ and $\age \in \Wor$, or
			\item[1(b)]: $(\age, \age)\notin \rema$ and $\age \in \Fir$.
		\end{enumerate}
		I say that loop $\edg = (\age, \age)$ is a \textit{proposal sink} if $\edg$ in not a proposal source; that is, if either 
		\begin{enumerate}
			\item[2(a)]: $(\age, \age) \notin \rema$ and $\age\in \Wor$ or
			\item[2(b)]: $(\age, \age)\in \rema$ and $\age\in \Fir$.
		\end{enumerate}
	\end{definition}
	A straightforward parity argument shows that if $\Path = (\edg_1, \ldots, \edg_n)$ is a complete, alternating, and acyclic path in $(\Age, \rema \cup \mat \cup \Imp(\mat))$, then $\edg_1$ is a proposal source and $\edg_n$ is a proposal sink.
	\begin{lemma}\label{lem:SS}
		Let $\Path = (\edg_1, \ldots, \edg_n)$ be a complete, alternating, and acyclic path in $(\Age, \rema \cup \mat \cup \Imp(\mat))$ with $n \geq 3$.
		Then $\edg_1$ is a proposal source and $\edg_n$ is a proposal sink.
	\end{lemma}
	\begin{proof}
		Because $\Path$ is acyclic and complete, $\edg_1$ and $\edg_n$ are both loops.
		Let $\edg_1 = (\age_1, \age_1)$ and $\edg_n = (\age_{n-1}, \age_{n-1})$.
		Similarly, let $\edg_2 = (\age_1, \age_2)$ and $\edg_{n-1} = (\age_{n-2}, \age_{n-1})$.
		Because $n \geq 3$, $\age_1 \neq \age_2$ and $\age_{n-2} \neq \age_{n-1}$.
		
		Consider the following cases:
		\begin{enumerate}
			\item $\age_1 \in \Wor$:
			Then because there are no edges between two distinct workers, it follows that $
			\age_2 \in \Fir$.
			Therefore, $\edg_2 \in \mat \cup \Imp(\mat)$.
			This implies that $\edg_1 \in \rema$.
			Therefore $\edg_1$ is a proposal source. 
			
			\item $\age_1 \in \Fir$:
			Then because there are no edges between two distinct workers, it follows that $
			\age_2 \in \Wor$.
			Therefore, $\edg_2 \in \rema$.
			This implies that $\edg_1 \in \mat\cup \Imp(\mat)$.
			Therefore $\edg_1$ is a proposal source. 
		\end{enumerate}
		Symmetric arguments show that $\edg_n$ is a proposal sink.

	\end{proof}
	
	\begin{proof}[\unskip\nopunct]
		\textbf{\textit{Proof of \Cref{lem:propose}:}}
		
		Suppose (toward a contradiction) that $\Path = (\edg_1, \ldots, \edg_n)$ is an acyclic blocking path of $\matProp$.
		Because $\Path$ is acyclic and complete, $\edg_1$ and $\edg_n$ are both loops and $n \geq 3$.
		By \Cref{lem:SS}, $\edg_1$ is a proposal source and $\edg_n$ is a proposal sink.
		Let
		\begin{align*}
			\edg_1 &= (\age_1) \\
			\edg_2 &= (\age_1, \age_2) \\ 
			& \:\;\vdots \\
			\edg_{n-1} &= (\age_{n-2}, \age_{n-1})\\
			\edg_{n} &= (\age_{n-1})
		\end{align*}
		
		I argue by induction that every worker $\age_k\in \Path$ makes a proposal during the \Propose algorithm.
		Because every agent contained in $\Path$ weakly prefers $\mat^\Path $ to $\matProp$, it follows that every worker contained in $\Path$ \textit{who proposes} proposed to his $\mat^\Path $-partner.
		In my base case I show that the worker with the lowest index contained in $\Path$ proposes during the \Propose algorithm.
		There are two possibilities:
		\begin{enumerate}
			\item \textit{$\age_1$ is a worker:}
			Because $\edg_1$ is a proposal source by definition $\rema(\age_1) = \age_1$.
			Hence $\age_1$ begins the \Propose algorithm activated.
			Therefore, $\age_1$ proposes during the \Propose algorithm.
			
			\item \textit{$\age_1$ is a firm:}
			Because $\edg_1$ is a proposal source, by definition $\rema(\age_1) \neq \age_1$.
			Therefore $ \rema(\age_1) = \age_2$.
			Because $\age_1$ prefers $\mat^\Path$ to $\rema$ and $\mat^\Path(\age_1) = \age_1$ because $\edg_1$ is loop, it follows that $\age_2$ is activated at the start of the \Propose algorithm.
			Therefore, $\age_2$ proposes during the \Propose algorithm.
		\end{enumerate}
		
		For the inductive step, suppose $\age_{k-1} \in \Wor$ makes a proposal; I will show that the worker with the next highest index makes a proposal.
		If $k-1 \geq n-2$, then $\age_{k-1}$ is the worker with the highest index and the claim is vacuous; therefore, suppose $k-1 < n-2$.
		Because $\mat^\Path(\age_{k-1}) = \age_k$, it follows that $\age_{k-1}$ proposes at some point to $\age_{k}$.
		Because $\matProp$ is individually rational and $\rema(\age_k) = \age_{k+1}$, it follows that $\age_k$ weakly prefers $\age_{k-1}$ to $\age_{k+1}$.
		Therefore $\age_{k+1}$ is activated at some point and thus $\age_{k+1}$ makes at least one proposal during the \Propose algorithm, concluding my inductive argument.
		
		Next, I show that an agent contained in a proposal sink never rejects a proposal from their $\mat^\Path$-partner.
		If $\age_{n-1}$ is a worker, then he never rejects a proposal from himself.
		If $\age_{n-1}$ is a firm, then $\rema(\age_{n-1}) = \age_{n-1}$ by definition.
		Because $\age_{n-1}$ prefers $\mat^\Path$ to both $\rema$ and $\matProp$ and because $\age_{n-1}$ receives no proposals she prefers to $\matProp(\age_{n-1})$ (by construction of $\matProp$), it follows that $\age_{n-1}$ does not reject a proposal from $\mat^\Path(\age_{n-1})$.
		
		Finally, I show that no worker contained in $\Path$ is rejected by his $\mat^\Path$-partner.
		To see this, suppose (toward a contradiction) that $k-1$ is the largest index such that $\age_{k-1}$ is rejected by $\mat^\Path(\age_{k-1})$.
		Because a proposal sink does not reject a proposal by his or her $\mat^\Path$-partner, it follows that $k-1 <  n-2$ (that is, $\age_{k-1}$ is not one of the last two agents in the path).
		
		Because $\age_{k}$ prefers $\age_{k-1}$ to $\matProp(\age_{k})$ and yet $\age_{k}$ rejects $\age_{k-1}$, it must be that $\rema(\age_{k}) = \matProp(\age_{k})$ (by construction of $\matProp$).
		Therefore $\age_{k}$ is matched to $\age_{k+1}$ by both $\rema$ \textit{and} $\matProp$.
		Because matches are bijective, I have $\matProp (\age_{k+1}) = \rema(\age_{k+1}) =\age_{k}$.
		Consider that, because $\Path$ is a complete and $n \geq 3$, it follows that $\mat^\Path(\age_{k+1}) \neq \matProp(\age_{k+1})$.
		Therefore $\age_{k+1}$ must be rejected by $\mat^\Path (\age_{k+1})$, a contradiction to my supposition that ${k-1}$ is the largest index for which a worker is rejected by his $\mat^\Path $-match.
		
		Therefore, because no worker in $\Path$ is rejected by his $\mat^\Path$-partner, it follows that $\mat^\Path$ agrees with $\matProp$ on $\Path$.
		Hence, every edge in $\Path$ from $\matProp \cup \Imp(\matProp)$ is from $\matProp$.
		But because $\Path $ is a blocking path, it must contain an edge from $\Imp(\matProp)$.
		Because $\mat \cap \Imp(\matProp) = \emptyset$, this is a contradiction.
		Therefore no blocking path of $\matProp$ is acyclic.
	\end{proof}
	
	\begin{proof}[\unskip\nopunct]
		\textbf{\textit{Proof of \Cref{pro:ord}:}}
		I say that a \textit{proposal order} is a function that, at every step of the Propose phase, indicates which worker makes the next proposal.
		Let $T$ and $T '$ be two proposal orders, and let the output of the Propose stage using order $T$ be $\mat$ and using $T ' $ be $\mat ' $.
		Suppose (toward a contradiction) that $\mat \neq \mat ' $.
		Let 
		\begin{align*}
			U &= \{\wor \in \Wor \spbr \mat(\wor) \neq \mat ' (\wor)\} \\
			V &= \{\wor \in \Wor \spbr \wor \t{ proposes under both $T$ and $T ' $} \}.
		\end{align*}
		There are two cases:
		\begin{enumerate}
			\item $U \cap V \neq \emptyset$:
			WLOG, there is some worker in $U \cap V$ who strictly prefers $\mat ' $ to $\mat$.
			Let $\wor$ be the first such worker who is rejected by $\fir ' \equiv \mat ' (\wor)$ in the Propose stage under $T$.
			Because $\fir ' $ is $\wor$'s $\mat ' $-partner, this implies that $\fir ' $ prefers $\wor$ to being unmatched.
			Therefore, $\fir ' $ must reject $\wor$ in favor of some $\wor^*$.
			Because $\fir ' $ is $\wor$'s $\mat ' $-partner, this implies that $\wor^* $ does not propose to $\fir ' $ under $T ' $.
			
			Because $\wor^*$ makes a proposal under $T$, it follows that there is some sequence of workers $\wor_1, \ldots, \wor_{n-1}, \wor^* \equiv \wor_n $ such that $\wor_1$ or $\rema(\wor_1)$ is a proposal source, and each $\wor_k$ makes the first proposal to $\rema(\wor_{k+1})$ under $T$.
			Let $k$ be the greatest index such that $\wor_k \in V$.
			Then it follows that $\wor_k$ strictly prefers $\mat ' $ to $\mat$.
			Therefore, $\wor_k$ must be rejected by $\mat ' (\wor_k)$ earlier than $\wor $ is rejected by $\fir ' $ under $T$, a contradiction.
			
			\item $U \cap V = \emptyset$:
			Observe that $U$ is nonempty by supposition.
			Let $\wor^* \in U$ and WLOG let $\wor^*$ strictly prefer $\mat$ to $\mat ' $.
			Thus, $\wor^*$ must make a proposal under $T$.
			It follows that there is some sequence of workers $\wor_1, \ldots, \wor_{n-1}, \wor^* \equiv \wor_n $ such that $\wor_1$ or $\rema(\wor_1)$ is a proposal source, and each $\wor_k$ makes the first proposal to $\rema(\wor_{k+1})$ under $T$.
			Observe that $\wor_1 \in V$.
			Furthermore, if $\wor_k \in V$, then $\mat(\wor_k) = \mat ' (\wor_k)$ by supposition.
			Hence, $\wor_{k+1}$ makes a proposal under both $T$ and $T ' $.
			Therefore, $\wor_{k+1} \in V$.
			It follows that $\wor^* \in V$, a contradiction to the supposition that $U \cap V$ is empty.
		\end{enumerate}
		Therefore, $\mat = \mat ' $.
	\end{proof}
	
	\begin{proof}[\unskip\nopunct]
		
		\textbf{\textit{Proof of \Cref{lem:exchange}:}}
		Suppose (toward a contradiction) $\Path = (\edg_1, \ldots \edg_n)$ is a cyclic blocking path in $(\Age, \rema \cup \matExch \cup \Imp(\matExch))$.
		Because there is a bijection\footnote{namely, $\rema$} between the workers and firms contained in $\Path$, $n$ is even.
		Define $m \equiv \frac{n}{2}$.
		
		From $\Path$ (after a possible relabeling) define a vector of agents $(\age_1, \age_2, \ldots, \age_{n} \equiv \age_0)$ such that $(\age_{k-1}, \age_{k}) = \edg_{k -1 }$, $\age_1 \in \Wor$, and $\edg_1 \in \Imp(\matExch)$.
		Because $\Path$ is alternating, every odd agent is a worker and every even agent is a firm.
		
		I first show that every agent in $\Path$ is active in the \Exchange stage.
		To see this, suppose (toward a contradiction) that some worker $\age_k$ in $\Path$ is not active during the \Exchange stage.
		Then $\age_k$ makes a proposal during the \Propose stage to $\age_{k+1}$.
		Therefore, $\age_{k+2}$ makes a proposal during the \Propose stage.
		I can iterate this argument to show that every worker in $\Path$ makes a proposal during the \Propose stage.
		Because $\Path$ is a blocking path, each firm in $\Path$ prefers her respective proposal to her $\matProp$-partner.
		Because $\age_{k-2}$ is rejected by $\age_{k-1}$, it necessarily follows that $\matProp(\age_k) = \age_{k-1}$.
		Therefore, $\age_k$ is active in the \Exchange stage, a contradiction.
		Therefore, every agent in $\Path$ is active during the \Exchange stage.
		
		Let $t_k$ be the iteration of the \textbf{while} \ldots \textbf{do}  loop of the \Exchange algorithm that $\age_k$ sits down in.\footnote{That is, if $\age_k$ sits down on the fourth iteration of the while loop, then $t_k = 4$.}
		During the \Exchange algorithm every worker $\age_{2k-1}$ points to firm $\age_{2k}$; hence, firm $\age_{2k}$ sits down weakly earlier than worker $\age_{2k-1}$.
		In symbols, $t_{2k-1} \geq t_{2k}$ for all $1 \leq k \leq m$.
		Because $\edg_1 \in \Imp(\matExch)$, it follows that $t_{1} > t_{2}$.
		Therefore,
		\begin{align*}
			\sum_{k = 1}^m t_{2k-1} > \sum_{k = 1}^m t_{2k} 
		\end{align*}
		
		However, every worker $\age_{2k+1}$ sits down at the same time firm $\age_{2k}$ sits down.
		In symbols, $t_{2k+1} = t_{2k}$ for all $1 \leq k \leq m$.
		Therefore,
		\begin{align*}
			\sum_{k = 1}^m t_{2k+1} = \sum_{k = 1}^m t_{2k} 
		\end{align*}
		Because $\sum_{k = 1}^m t_{2k+1} = \sum_{k = 1}^m t_{2k-1}$, I reach a contradiction.
	\end{proof}

	\begin{proof}[\unskip\nopunct]
		
		\textbf{\textit{Proof of \Cref{pro:nofreeblockingpairs}:}}
		For the first claim, suppose (toward a contradiction) that $\wor$ and $\fir$ are both free agents in $\mat$ who also both prefer each other to $\mat(\wor)$ and $\mat(\fir)$, respectively.
		I construct a blocking path in $(\Age, \rema \cup \mat \cup \Imp(\mat))$, a contradiction to the supposition that $\mat$ is in the agreeable core.
		
		Because $\wor$ is a free agent in $\mat$, $\wor$ lies on an acyclic, complete, and alternating path $\Path_\wor$ of $(\Age, \rema, \mat)$.
		Rewrite $\Path_\wor$ such that 
		\begin{align*}
			\Path_\wor &= \big(\edg_1, \ldots, \edg_{k-1},\big(\rema(\wor), \wor\big), \big(\wor, \mat(\wor)\big), \ldots\big) 
		\end{align*}
		Similarly, there is a complete and alternating $\Path_\fir$ such that 
		\begin{align*}
			\Path_\fir &= \big(\ldots, \big(\mat(\fir), \fir\big), \big(\fir, \rema(\fir)\big), \edg_{k+1}, \ldots, \edg_n\big) 
		\end{align*}
		There are two cases:
		\begin{enumerate}
			\item 
			\textit{$\Path_\wor$ and $\Path_\fir$ do not intersect:}
			Then
			\begin{align*}
				\big(\edg_1, \ldots, \edg_{k-1}, \big(\wor, \fir\big) , \edg_{k+1}, \edg_n\big)
			\end{align*} 
			is a blocking path of $\mat$.
			\item 
			\textit{$\Path_\wor$ and $\Path_\fir$ do intersect:}
			Then let $i$ be the greatest index less than $k$ such that $\edg_i$ is in $\Path_\fir$.
			Let $\edg_j$ be the edge in $\Path_\fir$ such that $\edg_i = \edg_j$.
			Therefore the path
			\begin{align*}
				\big(\edg_j, \ldots, \edg_{k-1}, \big(\wor, \fir\big) , \edg_{k+1}, \ldots \edg_{j - 1}\big)
			\end{align*} 
			is a blocking path of $\mat$.
		\end{enumerate}
		In either case there is a blocking path of $\mat$.
		But then $\mat$ is not in the agreeable core, a contradiction.
		
		For the second claim I can repeat the argument from the first claim, substituting the edge $(\wor, \wor)$ for $\{\wor, \mat(\wor)\}$ in path $\Path_\wor$ and $(\fir, \fir)$ for $\{\mat(\fir), \fir\}$ in path $\Path_\fir$.
	\end{proof}

	\begin{lemma}\label{lem:ruralhospitalsmeetjoin}
		Let $\mat $ and $\nat$ be structurally similar matches in the agreeable core.
		Then $(\mat \join \nat )(\wor) \in \Fir$ if and only if $\mat(\wor) \in \Fir$ \underline{or} $\nat(\wor) \in \Fir$.
		Similarly, $(\mat \join \nat )(\fir) \in \Wor$ if and only if $\mat(\fir) \in \Wor$ \underline{and} $\nat(\fir) \in \Wor$.
		A symmetric result holds for $\meet$.
	\end{lemma}
	\begin{proof}
		Both statements clearly hold for every agent that is not free in $\mat$ (and $\nat$ because $\mat$ and $\nat$ are structurally similar).
		Hence, I show that the statements hold for the free agents in $\mat$.
		
		For the first statement:
		\begin{itemize}
			\item
			\textit{For the $(\Rightarrow)$ direction:}
			I show that if $\mat(\wor) \notin \Fir$ and $\nat(\wor) \notin \Fir$, then $(\mat \join \nat )(\wor) \notin \Fir$.
			Then $\mat(\wor) = \nat(\wor) = \wor$, which implies $(\mat \join \nat )(\wor)  = \wor$.
			Thus $(\mat \join \nat )(\wor) \notin \Fir$.
			
			\item
			\textit{For the $(\Leftarrow)$ direction:}
			I show that if $\mat(\wor) \in \Fir$ or $\nat(\wor) \in \Fir$, then $(\mat \join \nat )(\wor) \in \Fir$.
			To see this, note that if $\mat(\wor) = \fir$ or $\nat(\wor) = \fir$, then $\wor$ strictly prefers $\fir$ to being unmatched ($\wor$) by \Cref{pro:nofreeblockingpairs}.
			Therefore, $\mat \join \nat$ cannot leave $\wor$ unmatched and therefore $(\mat \join \nat )(\wor) \in \Fir$.
		\end{itemize}
		
		For the second statement:
		\begin{itemize}
			\item
			\textit{For the $(\Rightarrow)$ direction:}
			I show that if either $\mat(\fir) \notin \Wor$ or $\nat(\fir) \notin \Wor$, then $(\mat \join \nat )(\fir) \notin \Wor$.
			Then $\mat(\fir) = \fir$ or $\nat(\fir) = \fir$.
			By \Cref{pro:nofreeblockingpairs}, $\fir$ weakly prefers both $\mat(\fir)$ and $\nat(\fir)$ being unmatched.
			By the definition of $\join$, $(\mat \join \nat)(\fir) = \fir$.
			Therefore, $(\mat \join \nat )(\fir) \notin \Wor$.
			
			\item
			\textit{For the $(\Leftarrow)$ direction:}
			I show that if $\mat(\fir) \in \Wor$ and $\nat(\fir) \in \Wor$, then $(\mat \join \nat )(\fir) \in \Wor$.
			Then $\{\mat(\fir), \nat(\fir)\} \subseteq \Wor$.
			Therefore $(\mat \join \nat) (\fir) \in \Wor$.
		\end{itemize}
		This completes the proof.
	\end{proof}

	\begin{proof}[\unskip\nopunct]
		\textbf{\textit{Proof of \Cref{lem:joinmeetmatch}:}}
		I draw my proof from the proof of Theorem 2.16 in \cite{Roth_Sotomayor_1990}.
		I show that $\mat \join \nat$ is a match; the argument for $\mat \meet \nat$ is symmetric.
		
		Because the free agents are the same in $\mat$ and $\nat$, I need only to show that $\mat \join \nat $ is a match on the free agents of $\mat $ and $\nat$; all other matches are left unchanged because $\mat$ and $\nat$ are structurally similar.
		It is immediate from the definition of $\join$ that items 1 and 2 from the definition of a match hold.
		That is, I only need that $(\mat \join \nat )(\age) = \bge \iff (\mat \join \nat )(\bge) = \age$.
		Of course, if $\age = \bge$ then the statement is tautological; hence, I prove for $\wor \in \Wor$ and $\fir \in \Fir$:
		\begin{align*}
			(\mat \join \nat )(\wor) = \fir \iff (\mat \join \nat )(\fir) = \wor.
		\end{align*}
		
		\textit{For the $(\Rightarrow)$ direction:}
		I show that $(\mat \join \nat )(\wor) = \fir$ implies $(\mat \join \nat )(\fir) = \wor$.
		I consider the case when $\mat(\wor) = \fir$; the other case is symmetric.
		Suppose (toward a contradiction) that $(\mat \join \nat )(\fir) \neq \wor$.
		Then $(\mat \join \nat )(\fir)  = \nat(\fir)$.
		Then $\fir$ strictly prefers $\wor$ to $\nat(\fir) $ and $\wor$ strictly prefers $\fir$ to $\nat(\wor)$, so $\wor$ and $\fir$ is a blocking pair of $\nat$, a contradiction by \Cref{pro:nofreeblockingpairs}.
		This completes this direction.
		
		\textit{For the $(\Leftarrow)$ direction:}
		I show that $(\mat \join \nat )(\fir) = \wor$ implies $(\mat \join \nat )(\wor) = \fir$.
		I define a sequence of sets, then study their cardinality.
		Let
		\begin{align*}
			\Wor ' &\equiv \{\wor \in \Wor \spbr (\mat \join \nat )(\wor) \in \Fir\} \\
			&= \{\wor \in \Wor \spbr \mat (\wor) \in \Fir \t{ or } \nat (\wor) \in \Fir\} 
			&\because \Cref{lem:ruralhospitalsmeetjoin}.
		\end{align*}
		and 
		\begin{align*}
			\Fir '  &\equiv  \{\fir \in \Fir \spbr (\mat \join \nat)(\fir) \in \Wor \}\\
			&= \{\fir \in \Fir \spbr \mat(\fir) \in \Wor \t{ and } \nat(\fir) \in \Wor \} 
			&\because \Cref{lem:ruralhospitalsmeetjoin}.
		\end{align*}
		Observe the following relations: 
		\begin{align*}
			|\Fir ' | &= |\mat(\Fir ' )| & \because \t{$\mat$ is a match} \\
			\mat(\Fir ' ) &\subseteq \Wor ' & \because \t{Definition of $\Fir ' $ and $\Wor ' $}
		\end{align*}
		Therefore $ |\Fir ' | \leq | \Wor '  |$.
		Similarly,
		\begin{align*}
			|\Wor '| & =  | (\mat \join \nat)(\Wor ' ) |  & \because \t{$(\Rightarrow)$ implication} \\
			(\mat \join \nat) (\Wor ' ) &\subseteq \Fir ' & \because \t{$(\Rightarrow)$ implication} 
		\end{align*}
		Therefore $ |\Wor ' | \leq | \Fir '  |$ and thus $|\Wor ' | = |\Fir ' |$.
		Therefore $ |(\mat \join \nat)(\Wor ') | = |\Fir ' | $ and thus $(\mat \join \nat)(\Wor ')  =  \Fir ' $.
		
		The final string of implications is as follows:
		If $(\mat \join \nat)(\fir) \in \Wor$, then $\fir \in \Fir ' $.
		If $\fir \in \Fir'$, then there exists $\wor$ in $\wor \in \Wor '$ such that $(\mat \join \nat)(\wor) = \fir$.
		This completes this direction.

		Therefore, $\mat \join \nat $ satisfies item 3 from the definition of a match and thus $\mat \join \nat $ is a match.
	\end{proof}

	\begin{lemma}\label{lem:structureofjoin}
		Let $\mat $ and $\nat$ be structurally similar matches in the agreeable core.
		Then $\mat \join \nat \subseteq \mat \cup \nat$ and $\Imp(\mat \join \nat) \subseteq \Imp(\mat) \cup \Imp(\nat)$.
		The same holds for $\mat \meet \nat$.
	\end{lemma}
	\begin{proof}
		By construction, $\mat \join \nat $ only contains matches from $\mat$ and $\nat$ and thus $\mat \join \nat  \subseteq \mat \cup \nat$.
		
		Let $\{\wor, \fir \}\in \Imp(\mat \join \nat)$ and let $\AgeFree $ be the free agents in $\mat$ (and $\nat$ because $\mat$ and $\nat$ are structurally similar).
		There are three cases:
		\begin{enumerate}
			\item $| \{\wor, \fir \} \cap \AgeFree | = 0$:
			Then $(\mat \join \nat)(\wor)= \mat(\wor)$ and $(\mat \join \nat) (\fir) = \mat(\fir)$ by construction, so $\{\wor, \fir\} \in \Imp(\mat)$.
			\item $| \{\wor, \fir \} \cap \AgeFree | = 1$:
			Suppose that $\wor \in \AgeFree$; the other case is symmetric.
			Then either $(\mat \join \nat)(\wor)= \mat(\wor)$ or $(\mat \join \nat) (\wor) = \nat(\fir)$; again, let $(\mat \join \nat)(\wor)= \mat(\wor)$ and the other case is symmetric.
			Then $(\mat \join \nat) (\fir) = \mat(\fir)$ by construction, so $\{\wor, \fir\} \in \Imp(\mat)$.
			
			\item $| \{\wor, \fir \} \cap \AgeFree | = 2$:
			This contradicts \Cref{pro:nofreeblockingpairs} and thus cannot happen.
		\end{enumerate}
		In the cases that do not lead to a contradiction I see that $\{\wor, \fir\} \in \Imp(\mat) \cup \Imp(\nat)$, which completes the proof.
	\end{proof}

	\begin{definition}
		A \textit{crossing edge} at $\mat$ contains both a free agent and an agent who is not free at $\mat$.
	\end{definition}

	\begin{lemma}\label{lem:blockingpathsmeetjoin}
		Let $\mat$ and $\nat$ be structurally similar matches in the agreeable core.
		Then any blocking path of $\mat \join \nat$ must contain two crossing edges at $\mat$.
		All crossing edges at $\mat$ of any blocking path of $\mat \join \nat$ are contained in either $\Imp(\mat)$ or $\Imp(\nat)$.
		
		A symmetric result holds for $\mat \meet \nat$.
	\end{lemma}
	\begin{proof}
		Let $\AgeFree$ denote the free agents in $\mat$ (and $\nat$ because $\mat$ and $\nat$ are structurally similar), and let $\Path$ be a blocking path of $\mat \join \nat$.
		
		I first prove that all crossing edges at $\mat$ of any blocking path of $\mat \join \nat$ are contained in either $\Imp(\mat)$ or $\Imp(\nat)$.
		To see this, let $(\age, \bge) \in \Path$ be a crossing edge with $\age\in \AgeFree$ and $\bge \in \Age\backslash\AgeFree$.
		Because $\rema(\AgeFree) = \AgeFree$, it follows that $(\age, \bge) \in \mat\join\nat \cup \Imp(\mat\join\nat)$.
		Because $\mat(\AgeFree) = \AgeFree$ and $\nat(\AgeFree) = \AgeFree$ by construction it follows that $\{\age, \bge\} \notin \mat\join\nat $.
		Therefore $(\age, \bge) \in \Imp(\mat\join\nat)$.
		
		Next, I show that $\Path \not \subseteq \AgeFree$ and $\Path \not \subseteq \Age \backslash \AgeFree$.
		To see this, consider both cases (toward a contradiction in each case):
		\begin{enumerate}
			\item \textit{Suppose $\Path \subseteq \AgeFree$}:
			Then exists an edge $\edg$ in $\Path$ such that $\edg \in \Imp(\mat \join \nat)$.
			By \Cref{lem:structureofjoin}, $\edg \in \Imp(\mat)$ (the other case is symmetric).
			If $\edg = (\wor, \fir)$, then $\edg$ constitutes a blocking pair and contradicts \Cref{pro:nofreeblockingpairs}.
			If $\edg =( \age, \age)$, then $\age$ strictly prefers being unmatched to $\mat$ and contradicts \Cref{pro:nofreeblockingpairs}.
			Therefore, $\Path \not \subseteq \AgeFree$.
			
			\item \textit{Suppose $\Path \subseteq \Age \backslash \AgeFree$}:
			Note that $\mat\join\nat$ agrees with $\mat$ on $\Age \backslash \AgeFree$.
			If $\Path$ blocks $\mat \join \nat$ then $\Path$ blocks $\mat$, a contradiction to the supposition that $\mat$ is in the agreeable core.
			Therefore, $\Path \not \subseteq \Age \backslash \AgeFree$.
		\end{enumerate}
		Therefore, $\Path$ intersects both $\Age$ and $\Age \backslash \AgeFree$.
		By the definition of a path, there exists some crossing edge at $\mat$ in $\Path$.
		
		Third, to see that two crossing edges at $\mat$ exist, suppose not.
		Let $(\age, \bge)$ be the crossing edge at $\mat$ in $\Path$ such that $\age \in \AgeFree$ and $\bge \notin \AgeFree$.
		As observed earlier, $ (\age, \bge) \in \Imp(\mat\join \nat) $.
		By \Cref{lem:structureofjoin}, it follows that $ (\age, \bge) \in \Imp(\mat ) $ (the other case is symmetric).
		Suppose that $\age\in\Wor$; the other case is symmetric.
		Then $\Path$ may be written
		\begin{align*}
			\Path &= (\overbrace{\edg_1, \ldots, \edg_{k-1} ,(\rema(\age), \age)}^{\t{contained in $\AgeFree$}}, \underbrace{(\age, \bge),}_{\t{contained in $\Imp(\mat)$}} \overbrace{\edg_{k} \ldots, \edg_K}^{\t{contained in $\Age\backslash\AgeFree$}}).
		\end{align*}
		Note that every edge from $\edg_k$ to $\edg_K$ exists in $(\Age, \rema \cup \mat \cup \Imp(\mat))$ because $\mat\join\nat$ agrees with $\mat$ for these agents.
		Because $\age\in\AgeFree$, there is an alternating, complete, and acyclic path $\Path^\age$ in $(\Age, \rema \cup \mat \cup \Imp(\mat))$ such that 
		\begin{align*}
			\Path^\age &= (\edg_1^\age, \ldots, \edg_{l-1}^\age ,(\rema(\age), \age), (\age, \mat(\age)), \edg_{l}^\age \ldots, \edg_L^\age ).
		\end{align*}
		Because $\rema(\AgeFree) = \AgeFree$ and $\mat(\AgeFree) = \AgeFree$ by construction, it follows that every agent in $\Path^\age$ is in $\AgeFree$.
		Observe that the path
		\begin{align*}
			\Path^* &= (\edg_1^\age, \ldots, \edg_{l-1}^\age ,(\rema(\age), \age), (\age, \bge), \edg_{k} \ldots, \edg_K ).
		\end{align*}
		is a blocking path of $\mat$, a contradiction.
		Hence, there are at least two edges that intersect both $\AgeFree$ and $\Age\backslash\AgeFree$.
	\end{proof}

	\begin{proof}[\unskip\nopunct]
		
		\textbf{\textit{Proof of \Cref{thm:joinmeet}:}}
		I show that $\mat \join \nat$ is in the agreeable core; the argument for $\mat \meet \nat$ is symmetric.
		By \Cref{lem:joinmeetmatch}, $\mat \join \nat$ is a match.
		Because $\mat $ and $\nat$ are both individually rational, $\mat \join \nat$ is individually rational.
		The remaining step is to show that there are no blocking paths of $\mat \join \nat$.
		
		Suppose (toward a contradiction) that $\mat \join \nat$ is blocked by an agreeable coalition.
		By \Cref{pro:ac-iff-nbps}, there is a blocking path $\Path$ of $\mat \join \nat$.
		Let $\AgeFree$ denote the free agents in $\mat$ (and $\nat$ because $\mat$ and $\nat$ are structurally similar).
		
		By \Cref{lem:blockingpathsmeetjoin}, there are two crossing edges at $\mat$ in $\Path$, and both of these is in $\Imp(\mat \join \nat)$.
		There are two cases:
		\begin{enumerate}
			\item
			\textit{There exists two crossing edges $\edg_k$ and $\edg_K$ at $\mat$ in path $\Path$ such that the edges $\edg_{k+1}, \ldots, \edg_{K-1}$ (if any) are contained within $\Age\backslash\AgeFree$.}
			Let $\{\wor, \fir\} = \edg_k$ and $(\wor' , \fir') = \edg_K$ with $\age_k, \age_K \in \AgeFree$ and $\bge_k, \bge_K \in \Age\backslash\AgeFree$.
			Because $\mat \join \nat \wpref_\Wor \mat$ and $\mat \join \nat \wpref_\Wor \nat$, it follows that one of $(\wor, \fir) \in \Imp(\mat)$ and $(\wor, \fir) \in \Imp(\nat)$.
			By \Cref{lem:structureofjoin}, let $(\wor ' , \fir ' ) \in \Imp(\mat)$ (the other case is symmetric).

			Because $\wor \in \AgeFree$, there exists an acyclic, complete, and alternating path $\Path^\wor$ of $(\Age, \rema \cup \mat\cup \Imp(\mat))$:
			\begin{align*}
				\Path^\wor = (\edg_1^\wor , \ldots, \edg_{i-1}^\wor, \{\rema(\wor), \wor\}, \{\wor, \mat(\wor) \}, \ldots).
			\end{align*}
			Similarly because $\fir ' \in \AgeFree$:
			\begin{align*}
				\Path^{\fir ' } = (\ldots, (\mat({\fir ' }), {\fir ' }), ({\fir ' }, \rema({\fir ' })), \edg_{j-1}^{\fir ' } , \ldots, \edg_{1}^{\fir ' }).
			\end{align*}
			Then the path
			\begin{align*}
				\Path^* = (
				\underbrace{
					\edg_1^\wor , \ldots, \edg_{i-1}^\wor, (\rema(\wor), \wor),
				}_{
					\Path^\wor
				}
				\overbrace{
					(\wor,\fir), \edg_{k+1}, \ldots, \edg_{K-1}, (\wor ', \fir '),
				}^{
					\Path
				}
				\underbrace{
					({\fir ' }, \rema({\fir ' })), \edg_{j-1}^{\fir ' } , \ldots, \edg_{1}^{\fir ' }
				}_{
					\Path^{\fir ' }
				}
				).
			\end{align*}
			is a blocking path of $\mat$, a contradiction to the supposition that $\mat$ is in the agreeable core.
			
			\item 
			\textit{There does not exist two crossing edges $\edg_k$ and $\edg_K$ at $\mat$ in path $\Path$ such that the edges $\edg_{k+1}, \ldots, \edg_{K-1}$ (if any) are contained within $\Age\backslash\AgeFree$.}
			Let $(\age, \bge)$ be a crossing edge of $\mat$ of $\Path$ with $\age\in \AgeFree$.
			Let $\bge\in\Wor$; the other case is symmetric.
			The supposition implies that $\Path$ must be acyclic and hence can be written
			\begin{align*}
				\Path = (\underbrace{\edg_1 , \ldots, \edg_{k-1},}_{\t{contained in $\Age\backslash\AgeFree$}} (\bge, \age ),  (\age, \rema(\age)), \ldots).
			\end{align*}
			Because $\age \in \AgeFree$, there exists an acyclic, complete, and alternating path $\Path^\age$ of $(\Age, \rema, \mat\cup \Imp(\mat))$:
				\begin{align*}
				\Path^\age = (\ldots, (\mat(\age), \age), (\age, \rema(\age)), \edg_{i-1}^\age , \ldots, \edg_{1}^\age).
			\end{align*}
			Then the path
			\begin{align*}
				\Path^* = (
				\overbrace{
					\edg_1 , \ldots, \edg_{k-1},
				}^{
					\Path
				}
				\underbrace{
					(\age, \rema(\age)), \edg_{i-1}^\age , \ldots, \edg_{1}^\age
				}_{
					\Path^{\age}
				}
				).
			\end{align*}
			is a blocking path of $\mat$ because $\mat$ and $\mat \join \nat$ agree on the agents in $\Age\backslash\AgeFree$.
			This is a contradiction to the supposition that $\mat$ is in the agreeable core.
		\end{enumerate}
		Therefore, there are no blocking paths of $\mat \join \nat$, which implies that $\mat \join \nat$ is in the agreeable core.
	\end{proof}

	\begin{proof}[\unskip\nopunct]
		
		\textbf{\textit{Proof of \Cref{pro:man}:}}
		Consider the following counterexample.
		There are three workers denoted by the numbers $1$, $2$, and $9$, and three firms denoted by the letters $A$, $B$, and $Z$.
		Workers $1$ and $2$ are reference matched to $A$ and $B$, respectively, while worker $9$ and firm $Z$ are each reference matched to him or herself.
		Formally:
		\begin{align*}
			&\rema(1) = A  &&\rema(2) = B  &&\rema(9) = 9\\
			&\rema(A) = 1  &&\rema(B) = 2  &&\rema(Z) = Z
		\end{align*}
		
		A profile of preferences $\spref$ and an alternate profile of worker preferences are given in \Cref{fig:pref}.
		I use the circles to indicate match $\matcir$, the squares to indicate match $\matsqu$, and $\tilde{\phantom{o}}$ to indicate $\mattil$.
		\begin{align*}
			&\matcir(1) = B  &&\matcir(2) = A  &&\matcir(9) = Z\\
			&\matcir(A) = 2  &&\matcir(B) = 1  &&\matcir(Z) = 9\\
			\\
			&\matsqu(1) = A  &&\matsqu(2) = Z  &&\matsqu(9) = B\\
			&\matsqu(A) =1  &&\matsqu(B) = 9  &&\matsqu(Z) = 2 \\
			\\
			&\mattil(1) = A  &&\mattil(2) = B  &&\mattil(9) = Z\\
			&\mattil(A) =1  &&\mattil(B) = 2  &&\mattil(Z) = 9
		\end{align*}
		I keep the firm preference profile fixed at $\spref_A$, $\spref_B$, and $\spref_Z$ for the firms and only specify preferences for the workers.
		
		\begin{figure}
			\begin{center}
				\begin{tikzpicture}
					\node[pref] (pal) [] {};
					
					\node[pref] (p2) [left = 0mm of pal] {
						\begin{tabular}{C | C | C }
							\spref_{1} ' & \spref_2 & \spref_9 \\
							\hline  \hline 
							\circled B & \squared Z &  \squared{ B} \\
							\textcolor{gray}{ Z} & \textcolor{black}{\circled{ A }}& \textcolor{black}{\circled{ Z}} \\
							\squared{ A} &  B &  \\
						\end{tabular}
					};	
					
					\node[pref] (p3) [right = 0mm of pal] {
						\begin{tabular}{C | C | C }
							\spref_{1} ' & \spref_2 ' & \spref_9 \\
							\hline  \hline 
							\circledinv{B} & \squared Z &  \squared{ B} \\
							\textcolor{gray}{\circledinv Z} & \textcolor{gray}{{ A }}& \textcolor{black}{{\tilde Z}} \\
							\squared{\tilde A} & \tilde B &  \\
						\end{tabular}
					};	
					
					\node[pref] (p1) [left = 2.5mm of p2] {
						\begin{tabular}{C | C | C }
							\spref_{1} & \spref_2 & \spref_9 \\
							\hline  \hline 
							\circled B &  Z &  { B} \\
							\textcolor{black}{ Z} & \textcolor{black}{\circled{ A }}& \textcolor{black}{\circled{ Z}} \\
							\squaredinv{ A} &  B &  \\
						\end{tabular}
					};
					
					\node[pref] (p4) [right = 2.5mm of p3] {
						\begin{tabular}{C | C | C }
							\spref_{1} ' & \spref_2 ' & \spref_9 ' \\
							\hline  \hline 
							{B} & \squared Z &  \squared{ B} \\
							\textcolor{gray}{\circledinv Z} & \textcolor{gray}{{ A }}& \textcolor{gray}{{ Z}} \\
							\squared{ A} &  B &  \\
						\end{tabular}
					};	
					
					\node[pref] (psch) [below = 30mm of pal] {
						\begin{tabular}{C | C | C }
							\spref_A& \spref_B & \spref_Z \\
							\hline \hline
							A & \squared { 9} & \circled {\tilde 9} \\
							\circled { 2} & \circled 1 &  1 \\
							\squared {\tilde 1}& \tilde 2 & \squared 2 \\
						\end{tabular}
					};
					
					\node[pref] (p1name) [above = 1mm of p1] {$P_1$};
					\node[pref] (p2name) [above = 1mm of p2] {$P_2$};
					\node[pref] (p3name) [above = 1mm of p3] {$P_3$};
					\node[pref] (p4name) [above = 1mm of p4] {$P_4$};

				\end{tikzpicture}
				\caption{Tables provide preferences $\spref$ and alternate worker preferences $\spref '$. A grayed-out firm in $\spref '$ indicates that the worker matching to himself more than to that firm. If the table does not specify a preference over an alternative, then they are worse than every alternative listed.}
				\label{fig:pref}
			\end{center}
		\end{figure}
		
		To prove the result, suppose that $\mech$ is not preference manipulable.
		I consider the sequence of preference profiles $P_1$, $P_2$, $P_3$, and $P_4$ formed by swapping $\spref_1 '$ for $\spref_1$, then $\spref_2 ' $ for $\spref_2$, and then $\spref_9' $ for $\spref_9$.
		I use the non-manipulability of $\mech$ to restrict $\mech$ to a unique match in each case.
		I then show that at $P_3$ worker $9$ can profitably deviate to $\spref_9'$, a contradiction to the non-manipulability of $\mech$.
		
		First, I limit the scope of matches I consider.
		Consider any $\mat$ and any $P_j$.
		\begin{itemize}
			\item If $A \spref_1 \mat(1)$ then $1$ strictly prefers $\rema(1)$ to $\mat(1)$, hence $\mat$ is not in the agreeable core; the same holds for $B \spref_2 \mat(2)$, $2 \spref_B \mat(B) $, and $1 \spref_A \mat(A)$.
			\item If $j \neq 4$ and $Z \spref_9 \mat(9)$, then $\{9, Z\}$ is an agreeable coalition that blocks $\mat$.
			\item If $j = 4$ and $Z \spref_9 \mat(9)$, then $\mat$ in the agreeable core implies that $\mat(1) \neq Z$ and hence $B\spref_9 \mat(9)$ implies that $\{2, 9, B, Z\}$ is an agreeable coalition that blocks $\mat$; hence, if $\mat$ is in the agreeable core then $\mat(9) = B$.
			\item If $\mat(1) = Z$ and $\mat(2) = A$, then for $P_1$ $\{1,A,Z\}$ is an agreeable blocking coalition and for $P_2$, $P_3$, and $P_4$ $A \spref_1 ' Z$.
			Hence for all $P_j$ $\mat(1) = Z$ and $\mat(2) = A$ imply that $\mat$ is not in the agreeable core.
		\end{itemize}
		It follows that every worker is matched to a firm, and thus every firm is matched to a worker.
		Therefore, any match in the agreeable core only occurs between agents who are listed on each other's preferences in \Cref{fig:pref}.
		An exhaustive search reveals that $\matcir$, $\matsqu$, and $\mattil$ are the only matches that meet these criteria.
		
		For $P_1$, the agreeable core is $\{\matcir\}$ because:
		\begin{itemize}
			\item[\cmark] $\matcir$ is the output of the PE algorithm and hence is in the agreeable core.
			\item[\xmark] $\matsqu$ is blocked by the agreeable coalition $\{1,A,Z\}$ with any deviation $\mat '$ such that $\mat ' (1 ) = Z$ and $\mat '(A) = A$.
			\item[\xmark] $\mattil$ is blocked by the agreeable coalition $\{1,2,A,B\}$ with any deviation $\mat '$ such that $\mat ' (1 ) = B$ and $\mat '(2) = A$.
		\end{itemize}
		Hence, $\mech(P_1 ) = \matcir$.
		
		For preferences $P_2$, the agreeable core is $\{\matcir, \matsqu\}$ because:
		\begin{itemize}
			\item[\cmark] $\matcir$ does not match any worker to a firm he dropped from his preference, so every blocking coalition under these preferences forms under the prior preferences.
			\item[\cmark] $\matsqu$ is the output of the PE algorithm and hence is in the agreeable core.
			\item[\xmark] $\mattil$ is blocked by the agreeable coalition $\{1,2,A,B\}$ with any deviation $\mat '$ such that $\mat ' (1 ) = B$ and $\mat '(2) = A$.
		\end{itemize}
		If $\mech(P_2) = \matsqu$, then consider the deviation by worker $1$ of misreporting $\spref_1$ at $P_2$.
		Because $\matcir(1) \spref_1 ' \matsqu(1)$, this is a profitable deviation.
		Therefore, because $\mech$ is not preference manipulable, $\mech(P_2) = \matcir$.
		
		For preferences $P_3$, the agreeable core is $\{\matsqu, \mattil\}$ because:
		\begin{itemize}
			\item[\xmark] $\matcir$ matches worker $2$ to firm $A$, which violates the requirement that $\mat(2)\wpref_2 B$.
			\item[\cmark] $\matsqu$ is the output of the PE algorithm and hence is in the agreeable core.
			\item[\cmark] $\mattil$:
			Observe that $Z$ cannot be strictly better off in any blocking coalition, and thus $2$ cannot be strictly better any blocking coalition.
			Furthermore, any agreeable coalition that makes $1$ strictly better off must include $B$ and hence, because the coalition is agreeable, $2$.
			Therefore, any agreeable blocking coalition cannot make any worker strictly better off.
			Hence, $\mattil$ is also in the agreeable core.
		\end{itemize}
		If $\mech(P_3) = \matsqu$, then consider the deviation by worker $2$ of reporting $\spref_2 '$ at $P_2$.
		Because $\matsqu(2) \spref_2  \matcir(2)$, this is a profitable deviation.
		Therefore, because $\mech$ is not preference manipulable, $\mech(P_3) = \mattil$.
		
		In this final step, I note that the core under $P_4$ is the singleton $\matsqu$.
		To see this, observe that $\matcir$ and $\mattil$ each match a worker to a firm he lists below his reference match, and therefore none of these three matches is in the agreeable core.
		$\matsqu$ is the output of the PE algorithm and hence is in the agreeable core.
		However, consider the deviation by worker $9$ of reporting $\spref_9 ' $ at $P_3$.
		Because $\matsqu(9) \spref_9 \mattil(9)$, this is a profitable deviation.
		Therefore, $\mech$ is preference manipulable, a contradiction.
	\end{proof}

	\textit{\textbf{Introduction to the proofs of \Cref{thm:nomanPE}:}}
	
	\begin{lemma}\label{lem:actBloPair}
		For any $\matProp$, there is no $\wor$ and $\fir$ such that all three conditions are true:
		\begin{enumerate}
			\item $\wor$ is active in the Propose stage; and
			\item $\rema(\fir) \neq \matProp(\fir)$; and
			\item $(\wor, \fir)$ is a blocking pair of $\matProp$.
		\end{enumerate}
	\end{lemma}
	\begin{proof}
		Toward a contradiction, suppose $(\wor, \fir)$ is such a pair.
		Because $\wor$ is active and $\wor$ strictly prefers $\fir$ to $\matProp(\wor)$, $\wor$ makes a proposal to $\fir$.
		Because $\rema(\fir) \neq \matProp(\fir)$ and $\fir$ strictly prefers $\wor$ to $\matProp(\fir)$, $\fir$ does not reject the proposal from $\wor$.
		This is a contradiction to the supposition that $(\wor, \fir)$ is a blocking pair.
		Therefore, no such pair exists.
	\end{proof}

	\begin{proof}[\unskip\nopunct]
		
		\textbf{\textit{Proof of \Cref{thm:nomanPE}:}}
		Suppose (toward a contradiction) that $\wor$ can profitably misreport $\wpref_\wor ' $ but that $\wor$ is not active in both the $\wpref_\wor '$-Propose and $\wpref_\wor '$-Exchange stages.
		First I consider the case when $\wor $ is not active in the $\wpref_\wor '$-Propose stage, and then the case when $\wor$ is not active in the $\wpref_\wor '$-Exchange stage.
		Before continuing, I note that $\wor$'s preferences do not affect whether $\wor$ is active in the $\wpref_\wor$-Propose or $\wpref_\wor '$-Propose stages.
		
		Suppose $\wor$ is not active in the $\wpref_\wor '$-Propose stage.
		The rest of the proof follows directly from the non-manipulability of the Top Trading Cycles algorithm.
		This is well-known in the literature; see \cite{ma_strategy-proofness_1994} for one such proof, and footnote 4 of \cite{dur_two-sided_2019} for a list of references to other proofs.
		This is a contradiction to the supposition that $\wor$ can profitably misreport $\wpref_\wor ' $.
		
		The remainder of the proof is built on the proof of the blocking lemma of \cite{Roth_Sotomayor_1990}.
		
		For the remainder of the proof, suppose that $\wor$ is active in the $\wpref_\wor '$-Propose stage but not in the $\wpref_\wor '$-Exchange stage.
		Therefore, $\wor$ is active in the $\wpref_\wor$-Propose stage as well.
		Let $\matProp' $ be the output of the $\wpref_\wor '$-Propose stage.
		Let $\Wor ' $ be the set of workers who strictly prefer $\matProp ' $ to $\matProp$ and are active in the $\wpref_\wor$-Propose stage.
		By supposition, $\wor \in \Wor ' $, so $\Wor ' $ is nonempty.
		Because $\matProp$ is individually rational, every worker in $\Wor '$ is active in the $\wpref_\wor '$-Propose stage but not active in the $\wpref_\wor '$-Exchange stage.
		
		Next, I show that there always exists a worker $\wor^*$ and firm $\fir^*$ such that the following four conditions hold:
		\begin{enumerate}
			\item $\wor^*$ is active in the $\wpref_\wor '$-Propose stage; and
			\item $\rema(\fir^*) \neq \matProp ' (\fir^*)$; and
			\item $(\wor^*, \fir^*)$ is a blocking pair of $\matProp ' $; and
			\item $\wor^* \neq \wor$.
		\end{enumerate}
		There are two cases:
		\begin{enumerate}
			\item $\matProp ' (\Wor ' ) = \matProp(\Wor ' )$:
			First, I show that every $\wor ' $ who is active in the $\wpref_\wor $-Propose stage is also active in the $\wpref_\wor ' $-Propose stage.
			To see this, note that there is a sequence of workers $\wor_1, \ldots, \wor_n \equiv \wor ' $ such that $\wor_k$ is acceptable\footnote{That is, $\fir^*$ prefers $\tilde\wor$ to $\rema(\fir^*)$.} to $\rema(\wor_{k+1})$ and $\wor_k$ is the first worker to propose to $\rema(\wor_{k+1})$ in the $\wpref_\wor $-Propose stage.
			Toward a contradiction, suppose that some workers in the sequence are not active in the $\wpref_\wor '$-Propose stage, and let $\wor_k$ be the one with the lowest index.
			Obviously, $k\neq 1$.
			By construction, $\wor_{k-1}$ is active in the $\wpref_\wor '$-Propose stage and prefers $\matProp ' $ to $\matProp$ because $\wor_{k-1}$ does not propose to $\rema(\wor_k)$.
			By supposition, $\matProp '(\Wor ' )= \matProp (\Wor ' )$.
			Therefore, there is some acceptable $\tilde \wor \in \Wor ' $ who proposes to $\rema(\wor_k)$ in the $\wpref_\wor '$-Propose stage.
			Hence $\wor_k$ is active in the $\wpref_\wor '$-Propose stage, a contradiction.
			Therefore $\wor '$ is active in the $\wpref_\wor '$-Propose stage.
			
			Let $\Fir ' \equiv \matProp ' (\Wor ' )$.
			Fix an arbitrary order of proposals and let $\fir^* $ be the last firm in $\Fir ' $ to receive a proposal from an acceptable worker in $\Wor ' $ in the $\wpref_\wor $-Propose stage. 
			Because $\matProp ' $ is individually rational, each worker in $\Wor ' $ is acceptable to her $\matProp '$-partner.
			Because $ \Wor' $ is nonempty and every worker in $\Wor ' $ makes a proposal in the $\wpref_\wor $-Propose stage, such a firm exists.
			
			Because every worker in $\Wor ' $ strictly prefers $\matProp '  $ to $\matProp$ and is active in the $\wpref_\wor $-Propose stage, every firm in $\Fir ' $ must have rejected at least one proposal from an acceptable worker in $\Wor ' $ in the $\wpref_\wor $-Propose stage (namely, the firm's $\matProp ' $-partner).
			Thus $\fir^*$ was matched to some $\wor^* \in \Wor $ when she received this last proposal and $\fir^*$ rejects $\wor^*$.
			Note that $\wor^*$ cannot be in $\Wor ' $; otherwise, after being rejected by $\fir^*$, $\wor^*$ would have proposed to another firm in $\Fir ' $ because $\matProp (\Wor ' ) = \Fir ' $.
			Hence, $\wor^* \neq \wor$.
			Note that $\wor^*$ is active in the $\wpref_\wor $-Propose stage, so he is also active in the $\wpref_\wor ' $-Propose stage.
			\textit{\textbf{This satisfies conditions 1 and 4.}}
			 
			Next, note that $\rema(\fir^*) \neq \matProp ' (\fir^*)$ because $\matProp  (\fir^*) \in \Wor ' $ and no worker in $\Wor ' $ is active in the $\wpref_\wor ' $-Exchange stage (see earlier comment).
			\textit{\textbf{This satisfies condition 2.}}
			
			Finally, note that $\fir^*$ strictly prefers $\wor^*$ to $\matProp ' (\fir^*)$ because  $\fir^*$ must have rejected $\matProp ' (\fir^*)$ but $\wor^*$ was tentatively accepted immediately prior to $\fir^*$ accepting $\matProp (\fir^*)$ in the $\wpref_\wor$-Propose stage.
			Because $\wor^*$ is active in both the $\wpref_\wor$- and $\wpref_\wor '$-Propose stages and $\wor\notin\Wor ' $, it follows that $\wor$ weakly prefers $\matProp$ to $\matProp '$ 
			Because $\wor^*$ strictly prefers $\fir$ to $\matProp(\wor)$ and $\wor$ weakly prefers $\matProp$ to $\matProp ' $, it follows that $\wor^*$ strictly prefers $\fir$ to $\matProp' (\wor^*)$.
			Therefore, $(\wor^*, \fir^*)$ is a blocking pair of $\matProp ' $.
			\textit{\textbf{This satisfies condition 3.}}
			
			This completes this case.

			\item $\matProp ' (\Wor ' ) \neq \matProp(\Wor ' )$:
			Fix an arbitrary order of proposals and let $\fir^*$ be the first firm in $\matProp ' (\Wor ' ) \backslash \matProp (\Wor ' )$ to receive a proposal from $\matProp' (\fir^*)$ in the $\wpref_\wor '$-Propose stage.
			Note that $\rema(\fir^*) \neq \matProp ' (\fir^*)$ because $\matProp ' (\fir^*) \in \Wor ' $ and no worker in $\Wor ' $ is active in the $\wpref_\wor ' $-Exchange stage (see earlier comment).
			\textit{\textbf{This satisfies condition 2.}}
			
			Let $\wor^* \equiv \matProp(\fir^*)$.
			Note that $\wor^* \notin \Wor ' $ and thus $\wor^* \neq \wor$.
			\textit{\textbf{This satisfies condition 4.}}
			
			Let $\wor ' \equiv \matProp ' (\fir^*)$.
			Note that $\wor ' $ proposes to $\fir^*$ in the $\wpref_\wor $-Propose stage because $\wor ' \in \Wor ' $.
			Therefore, $\wor^*$ is active in the $\wpref_\wor $-Propose stage.
			
			Next, I show that $\wor^* $ is active in the $\wpref_\wor '$-Propose stage.
			To see this, note that there is a sequence of workers $\wor_1, \ldots, \wor_n \equiv \wor^*$ such that in the $\wpref_\wor $-Propose stage, $\wor_k$ is acceptable to $\rema(\wor_{k+1})$ and $\wor_k$ is the first worker to propose to $\rema(\wor_{k+1})$.
			Toward a contradiction, suppose that some workers in the sequence are not active in the $\wpref_\wor '$-Propose stage, and let $\wor_k$ be the one with the lowest index.
			Obviously, $k\neq 1$.
			By construction, $\wor_{k-1}$ is active in the $\wpref_\wor '$-Propose stage and prefers $\matProp ' $ to $\matProp$ because $\wor_{k-1}$ does not propose to $\rema(\wor_k)$.
			Therefore, $\wor_{k-1}$ must propose to $\matProp ' (\wor_{k-1})$ at an earlier step of the $\wpref_\wor '$-Propose stage than $\wor ' $ proposes to $\fir^*$, a contradiction to the supposition that $\wor ' $ is the first such worker to do so.
			Hence $\wor_k$ is active in the $\wpref_\wor '$-Propose stage, a contradiction.
			Therefore $\wor^*$ is active in the $\wpref_\wor '$-Propose stage.
			\textit{\textbf{This satisfies condition 1.}}
			
			Note that $\wor^*$ strictly prefers $\fir^*$ to $\matProp ' (\wor^*)$ because $\wor^* \notin \Wor ' $, $\wor^*$ is active in both Propose stages, and $\fir^* = \matProp (\wor^*)  \neq \matProp ' (\wor^*)$.
			Similarly, $\wor^* \neq \rema(\fir^*)$ because $\matProp ' $ is individually rational.
			Because $\wor ' $ is rejected by $\fir^*$ in favor of $\wor^*$ in the $\wpref_\wor $-Propose stage, it follows that $\fir^*$ strictly prefers $\wor^*$ to $\wor ' $.
			Therefore, $(\wor^*, \fir^*)$ is a blocking pair of $\matProp '$.
			\textit{\textbf{This satisfies condition 3.}}
			
			This completes this case.
		\end{enumerate}
		
		Because only $\wor$ misreports, $\wor^*$ in each case has the same preferences.
		Therefore, the conditions of \cref{lem:actBloPair} are met, a contradiction to the supposition that $\matProp ' $ is the output of the $\wpref_\wor ' $-Propose stage.
		This completes the proof.
		
	\end{proof}

	\begin{proof}[\unskip\nopunct]
		
		\textbf{\textit{Proof of \Cref{thm:art}:}}
		This proof has two parts.
		In the first, I show that $\matProp ' (\wor) = \fir$.
		In the second, I show that $\wor$ is not active the $\rema ' $-Propose stage.
		
		Suppose (toward a contradiction) that $\matProp ' (\wor) \neq \fir$.
		I show that every worker who proposes in the $\rema$-Propose stage weakly prefers $\matProp $ to $\matProp ' $.
		This contradicts the supposition that $\wor$ strictly prefers $\matProp ' $ to $\matProp$.
		
		First, choose an arbitrary proposal order for the $\rema$-Propose stage such that $\wor$ only makes his first proposal if he is the only active worker.
		Use the notation $(\tilde\wor, \tilde\fir)$ to indicate that $\tilde\wor$ proposes to $\tilde\fir$, and let $(\wor_1, \fir_1), (\wor_2, \fir_2), \ldots, (\wor_n, \fir_n)$ be the order of proposals.
		By \Cref{pro:ord} the output of the Propose stage is independent of the proposal order.
		
		Second, I argue by induction that there is a proposal order for the $\rema '$-Propose stage such that the first $n$ proposals are $(\wor_1, \fir_1), (\wor_2, \fir_2), \ldots, (\wor_n, \fir_n)$.
		In the base case, consider $(\wor_1, \fir_1)$.
		There are two cases:
		\begin{enumerate}
			\item $\wor_1 \neq \wor$:
			Then $\wor_1$ or $\rema(\wor_1)$ is a proposal source in $\rema$.
			Thus $\wor_1$ or $\rema(\wor_1)$ is a proposal source in $\rema ' $.
			Therefore $\wor_1$ is active at the start of the $\rema '$-Propose stage.
			\item $\wor_1 = \wor$:
			Then $\wor$ is the only active worker at the start of the $\rema$-Propose stage.
			Because $\matProp ' (\wor) \neq \fir$, this implies that $\wor$ is active at some point in the $\rema '$-Propose stage.
			Therefore, $\wor$ is active at the start of the $\rema '$-Propose stage.
		\end{enumerate}
		Therefore there is a proposal order such that $(\wor_1, \fir_1)$ is the first proposal in the $\rema '$-Propose stage.
		
		For the inductive step, suppose that there is a proposal order such that $(\wor_1, \allowbreak \fir_1), (\wor_2, \fir_2), \ldots, (\wor_{k-1}, \fir_{k-1})$ are the first $k-1$ proposals in the $\rema '$-Propose stage.
		There are two cases:
		\begin{enumerate}
			\item \textit{$\wor_{j} \neq \wor$ for any $j < k$}:
			Observe that there are weakly more rejections in the $\rema '$-Propose stage.
			Therefore, the set of active agents is weakly larger in the $\rema '$-Propose stage, with the possible exception of $\wor$.
			If $\wor_k = \wor$, then $\wor$ is the only active worker in the $\rema$-Propose stage.
			Because $\matProp ' (\wor) \neq \fir$, this implies that $\wor$ is active at some point in the $\rema '$-Propose stage.
			Therefore $\wor$ must be active at the $k^\t{th}$ step of the $\rema '$-Propose stage
			Therefore $\wor_k$ must be active at the $k^\t{th}$ step of the $\rema '$-Propose stage
			\item \textit{$\wor_{j} = \wor$ for some $j < k$}:
			Observe that there are weakly more rejections in the $\rema '$-Propose stage.
			Therefore, the set of active agents is weakly larger in the $\rema '$-Propose stage because $\wor$ has been active at least once.
			Therefore $\wor_k$ must be active at the $k^\t{th}$ step of the $\rema '$-Propose stage.
		\end{enumerate}
		Therefore, $\wor$ makes weakly more proposals in the $\rema '$-Propose stage, a contradiction to the supposition that $\wor$ and $\fir$ profitably misreport the initial match.
		Therefore, $\matProp ' (\wor) = \fir$.

		Suppose (toward a contradiction) that $\wor$ is active in the Propose phase $\rema ' $-Propose stage.
		Let
		\begin{align*}
			\wor_1 \equiv \wor, \fir_1 \equiv \matExch ' (\wor_1), \wor_2 \equiv \rema ' (\fir_1), \ldots, \fir_n \equiv \fir 
		\end{align*}
		be the cycle in which $\wor$ and $\fir$ sit down in in the $\rema ' $-Exchange stage.
		
		Consider any $\wor_k$ in this cycle.
		If $\wor_k$ is active in the $\rema '$-Propose stage, then $\wor_k$ proposes to $\fir_k$ in the $\rema '$-Propose stage because $\matProp ' (\wor_k) = \rema ' (\wor_k)$.
		Because $\matExch ' (\fir_k) = \wor_k$, it follows that $\fir_k $ weakly prefers $\wor_k$ to $\rema ' (\fir_k)$.
		Because $\fir_k$ rejects $\wor_k$ at some point of the $\rema ' $-Propose stage, it then follows that $\wor_{k+1}$ is active in the $\rema '$-Propose stage.
		By supposition, $\wor$ is active in the $\rema '$-Propose stage.
		
		Therefore, $\wor_n$ is active in the $\rema '$-Propose stage.
		Therefore, $\wor_n$ proposes to $\fir$ in the $\rema '$-Propose stage but $\fir$ rejects $\wor_n$.
		Because $\fir$ strictly prefers $\matExch '(\fir) $ to $\matExch(\fir)$, and weakly prefers $\matExch(\fir)$ to being unmatched, it follows that $\fir$ does not reject a proposal from $\wor_n$, a contradiction.
		Therefore, $\wor$ is not active in the $\rema ' $-Propose stage.
	\end{proof}

	
\end{document}